\documentclass[aps,pra,twocolumn,nofootinbib, superscriptaddress,10pt]{revtex4-2}
\usepackage{graphicx} 
\usepackage{subfigure}
\usepackage{amsmath,amsfonts,amssymb,amsthm,bm}
\usepackage{bbm}
\usepackage{url}
\usepackage{lmodern}
\usepackage{mathrsfs}
\usepackage{appendix}
\usepackage[dvipsnames]{xcolor}
\usepackage[colorlinks=true,citecolor=blue,linkcolor=blue]{hyperref}
\usepackage{csquotes}
\MakeOuterQuote{"}

\newtheorem{theorem}{Theorem}

\newtheorem{lemma}{Lemma}

\newtheorem{proposition}{Proposition}

\newtheorem{corollary}{Corollary}
\theoremstyle{plain}

\newcommand{\Tensor}[2]{#1^{\otimes #2}}
\newcommand{\dagtensor}[2]{#1^{\dagger \otimes #2}}

\newcommand{\diag}{\operatorname{diag}}
\newcommand{\rank}{\operatorname{rank}\;}
\newcommand{\supp}{\operatorname{supp}}
\newcommand{\tr}{\operatorname{tr}}

\newcommand{\Cl}{\mathrm{Cl}}
\newcommand{\CSS}{\mathrm{CSS}}

\newcommand{\haar}{\mathrm{Haar}}

\newcommand{\rme}{\operatorname{e}}
\newcommand{\rmi}{\mathrm{i}}

\newcommand{\rmH}{\mathrm{H}}

\newcommand{\rmU}{\mathrm{U}}
\newcommand{\rmW}{\mathrm{W}}

\newcommand{\SWAP}{\mathrm{SWAP}}

\newcommand{\bfa}{\mathbf{a}}
\newcommand{\bfb}{\mathbf{b}}
\newcommand{\bfc}{\mathbf{c}}

\newcommand{\bfp}{\mathbf{p}}
\newcommand{\bfq}{\mathbf{q}}
\newcommand{\bfu}{\mathbf{u}}

\newcommand{\bfw}{\mathbf{w}}
\newcommand{\bfx}{\mathbf{x}}
\newcommand{\bfy}{\mathbf{y}}
\newcommand{\bfz}{\mathbf{z}}

\newcommand{\bmtheta}{\bm{\theta}}

\newcommand{\caB}{\mathcal{B}}
\newcommand{\caD}{\mathcal{D}}
\newcommand{\caE}{\mathcal{E}}
\newcommand{\caH}{\mathcal{H}}
\newcommand{\caL}{\mathcal{L}}
\newcommand{\caM}{\mathcal{M}}
\newcommand{\caN}{\mathcal{N}}
\newcommand{\caO}{\mathcal{O}}
\newcommand{\caP}{\mathcal{P}}
\newcommand{\caR}{\mathcal{R}}
\newcommand{\caS}{\mathcal{S}}
\newcommand{\caT}{\mathcal{T}}
\newcommand{\caU}{\mathcal{U}}
\newcommand{\caV}{\mathcal{V}}
\newcommand{\caW}{\mathcal{W}}

\newcommand{\bcaP}{\overline{\mathcal{P}}}

\newcommand{\bbC}{\mathbb{C}}

\newcommand{\bbE}{\mathbb{E}}
\newcommand{\bbF}{\mathbb{F}}

\newcommand{\bbU}{\mathbb{U}}
\newcommand{\bbZ}{\mathbb{Z}}

\newcommand{\bbone}{\mathbbm{1}}

\newcommand{\scrS}{\mathscr{S}}
\newcommand{\scrT}{\mathscr{T}}
\newcommand{\scrW}{\mathscr{W}}

\newcommand{\tbfu}{\tilde{\bfu}}

\newcommand{\tcaU}{\tilde{\mathcal{U}}}
\newcommand{\tbbU}{\tilde{\mathbb{U}}}

\newcommand{\tM}{\tilde{M}}
\newcommand{\tS}{\tilde{S}}

\newcommand{\tXi}{\tilde{\Xi}}

\newcommand{\frq}{\mathfrak{q}}

\newcommand{\romanNum}[1]{\uppercase\expandafter{\romannumeral#1}}

\newcommand{\lref}[1]{Lemma~\ref{#1}}
\newcommand{\lsref}[1]{Lemmas~\ref{#1}}
\newcommand{\Lref}[1]{Lemma~\ref{#1}}
\newcommand{\Lsref}[1]{Lemmas~\ref{#1}}

\newcommand{\thref}[1]{Theorem~\ref{#1}}
\newcommand{\thsref}[1]{Theorems~\ref{#1}}
\newcommand{\Thref}[1]{Theorem~\ref{#1}}
\newcommand{\Thsref}[1]{Theorems~\ref{#1}}

\newcommand{\pref}[1]{Proposition~\ref{#1}}
\newcommand{\psref}[1]{Propositions~\ref{#1}}
\newcommand{\Pref}[1]{Proposition~\ref{#1}}
\newcommand{\Psref}[1]{Propositions~\ref{#1}}

\newcommand{\coref}[1]{Corollary~\ref{#1}}

\def\eqref#1{\textup{(\ref{#1})}}
\newcommand{\eref}[1]{Eq.~\textup{(\ref{#1})}}
\newcommand{\eqsref}[2]{Eqs.~(\ref{#1}) and (\ref{#2})}

\newcommand{\Eref}[1]{Equation~\textup{(\ref{#1})}}
\newcommand{\Eqsref}[2]{Equations~(\ref{#1}) and (\ref{#2})}

\newcommand{\tref}[1]{Table~\ref{#1}}

\newcommand{\sref}[1]{Sec.~\ref{#1}}

\newcommand{\fref}[1]{Fig.~\ref{#1}}
\newcommand{\Fref}[1]{Figure~\ref{#1}}

\newcommand{\aref}[1]{Appendix~\ref{#1}}
\newcommand{\asref}[1]{Appendices~\ref{#1}}

\def\<{\langle}  
\def\>{\rangle}  

\newcommand{\rcite}[1]{Ref.~\cite{#1}}
\newcommand{\rscite}[1]{Refs.~\cite{#1}}

\begin{document}
\title{Nonstabilizerness Enhances Thrifty Shadow Estimation}
\author{Datong Chen}
\author{Huangjun Zhu}
\email{zhuhuangjun@fudan.edu.cn}
\affiliation{State Key Laboratory of Surface Physics, Department of Physics, and Center for Field Theory and Particle Physics, Fudan University, Shanghai 200433, China}

\affiliation{Institute for Nanoelectronic Devices and Quantum Computing, Fudan University, Shanghai 200433, China}

\affiliation{Shanghai Research Center for Quantum Sciences, Shanghai 201315, China}

\date{\today}

\begin{abstract}
Shadow estimation is a powerful approach for estimating the expectation values of many observables.  Thrifty shadow estimation is a simple variant that is proposed to reduce the experimental overhead by reusing random circuits repeatedly. Although this idea is so simple, its performance is quite elusive. In this work we show that thrifty shadow estimation is effective on average whenever the unitary ensemble forms a 2-design, in sharp contrast with the previous expectation. In thrifty shadow estimation based on the Clifford group, the variance is inversely correlated with the degree of nonstabilizerness of the state and observable, which is a key resource in quantum information processing. For fidelity estimation, it decreases exponentially with the stabilizer 2-R\'{e}nyi entropy of the target state, which endows the stabilizer 2-R\'{e}nyi entropy with a clear operational meaning. In addition,
we propose a simple circuit to enhance 
the efficiency, which requires only 
 one layer of $T$ gates and is particularly appealing in the NISQ era.
\end{abstract}

\maketitle

\emph{Introduction.}---Efficient learning of many-body quantum systems is crucial to many tasks in quantum information processing \cite{Eisert2020Certification,Kliesch2021Tomography,Andreas2023RandomMeasurement}. However, traditional tomographic approaches are too prohibitive to be applied to large and intermediate quantum systems due to the exponential growth in the resource cost \cite{Haffner2005WState,Haah2017Tomography,Kliesch2021Tomography}. Recently, a powerful alternative approach known as shadow estimation \cite{Huang2020Shadow} was proposed to circumvent this limitation and has found a wide spectrum of applications, including fidelity estimation \cite{Huang2020Shadow,Elben2020CrossPlatform}, entanglement detection \cite{Brydges2019Entanglement,Elben2020Entanglement,Zhou2020Entanglement}, Hamiltonian learning \cite{Anshu2020HamiltonianLearning,Hadfield2022HamiltonianLearning,Huang2023HamiltonianLearning}, and machine learning \cite{Huang2021ML, Jerbi2023ML}.  Many applications have been demonstrated in various experimental platforms \cite{Struchalin2021Experiment,Zhang2021Experiment,Stricker2022Experiment,William2022Experiment}.

The basic idea of shadow estimation is quite simple: First apply a random unitary transformation chosen from a given unitary ensemble and then perform the computational basis measurement. After $N$ repetitions, we can estimate the expectation values of many observables based on the measurement outcomes. Popular choices for the unitary ensemble include the Clifford group and the local Clifford group. The Clifford group is particularly attractive 
because the resulting estimation protocol
can be implemented efficiently, and the overhead in classical data processing is small \cite{Aaronson2004StabFormalism}. Moreover, this protocol can achieve constant sample complexity in fidelity estimation, irrespective of the system size. Technically, this merit is tied to the fact that the Clifford group forms a unitary 3-design \cite{Kueng20153design,Webb20163Design,Zhu20173Design}. In addition, the Clifford group plays crucial roles in many other tasks in quantum information processing, including quantum computation \cite{Bravyi2005MagicState, Briegel2009QComputation} and quantum error correction \cite{Gottesman1997StabilizerCode,Campbell2017FaultTolerant}.

To implement the original shadow estimation protocol, it is necessary to adjust the experimental configuration to apply a new random unitary transformation before every measurement. This procedure is 
quite resource-consuming and not  friendly to many platforms. To address this drawback, an alternative approach known as thrifty (or multi-shot) shadow estimation was proposed recently \cite{Helsen2023MultiShot,Zhou2023MultiShot,Seif2023MultiShot,Vermersch2024MultiShot,Struchalin2021Experiment,William2022Experiment}: each random unitary can be applied several times repeatedly.  Although this idea is so simple, the performance of thrifty shadow estimation is quite elusive even for the most common unitary ensembles, such as the Clifford group. According to Ref.~\cite{Helsen2023MultiShot}, thrifty shadow is maximally effective when the unitary ensemble forms a 4-design, but  maximally ineffective for the Clifford group, which is not a 4-design \cite{Zhu2016Fail4Design,Gross2021Duality}. However, the analysis of the Clifford group is based on special examples and thus cannot capture the whole picture.

In this work, we show that thrifty shadow estimation has a much wider scope of applications than expected. Notably, it is effective on average even with a unitary 2-design.
When the Clifford group is chosen as the unitary ensemble, the variance in thrifty shadow estimation is inversely correlated with the degree of nonstabilizerness of the state and observable under consideration, which is a key resource in quantum information processing \cite{Veitch2014StabResource,Bravyi2016TCount,Howard2017RoM,Hinsche2023THard}. For fidelity estimation, the variance decreases exponentially with the stabilizer $2$-R\'{e}nyi entropy introduced in \rcite{Lorenzo2022SRE}, and thrifty shadow estimation is effective for most states of interest. Surprisingly, it is easier to estimate the fidelity of quantum states that are more difficult to prepare. In addition, we show that nonstabilizerness in the unitary ensemble can further reduce the variance in thrifty shadow estimation. By adding one layer of $T$ gates, we propose a simple circuit that can significantly boost the efficiency. Compared with popular interleaved Clifford circuits \cite{Leone2021Interleaved,Haferkamp2022Interleaved,Haferkamp2023Interleaved}, our circuit is much shallower and more appealing to the NISQ era, but  can achieve almost the same performance given the total number of $T$ gates.

\emph{Setup.}---Consider an $n$-qubit quantum system with Hilbert space $\caH$, which has dimension $d=2^n$. Let $\{|\bfb\>\}_{\bfb\in \{0,1\}^n}$ be the computational basis in $\caH$. Denote by $\caL(\caH)$, $\caL^\rmH(\caH)$, $\caL^\rmH_0(\caH)$, and  $\caD(\caH)$ the sets of linear operators, Hermitian operators, traceless Hermitian operators, and density operators on $\caH$, respectively.  
Denote by $\bbone$  the identity operator on  $\caH$. 

%%\begin{equation}
%\caM(\rho)=\sum_{\bfb}\bbE_{U\sim \caU} \<\bfb|U\rho U^\dag |\bfb\> U^\dag|\bfb\>\<\bfb| U.
%\end{equation}

In standard shadow estimation \cite{Huang2020Shadow},  a random unitary transformation $U$ from a given ensemble $\caU$ is chosen and applied to the quantum state $\rho$ of interest. Then a projective measurement on the computational basis is performed, which leads to a record $U^\dag |\bfb\>\<\bfb| U$ associated with the measurement outcome  $\bfb$. The average over the records can be viewed as a post-processing channel $\caM$, which depends on the unitary ensemble $\caU$.
If the ensemble $\{U^\dag|\bfb\>\}_{U\sim\caU, \bfb}$ is informationally complete, then the  channel $\caM$ is invertible, and its inverse $\caM^{-1}$ is termed the reconstruction map. 
The estimator  $\hat{\rho} = \caM^{-1}(U^\dag|\bfb\>\<\bfb| U)$  associated with the unitary $U$ and outcome $\bfb$ is called 
a snapshot of the state $\rho$, from which we can predict the properties of many observables.

In thrifty shadow estimation \cite{Helsen2023MultiShot, Zhou2023MultiShot}, each unitary $U$ sampled  is reused $R$ times, and a  snapshot of the state $\rho$ has the form  $\hat{\rho}_R = \sum_{i=1}^R \caM^{-1}(U^\dag|\bfb_i\>\<\bfb_i|U)/R$, where $\bfb_i$ denotes the outcome of the $i$th measurement. Given  $O\in \caL^\rmH(\caH)$, the variance of  the estimator $\tr(O\hat{\rho}_R)$ reads
\begin{equation}\label{eq:VarMultiShot}
V_R(O,\rho) =\frac{1}{R} V(O,\rho) + \frac{R-1}{R}V_* (O,\rho),
\end{equation}
where $V(O,\rho)$ is the variance in the original estimation protocol, and  $V_*(O,\rho)$ characterizes the overhead introduced by thrifty shadow estimation as a tradeoff for reusing circuits. The total sample cost is proportional to $RV_R(O,\rho)$, which is larger than $V(O,\rho)$. Thrifty shadow estimation is effective when  $RV_*(O,\rho)\lesssim V(O,\rho)$.

It is known that $V(O,\rho)$ and $V_* (O,\rho)$ only depend on the traceless part of $O$ \cite{Huang2020Shadow, Zhou2023MultiShot}, so we assume  $O$ is a traceless observable, that is, $O\in \caL^\rmH_0(\caH)$, in the rest of this paper. Then  $V_* (O,\rho)$ can be expressed as follows,
\begin{equation}\label{eq:DefV*}
V_*(O,\rho) = \tr\bigl\{\Omega(\caU)[\caM^{-1}(O)\otimes\rho]^{\otimes 2}\bigr\}-[\tr(O\rho)]^2,
\end{equation}
where 
\begin{equation}\label{eq:OmegaU}
\Omega(\caU) := \sum_{\bfa,\bfb} \bbE_{U\sim\caU} \dagtensor{U}{4} \bigl[\Tensor{(|\bfa\>\<\bfa|)}{2}\otimes \Tensor{(|\bfb\>\<\bfb|)}{2}\bigr] \Tensor{U}{4}
\end{equation}
is the 4th \emph{cross moment operator}, or cross moment operator for short. 
It determines the performance of thrifty shadow estimation; see \aref{app:Omega} for more details.

Henceforth we denote by $\wp=\wp(\rho) = \tr(\rho^2)$ the purity of $\rho$, by $\mathring{\rho}$ the traceless part of $\rho$, and by $V_*(O)=\max_\rho V_*(O,\rho)$ the maximum variance associated with $O$. When $\rho=|\phi\>\<\phi|$ is a pure state, $V(O,\rho)$ and $V_*(O,\rho)$
can be abbreviated as $V(O,\phi)$ and $V_*(O,\phi)$.

\emph{Average performance of thrifty shadow estimation based on 2-designs.}---Given any observable $O$ on $\caH$, we can define an ensemble as follows,
\begin{equation}\label{eq:DefcaEO}
\caE(O) := \{UOU^\dag \in \caE(O)\;|\;U\in\rmU(d)\}, 
\end{equation}
which inherits the Haar measure on the unitary group $\rmU(d)$. Here the unitary group $\rmU(d)$ can also be replaced by a unitary 2-design. Define the state ensemble $\caE(\rho)$ similarly. Define $V(\caE(O),\rho)$ as the  variance averaged over $\caE(O)$; define $V_*(\caE(O),\rho)$, $V(O,\caE(\rho))$, and $V_*(O,\caE(\rho))$ accordingly. 
\begin{theorem}\label{thm:AverageShadow}
Suppose $\caU$ is a unitary 2-design on $\caH$, $\rho\in \caD(\caH)$, and $O\in\caL^\rmH_0(\caH)$. Then 
\begin{align}
V(\caE(O),\rho) &= V(O,\caE(\rho)) = \left(1+\frac{d-\wp}{d^2-1}\right)\|O\|_2^2, \label{eq:VOrhoAvg}\\
V_*(\caE(O),\rho) &= V_*(O,\caE(\rho)) = \frac{d\wp-1}{d^2-1}\|O\|_2^2<\frac{\|O\|_2^2}{d}.  \label{eq:V*OrhoAvg}
\end{align}
\end{theorem}

\Thref{thm:AverageShadow} shows that on average $V_*(O,\rho)$ is smaller than $V(O,\rho)/d$, so thrifty shadow is very effective in reducing the circuit sample complexity, especially when $d$ is large. This result is in sharp contrast with  \rscite{Helsen2023MultiShot, Zhou2023MultiShot}.

\emph{Thrifty shadow estimation based on  4-designs.}---As a benchmark, the proposition below clarifies the variance in the original estimation protocol, in which 3-design is enough. The upper bound in \eref{eq:VarShadowBound} is known before \cite{Huang2020Shadow}.
\begin{proposition}\label{pro:VarNoiseShadow}
Suppose $\caU$ is a unitary 3-design on $\caH$, $\rho\in \caD(\caH)$,  and $O\in \caL^\rmH_0(\caH)$.  Then
\begin{align}\label{eq:VarShadowBound}
\frac{d+1}{d+2}\|O\|_2^2\leq 	V(O,\rho)\leq 3\|O\|_2^2.
\end{align}	
If  $O=|\phi\>\<\phi|-\bbone/d$ for some pure state $|\phi\>$ in $\caH$, then 	
\begin{equation}\label{eq:VarFShadow}
\frac{d}{d+2}\leq V(O,\rho) =-F^2+\frac{d(2F + 1)}{d+2}\le \frac{2d(d+1)}{(d+2)^2} ,
\end{equation}
where $F = \<\phi|\rho|\phi\>$ is the fidelity between $\rho$ and $|\phi\>$. The lower bound is saturated
when $F=0$, and the upper bound is saturated when $F = d/(d+2)$.
\end{proposition}

Next, we turn to thrifty shadow estimation based on a 4-design, in which case  $V_*(O,\rho)=\caO(\|O\|_2^2/d)$ by \rcite{Helsen2023MultiShot}. Here we provide a more precise upper bound.
\begin{theorem}\label{thm:VarHaar}
Suppose $\caU$ is  a unitary 4-design on $\caH$ and $O\in \caL^\rmH_0(\caH)$. Then
\begin{equation}\label{eq:VarFHaarBound}
V_*(O) \le \frac{4}{d} \|O\|_2^2.
\end{equation}
If  $O=|\phi\>\<\phi|-\bbone/d$ for some pure state $|\phi\>$ in $\caH$, then 
\begin{equation}\label{eq:VarFHaar}
V_*(O)=V_*(O,\phi) = \frac{4(d-1)}{(d+2)(d+3)}<\frac{4}{d}.
\end{equation}
\end{theorem}
Thrifty shadow estimation based on a unitary $4$-design  can significantly reduce the number of circuits required to achieve a given precision. Unfortunately, it is not easy to construct a 4-design with a simple structure.

\emph{Connection with nonstabilizerness.}---Now, we turn to thrifty shadow estimation based on the Clifford group  $\Cl_n$ on $\caH$ (see \aref{app:Clifford}), which is only a 3-design, and highlight the role of nonstabilizerness in reducing the circuit sample complexity.  Nonstabilizerness is a key resource in quantum information processing \cite{Veitch2014StabResource,Bravyi2016TCount,Howard2017RoM,Hinsche2023THard}, and various measures have been proposed, such as the stabilizer rank~\cite{Bravyi2016TCount,Bravyi2016TCount2},   robustness of magic~\cite{Howard2017RoM}, and stabilizer $\alpha$-R\'{e}nyi entropy ($\alpha$-SRE)~\cite{Lorenzo2022SRE}.  The $\alpha$-SRE is particularly appealing because of its simplicity, but lacks 
a clear operational meaning so far. Here we shall establish a precise connection between the 2-SRE and the performance of thrifty shadow estimation, thereby endowing the 2-SRE with a clear operational interpretation.

Let $\bcaP_n=\{I,X,Y,Z\}^{\otimes n}$ be the set of  $n$-qubit Pauli operators. The characteristic function $\Xi_O$ of an observable $O$ on $\caH$ is a function on $\bcaP_n$ with $\Xi_O(P)=\tr(OP)$ for $P\in \bcaP_n$ \cite{Gross2006CharFunction,Dai2022CharFunction}, and
can be regarded as a vector composed of $d^2$ entries.  The characteristic function $\Xi_\rho$ of a state $\rho$ is defined in the same way. As a generalization, here we define the cross characteristic function $\Xi_{\rho,O}$ and twisted cross characteristic function $\tXi_{\rho,O}$ as follows,
\begin{align}
\Xi_{\rho,O}(P) = \tr(\rho P)\tr(OP), \;\; \tXi_{\rho,O}(P)=\tr(\rho PO P). 
\end{align}
 Then, the 2-SRE of a state $\rho$ is defined as
\begin{equation}\label{eq:DefM2}
\tM_2(\rho) := -\log_2 \frac{\|\Xi_\rho\|_4^4}{\|\Xi_\rho\|_2^2} = -\log_2 \frac{\|\Xi_\rho\|_4^4}{d\wp(\rho)},
\end{equation}
where $\wp(\rho)$ is the purity of $\rho$. When $\rho=|\phi\>\<\phi|$ is a pure state, the argument $\rho$ can be abbreviated as $\phi$, and $\tM_2(\rho)$ can be written as $M_2(\phi)$.

\emph{Nonstabilizerness from states and observables.}---Define
\begin{align}\label{eq:VTriangleDef}
V_\triangle(O,\rho):=\frac{d+1}{d(d+2)}\left( \|\Xi_{\rho, O}\|_2^2 + \tXi_{\rho, O}\cdot\Xi_{\rho, O}\right),
\end{align}
where $\tXi_{\rho, O}\cdot\Xi_{\rho, O}$ denotes the inner product between the two cross characteristic functions regarded as vectors. 
\begin{theorem}\label{thm:VarCl}
Suppose $\caU=\Cl_n$, $\rho\in \caD(\caH)$, and $O\in \caL^\rmH_0(\caH)$. Then
\begin{align}\label{eq:VarClUpper}
V_*(O,\rho) &= V_\triangle(O,\rho)-\frac{[\tr(O\rho)]^2}{d+2}\leq V_\triangle(O,\rho)\nonumber\\
&\leq \frac{2(d+1)}{d(d+2)} \|\Xi_{\rho, O}\|_2^2 \leq  \frac{2(d+1)}{d+2} \|O\|_2^2.
\end{align}
\end{theorem}
The upper bound  $V_\triangle(O,\rho)$ for  $V_*(O,\rho)$ is almost tight, while the second bound is usually around two times of $V_*(O,\rho)$, as illustrated in \fref{fig:CompareUpper} in \aref{app:BoundsAdd}. \Thref{thm:VarCl} shows that $\rho$ and $O$ play symmetric roles in $V_*(O,\rho)$. Given the purity $\wp(\rho)$  and Hilbert-Schmidt norm $\|O\|_2$, the variance  $V_*(O,\rho)$ is mainly determined by the degree of concentration of the cross characteristic function $\Xi_{\rho,O}$: the higher degree of concentration, the larger the variance. This variance is large [comparable to $V(O,\rho)$] only when both  $\Xi_\rho$ and $\Xi_O$ are highly concentrated and aligned in a certain sense. In a generic situation, these conditions cannot be satisfied, and  $V_*(O,\rho)$ is much smaller than $V(O,\rho)$, as illustrated in \fref{fig:RandomState}. The following corollary further corroborates this intuition. 
\begin{corollary}\label{cor:VarCl}
Suppose $\caU=\Cl_n$, $\rho\in \caD(\caH)$, and $O\in \caL^\rmH_0(\caH)$. Then
\begin{align}
V_*(O,\rho) &\le \frac{2(d+1)}{d(d+2)} \sqrt{2^{-\tM_2(\rho)}d\wp\|\Xi_O\|_4^4},\\
V_*(O,\rho)&\leq \frac{2(d+1)}{d+2}\|\Xi_{\mathring{\rho}}\|_\infty^2\|O\|_2^2,\\
V_*(O) &\le \frac{2(d+1)}{d(d+2)}
\|\Xi_O^2\|_{[d]},
\end{align}
where  $\|\Xi_O^2\|_{[d]}$ is the sum of the $d$ largest entries of the squared characteristic function $\Xi_O^2$. 
\end{corollary}

\begin{figure}[t]
	\centering
	\includegraphics[width=0.4\textwidth]{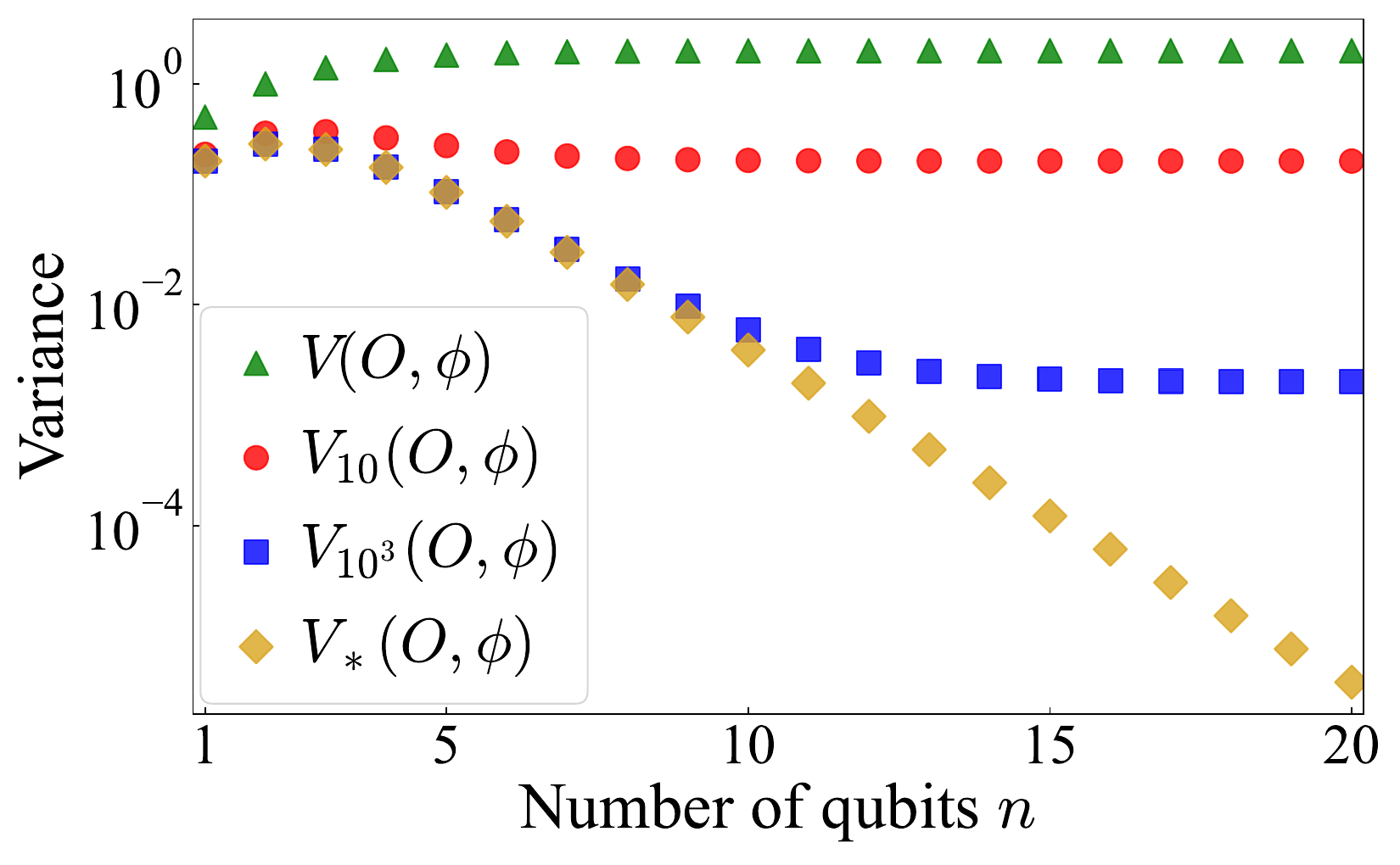}
	\caption{Mean variances in fidelity estimation based on thrifty shadow and Clifford measurements. Here  $O=|\phi\>\<\phi|-\bbone/d$, and each data point is the average over the Haar random pure state $|\phi\>$. The variance $V(O,\phi)$ is determined by \pref{pro:VarNoiseShadow} with $F=1$; the mean value of 
$V_*(O,\phi)$ is determined by \pref{pro:AverageVF} in \aref{app:BoundsAdd}, which coincides with the variance in \eref{eq:VarFHaar}; $V_{10}(O,\phi)$ and $V_{10^3}(O,\phi)$ are determined  by \eref{eq:VarMultiShot} with $R=10$ and $R=1000$, respectively.}
	\label{fig:RandomState}
\end{figure}

It is instructive to consider the estimation of a nontrivial Pauli operator, that is, $O=P$ for $P\in \bcaP_n$ and $P\neq \bbone$. 
Now, $\Xi_O$ is concentrated at a single point labeled by $P$,
\begin{equation}
V(P,\rho) = d+1-\Xi_\rho(P), \quad V_*(P,\rho) = d\Xi_\rho^2(P).
\end{equation} 
Both $V(P,\rho)$ and $V_*(P,\rho)$ are determined by $\Xi_\rho(P)$. When  $\Xi_\rho^2(P)=\caO(1)$, $V_*(P,\rho)$ is comparable to $V(P,\rho)$, and thrifty shadow estimation is not so effective. Nevertheless, the average of $V_*(P,\rho)$ over all nontrivial Pauli operators is smaller than $\wp(\rho)$, which is much smaller than the average of $V(P,\rho)$, so thrifty shadow estimation is effective to estimating most Pauli operators.

For a pure state $\rho$,
the degree of concentration of $\Xi_\rho$ is negatively correlated with the degree of nonstabilizerness. Notably, $\Xi_\rho$ has at least $d$ nonzero entries, and the lower bound is saturated if and only if $\rho$ is a stabilizer state. So it is not surprising that $V_*(O,\rho)$ is tied to nonstabilizerness. This connection is most prominent in fidelity estimation in which $O=|\phi\>\<\phi|-\bbone/d$ is determined by a  pure state $|\phi\>\in \caH$. \Thref{thm:VarFCl} below establishes a precise connection between $V_*(O,\rho)$ and  2-SRE $M_2(\phi)$ and offers a simple operational interpretation of $M_2(\phi)$. 

\begin{theorem}\label{thm:VarFCl}
Suppose $\caU=\Cl_n$ and  $O=|\phi\>\<\phi|- \bbone/d$ with $|\phi\>\in \caH$. Then 
\begin{align}
V_*(O) &<  \frac{2^{1-M_2(\phi)/2}(d+1)}{d+2},\\ 
V_*(O,\phi) &= \frac{2^{1-M_2(\phi)}(d+1)-4}{d+2}<2^{1-M_2(\phi)}\label{eq:VarFCl}.
\end{align}
\end{theorem}
Whenever $M_2(\phi)\gg1$, we have $V_*(O,\phi)\ll V(O,\phi)$, and reusing circuits is effective in reducing the  variance. Notably, generic $n$-qubit pure states have nearly maximal 2-SRE~\cite{Lorenzo2022SRE,Gu2024Pseudomagic}, and the 2-SREs of various many-body quantum states are $\caO(n)$, such as ground states of  quantum Ising chains \cite{Oliviero2022IsingSRE,Haug2023IsingSRE,Tarabunga2023IsingSRE} and most hypergraph states~\cite{Liu2022HypergraphSRE,Chen2024HypergraphSRE}.   Usually, it is more resource-consuming to prepare quantum states with larger 2-SREs. Fortunately,  such states are easier to verify 
with respect to the circuit sample complexity, in sharp contrast with the conclusion in \rcite{Leone2023NonstabilizerHard}. This result may have profound implications for quantum information processing.

\begin{figure}[t]
    \centering
    \includegraphics[width=0.45\textwidth]{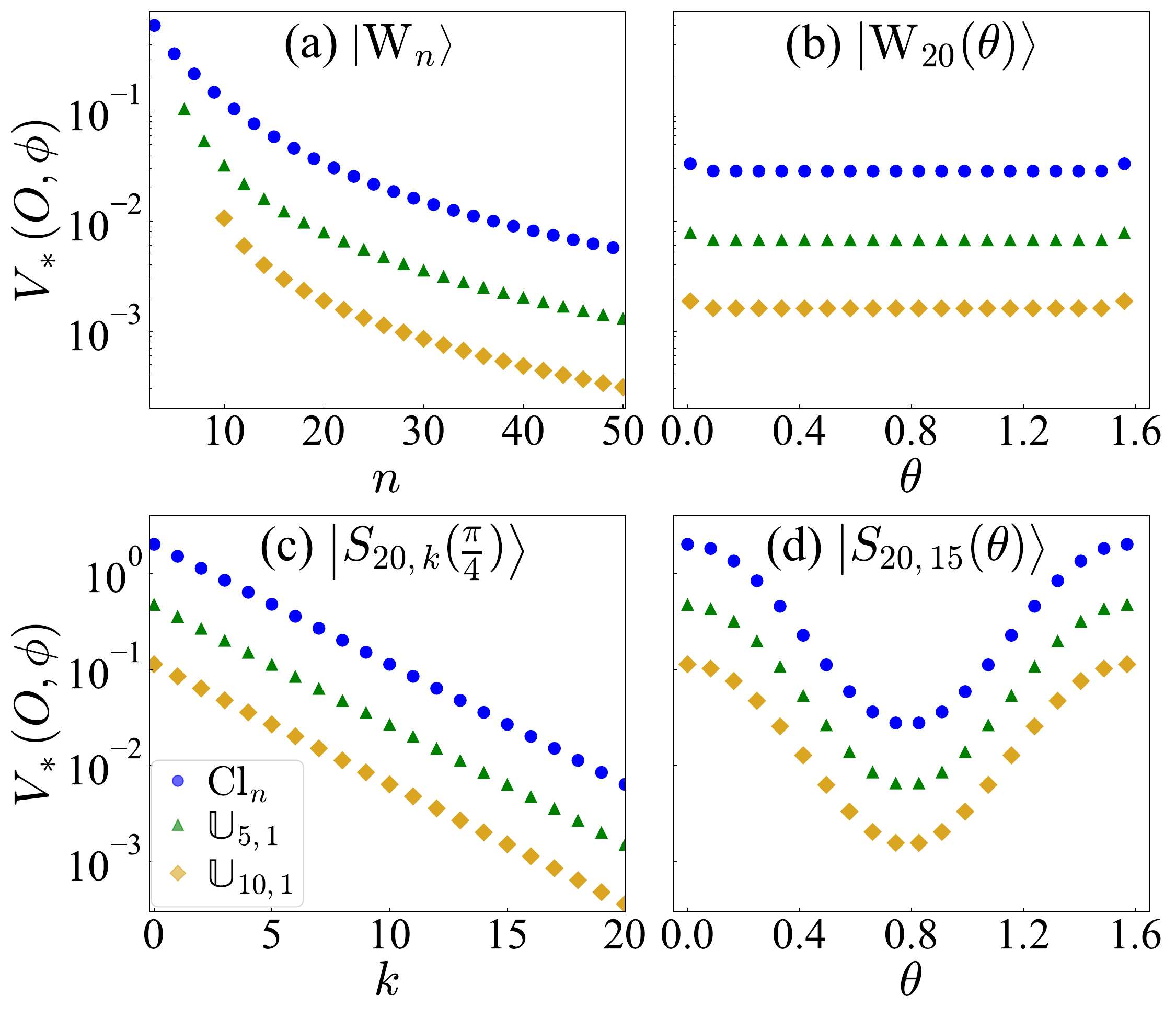}
    \caption{The variance $V_*(O,\phi)$ in fidelity estimation based on thrifty shadow and the three unitary ensembles $\Cl_n$, $\bbU_{5, 1}$, and $\bbU_{10, 1}$. 
 Here $O=|\phi\>\<\phi|-\bbone/d$, $|\phi\>$ is shown in each plot, and  $V_*(O,\phi)$ is determined by \thref{thm:VarFCl} and \pref{pro:VarUklFphi}  in \aref{app:BoundsAdd}, where $M_2(\phi)$ is determined in \pref{pro:CharNormPhasedW} and \eref{eq:SRESnk}.}
    \label{fig:VarTypes}
\end{figure}

When the actual state coincides with the target state, that is, $\rho=|\phi\>\<\phi|$, the variance $V_*(O,\rho)$ decreases monotonically with $M_2(\phi)$ by \eref{eq:VarFCl}. When $M_2(\phi)=0$, which holds for stabilizer states, $V_*(O,\rho)$ attains its maximum and  equals 
 $V(O,\rho)$. In this case, thrifty shadow estimation is not effective, but stabilizer states can be verified efficiently using  Pauli measurements \cite{Hayashi2015GraphVerify,Pallister2018Verification,Zhu2019Verification,Zhu2020StabVerify}.
When  $M_2(\phi) = \log_2 [(d+3)/4]$, the Clifford orbit of $|\phi\>$ forms a 4-design~\cite{Zhu2016Fail4Design}, and $V_*(O,\rho)$ is equal to the counterpart based on a 4-design ensemble as shown in \eref{eq:VarFHaar}. When $M_2(\phi) > \log_2 [(d+3)/4]$, $V_*(O,\rho)$ gets even smaller, and the Clifford group is even better than 4-design ensembles.

To  manifest the role of 2-SRE in fidelity estimation based on thrifty shadow estimation, as examples of the target state $|\phi\>$ we consider  the following states,
\begin{align}
|\rmW_n(\theta)\> &= \frac{1}{\sqrt{n}}\sum_{j=1}^n e^{\rmi j\theta}X_j \Tensor{|0\>}{n}, \label{eq:WState}\\
|S_{n, k}(\theta)\> &= \Tensor{|0\>}{n-k}\otimes \Tensor{\left[\frac{1}{\sqrt{2}}\left(|0\>+e^{\rmi\theta}|1\>\right)\right]}{k}. \label{eq:MagicState}
\end{align}
Here $|\rmW_n(\theta)\>$ is a generalized W state \cite{Haffner2005WState,Jovan2023WSRE,Catalano2024GeneralW, Yi2023GeneralW} and $X_j$ denotes the Pauli operator with $X$ acting on the $j$th qubit; when $\theta=0$, $|\rmW_n(\theta)\>$ reduces to the  W state, denoted by  $|\rmW_n\>$. 
The state  $|S_{n, k}(\theta)\>$ was  considered in \rcite{Bravyi2005MagicState}. The  2-SREs of these states can be found in \aref{app:Stab2Entropy}.  \Fref{fig:VarTypes} illustrates the dependence of the variance $V_R(O,\phi)$ with $O=|\phi\>\<\phi|-\bbone/d$ on various relevant parameters, demonstrating the effectiveness of thrifty shadow in fidelity estimation.

\begin{figure}[b]
	\centering
	\includegraphics[width=0.4\textwidth]{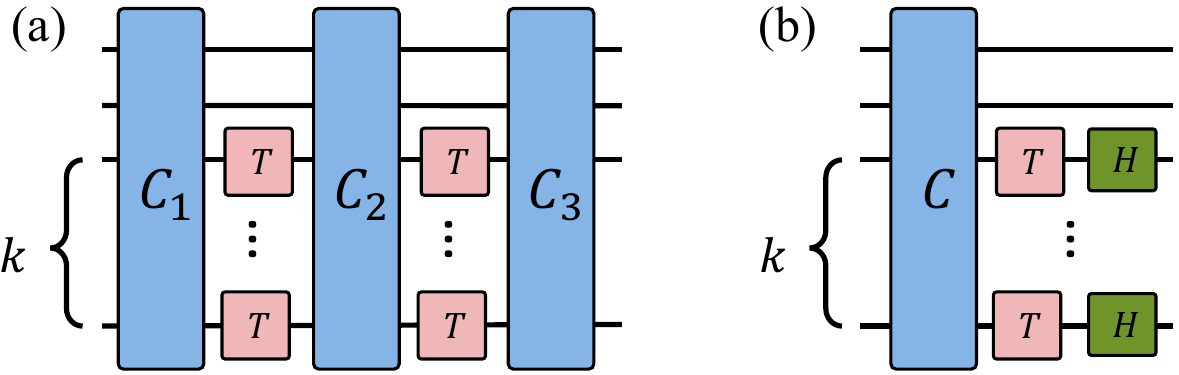}
	\caption{(a) Clifford circuit with  $l=2$ interleaved layers of $T$ gates, which corresponds to the ensemble $\bbU_{k, l}$. (b) A simple circuit underlying the ensemble  $\tbbU_k$, which is equally effective as the interleaved Clifford circuit in thrifty shadow estimation. 
	}
	\label{fig:ModelCircuits}
\end{figure}

\emph{Nonstabilizerness from unitary ensembles.}---Next, we clarify the impact of nonstabilizerness in the unitary ensemble on thrifty shadow estimation and propose a simple recipe to boost efficiency. 
First, we consider Clifford circuits interleaved by $T:=\diag(1,\rme^{\pi\rmi/4})$ gates (also called $\pi/8$ gates) as shown in \fref{fig:ModelCircuits}(a) \cite{Haferkamp2023Interleaved,Haferkamp2022Interleaved,Leone2021Interleaved,Helsen2023MultiShot}. Denote by  $\bbU_{k, l}$ the unitary ensemble generated by Clifford circuits with  $l$ interleaved layers composed of $k$ $T$ gates each. Most previous works focus on the special case with $k=1$. Let $|+\>=(|0\>+|1\>)/\sqrt{2}$ and $\gamma\equiv 2^{-M_2(T|+\>)}=3/4$.

\begin{theorem}\label{thm:VarUkl}
Suppose $\caU=\bbU_{k, l}$ with $0\leq k\leq n$ and $l\geq 0$, $\rho\in \caD(\caH)$, and $O\in \caL^\rmH_0(\caH)$.
Then
\begin{gather}
\left|V_*(O, \rho)-\gamma^{kl}V_\triangle(O,\rho)\right|\le \frac{6}{d}\|O\|_2^2.\label{eq:VarUklUpper}
\end{gather}
\end{theorem}
\begin{theorem}\label{thm:VarFUkl}
	Suppose $\caU=\bbU_{k,l}$ with $0\leq k\leq n$ and $l\geq 0$, and $O=|\phi\>\<\phi|- \bbone/d$ with $|\phi\>\in \caH$.  Then 
	\begin{gather}
	V_*(O) \le 2^{1-M_2(\phi)/2}\gamma^{kl}  + \frac{6}{d},\label{eq:VarFUklBound}\\
	-\frac{6}{d}<V_*(O,\phi) -2^{1-M_2(\phi)}\gamma^{kl} <  \frac{4}{d}.\label{eq:VarFUklIdeal}
	\end{gather}
\end{theorem}

\Thsref{thm:VarUkl} and \ref{thm:VarFUkl}  highlight the roles of three types of  nonstabilizerness associated with the state $\rho$, observable $O$, and unitary ensemble  $\bbU_{k, l}$. Compared with the Clifford ensemble, $\bbU_{k, l}$ can reduce the variance $V_*(O,\rho)$ by a factor of $\gamma^{kl}$, which is  determined by the total number of $T$ gates. Notably, the two ensembles  $\bbU_{k, 1}$ and $\bbU_{1, k}$ have  similar performances, as illustrated in \fref{fig:VarInterleaved}. This fact is instrumental to practical applications: $\bbU_{k, 1}$ can be realized with much fewer elementary gates and much shallower circuits and is thus more appealing.

\begin{figure}[t]
	\centering
	\includegraphics[width = 0.4\textwidth]{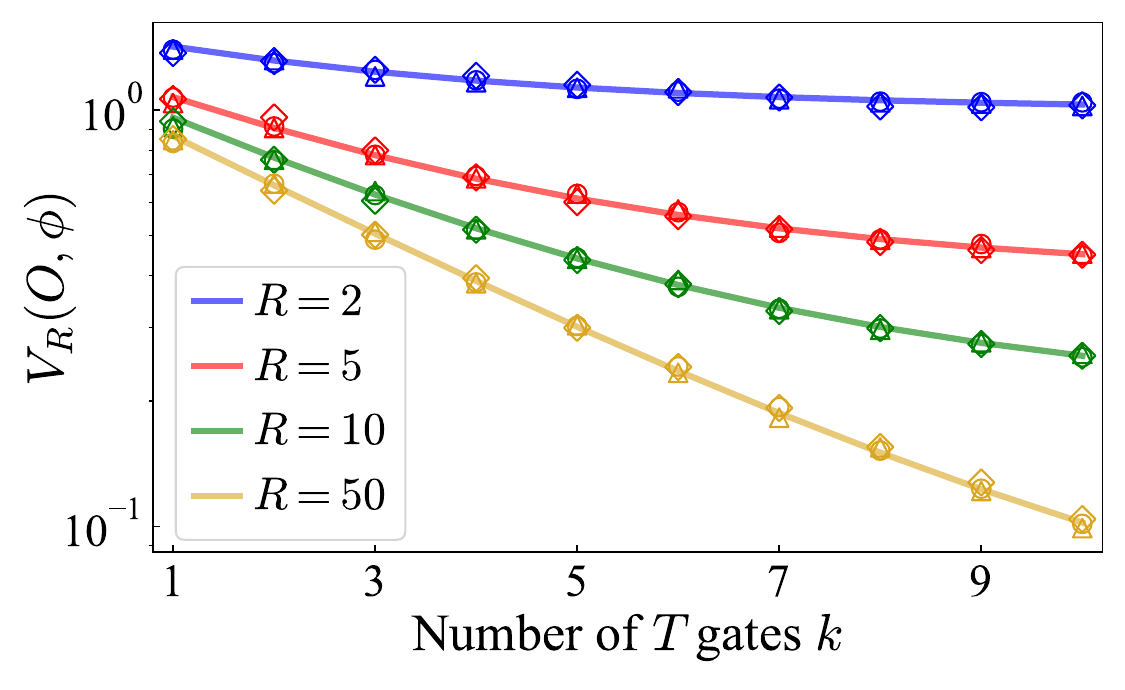}
	\caption{The variance in fidelity estimation based on thrifty shadow, where $O=|\phi\>\<\phi|-\bbone/d$ with  $|\phi\>=|S_{20,2}(\pi/4)\>$. The circles, 
triangles, and diamonds correspond to the unitary  ensembles 	$\bbU_{k, 1}$, $\bbU_{1, k}$, and $\tbbU_k$, respectively. For each ensemble, $50000$ unitaries are sampled and each one is reused $R$ times.	
}\label{fig:VarInterleaved}
\end{figure}

Surprisingly, the circuit underlying the ensemble $\bbU_{k, 1}$ can be simplified further. As shown in \fref{fig:ModelCircuits}(b), the second Clifford layer can be replaced by a fixed Clifford unitary $\Tensor{I}{(n-k)}\otimes\Tensor{H}{k}$, where $H$ is the Hadamard gate. The resulting ensemble reads $\tbbU_k=\Tensor{I}{(n-k)}\otimes\Tensor{(HT)}{k}\Cl_n$. Compared with $\bbU_{k,1}$, the circuit depth can be reduced by one half.  The following theorem shows that
the ensemble $\tbbU_k$ can achieve almost the same performance  as  $\bbU_{k, 1}$ and $\bbU_{1, k}$, as illustrated in \fref{fig:VarInterleaved}. See  \aref{app:BoundsAdd} for the counterpart of \thref{thm:VarFUkl} and additional results.
\begin{theorem}\label{thm:VartUk}
Suppose $\caU=\tbbU_k$ with $0\leq k\leq n$, $\rho\in \caD(\caH)$, and $O\in \caL^\rmH_0(\caH)$. Then
\begin{gather}
\left| V_*(O, \rho) - \gamma^{k}V_\triangle(O,\rho) \right| \le \frac{6}{d}\|O\|_2^2.\label{eq:VartUkBound}
\end{gather}
\end{theorem}

\emph{Summary.}---We showed that
thrifty shadow estimation is effective on average whenever the underlying unitary ensemble forms a unitary 2-design, which exceeds expectations. In estimation protocol based on the Clifford group, 
the variance is tied to the degree of nonstabilizerness and decreases exponentially with the stabilizer $2$-R\'{e}nyi entropy of the target state for the task of fidelity estimation. Furthermore,  nonstabilizerness in the unitary ensemble can help reduce the variance in addition to the contribution from states and observables. On this basis, we proposed a simple circuit to enhance the efficiency, which requires only one layer of $T$ gates,  but can achieve almost the same performance as popular interleaved Clifford circuits.
Our work endows the stabilizer $2$-R\'{e}nyi entropy with a clear operational interpretation and reveals an intriguing connection between shadow estimation and the resource theory of magic state quantum computation, which may have profound implications for various related topics. In the course of our study, we clarified the properties of the cross moment operator, which is also useful to studying common randomized measurements~\cite{Vermersch2024MultiShot} and cross-platform verification~\cite{Elben2020CrossPlatform,Anshu2022CrossPlatform}.

%\bigskip

\let\oldaddcontentsline\addcontentsline
\renewcommand{\addcontentsline}[3]{}

\acknowledgments
We thank Chengsi Mao for valuable discussions. This work is supported by the National Natural Science Foundation of China (Grant No.~92165109),  National Key Research and Development Program of China (Grant No. 2022YFA1404204), and Shanghai Municipal Science and Technology Major Project (Grant No.~2019SHZDZX01).

\bibliography{reference}

\begin{thebibliography}{61}%
\makeatletter
\providecommand \@ifxundefined [1]{%
 \@ifx{#1\undefined}
}%
\providecommand \@ifnum [1]{%
 \ifnum #1\expandafter \@firstoftwo
 \else \expandafter \@secondoftwo
 \fi
}%
\providecommand \@ifx [1]{%
 \ifx #1\expandafter \@firstoftwo
 \else \expandafter \@secondoftwo
 \fi
}%
\providecommand \natexlab [1]{#1}%
\providecommand \enquote  [1]{``#1''}%
\providecommand \bibnamefont  [1]{#1}%
\providecommand \bibfnamefont [1]{#1}%
\providecommand \citenamefont [1]{#1}%
\providecommand \href@noop [0]{\@secondoftwo}%
\providecommand \href [0]{\begingroup \@sanitize@url \@href}%
\providecommand \@href[1]{\@@startlink{#1}\@@href}%
\providecommand \@@href[1]{\endgroup#1\@@endlink}%
\providecommand \@sanitize@url [0]{\catcode `\\12\catcode `\$12\catcode
  `\&12\catcode `\#12\catcode `\^12\catcode `\_12\catcode `\%12\relax}%
\providecommand \@@startlink[1]{}%
\providecommand \@@endlink[0]{}%
\providecommand \url  [0]{\begingroup\@sanitize@url \@url }%
\providecommand \@url [1]{\endgroup\@href {#1}{\urlprefix }}%
\providecommand \urlprefix  [0]{URL }%
\providecommand \Eprint [0]{\href }%
\providecommand \doibase [0]{https://doi.org/}%
\providecommand \selectlanguage [0]{\@gobble}%
\providecommand \bibinfo  [0]{\@secondoftwo}%
\providecommand \bibfield  [0]{\@secondoftwo}%
\providecommand \translation [1]{[#1]}%
\providecommand \BibitemOpen [0]{}%
\providecommand \bibitemStop [0]{}%
\providecommand \bibitemNoStop [0]{.\EOS\space}%
\providecommand \EOS [0]{\spacefactor3000\relax}%
\providecommand \BibitemShut  [1]{\csname bibitem#1\endcsname}%
\let\auto@bib@innerbib\@empty
%</preamble>
\bibitem [{\citenamefont {{Eisert}}\ \emph {et~al.}(2020)\citenamefont
  {{Eisert}}, \citenamefont {{Hangleiter}}, \citenamefont {{Walk}},
  \citenamefont {{Roth}}, \citenamefont {{Markham}}, \citenamefont {{Parekh}},
  \citenamefont {{Chabaud}},\ and\ \citenamefont
  {{Kashefi}}}]{Eisert2020Certification}%
  \BibitemOpen
  \bibfield  {author} {\bibinfo {author} {\bibfnamefont {J.}~\bibnamefont
  {{Eisert}}}, \bibinfo {author} {\bibfnamefont {D.}~\bibnamefont
  {{Hangleiter}}}, \bibinfo {author} {\bibfnamefont {N.}~\bibnamefont
  {{Walk}}}, \bibinfo {author} {\bibfnamefont {I.}~\bibnamefont {{Roth}}},
  \bibinfo {author} {\bibfnamefont {D.}~\bibnamefont {{Markham}}}, \bibinfo
  {author} {\bibfnamefont {R.}~\bibnamefont {{Parekh}}}, \bibinfo {author}
  {\bibfnamefont {U.}~\bibnamefont {{Chabaud}}},\ and\ \bibinfo {author}
  {\bibfnamefont {E.}~\bibnamefont {{Kashefi}}},\ }\bibfield  {title} {\bibinfo
  {title} {{Quantum certification and benchmarking}},\ }\href
  {https://doi.org/10.1038/s42254-020-0186-4} {\bibfield  {journal} {\bibinfo
  {journal} {Nat. Rev. Phys.}\ }\textbf {\bibinfo {volume} {2}},\ \bibinfo
  {pages} {382} (\bibinfo {year} {2020})}\BibitemShut {NoStop}%
\bibitem [{\citenamefont {Kliesch}\ and\ \citenamefont
  {Roth}(2021)}]{Kliesch2021Tomography}%
  \BibitemOpen
  \bibfield  {author} {\bibinfo {author} {\bibfnamefont {M.}~\bibnamefont
  {Kliesch}}\ and\ \bibinfo {author} {\bibfnamefont {I.}~\bibnamefont {Roth}},\
  }\bibfield  {title} {\bibinfo {title} {{Theory of Quantum System
  Certification}},\ }\href {https://doi.org/10.1103/PRXQuantum.2.010201}
  {\bibfield  {journal} {\bibinfo  {journal} {PRX Quantum}\ }\textbf {\bibinfo
  {volume} {2}},\ \bibinfo {pages} {010201} (\bibinfo {year}
  {2021})}\BibitemShut {NoStop}%
\bibitem [{\citenamefont {Elben}\ \emph {et~al.}(2023)\citenamefont {Elben},
  \citenamefont {Flammia}, \citenamefont {Huang}, \citenamefont {Kueng},
  \citenamefont {Preskill}, \citenamefont {Vermersch},\ and\ \citenamefont
  {Zoller}}]{Andreas2023RandomMeasurement}%
  \BibitemOpen
  \bibfield  {author} {\bibinfo {author} {\bibfnamefont {A.}~\bibnamefont
  {Elben}}, \bibinfo {author} {\bibfnamefont {S.~T.}\ \bibnamefont {Flammia}},
  \bibinfo {author} {\bibfnamefont {H.-Y.}\ \bibnamefont {Huang}}, \bibinfo
  {author} {\bibfnamefont {R.}~\bibnamefont {Kueng}}, \bibinfo {author}
  {\bibfnamefont {J.}~\bibnamefont {Preskill}}, \bibinfo {author}
  {\bibfnamefont {B.}~\bibnamefont {Vermersch}},\ and\ \bibinfo {author}
  {\bibfnamefont {P.}~\bibnamefont {Zoller}},\ }\bibfield  {title} {\bibinfo
  {title} {{The randomized measurement toolbox}},\ }\href
  {https://doi.org/10.1038/s42254-022-00535-2} {\bibfield  {journal} {\bibinfo
  {journal} {Nat. Rev. Phys.}\ }\textbf {\bibinfo {volume} {5}},\ \bibinfo
  {pages} {9} (\bibinfo {year} {2023})}\BibitemShut {NoStop}%
\bibitem [{\citenamefont {{H{\"a}ffner}}\ \emph {et~al.}(2005)\citenamefont
  {{H{\"a}ffner}}, \citenamefont {{H{\"a}nsel}}, \citenamefont {{Roos}},
  \citenamefont {{Benhelm}}, \citenamefont {{Chek-Al-Kar}}, \citenamefont
  {{Chwalla}}, \citenamefont {{K{\"o}rber}}, \citenamefont {{Rapol}},
  \citenamefont {{Riebe}}, \citenamefont {{Schmidt}}, \citenamefont {{Becher}},
  \citenamefont {{G{\"u}hne}}, \citenamefont {{D{\"u}r}},\ and\ \citenamefont
  {{Blatt}}}]{Haffner2005WState}%
  \BibitemOpen
  \bibfield  {author} {\bibinfo {author} {\bibfnamefont {H.}~\bibnamefont
  {{H{\"a}ffner}}}, \bibinfo {author} {\bibfnamefont {W.}~\bibnamefont
  {{H{\"a}nsel}}}, \bibinfo {author} {\bibfnamefont {C.~F.}\ \bibnamefont
  {{Roos}}}, \bibinfo {author} {\bibfnamefont {J.}~\bibnamefont {{Benhelm}}},
  \bibinfo {author} {\bibfnamefont {D.}~\bibnamefont {{Chek-Al-Kar}}}, \bibinfo
  {author} {\bibfnamefont {M.}~\bibnamefont {{Chwalla}}}, \bibinfo {author}
  {\bibfnamefont {T.}~\bibnamefont {{K{\"o}rber}}}, \bibinfo {author}
  {\bibfnamefont {U.~D.}\ \bibnamefont {{Rapol}}}, \bibinfo {author}
  {\bibfnamefont {M.}~\bibnamefont {{Riebe}}}, \bibinfo {author} {\bibfnamefont
  {P.~O.}\ \bibnamefont {{Schmidt}}}, \bibinfo {author} {\bibfnamefont
  {C.}~\bibnamefont {{Becher}}}, \bibinfo {author} {\bibfnamefont
  {O.}~\bibnamefont {{G{\"u}hne}}}, \bibinfo {author} {\bibfnamefont
  {W.}~\bibnamefont {{D{\"u}r}}},\ and\ \bibinfo {author} {\bibfnamefont
  {R.}~\bibnamefont {{Blatt}}},\ }\bibfield  {title} {\bibinfo {title}
  {{Scalable multiparticle entanglement of trapped ions}},\ }\href
  {https://doi.org/10.1038/nature04279} {\bibfield  {journal} {\bibinfo
  {journal} {Nature}\ }\textbf {\bibinfo {volume} {438}},\ \bibinfo {pages}
  {643} (\bibinfo {year} {2005})}\BibitemShut {NoStop}%
\bibitem [{\citenamefont {Haah}\ \emph {et~al.}(2017)\citenamefont {Haah},
  \citenamefont {Harrow}, \citenamefont {Ji}, \citenamefont {Wu},\ and\
  \citenamefont {Yu}}]{Haah2017Tomography}%
  \BibitemOpen
  \bibfield  {author} {\bibinfo {author} {\bibfnamefont {J.}~\bibnamefont
  {Haah}}, \bibinfo {author} {\bibfnamefont {A.~W.}\ \bibnamefont {Harrow}},
  \bibinfo {author} {\bibfnamefont {Z.}~\bibnamefont {Ji}}, \bibinfo {author}
  {\bibfnamefont {X.}~\bibnamefont {Wu}},\ and\ \bibinfo {author}
  {\bibfnamefont {N.}~\bibnamefont {Yu}},\ }\bibfield  {title} {\bibinfo
  {title} {{Sample-Optimal Tomography of Quantum States}},\ }\href
  {https://doi.org/10.1109/TIT.2017.2719044} {\bibfield  {journal} {\bibinfo
  {journal} {IEEE Trans. Inf. Theory}\ }\textbf {\bibinfo {volume} {63}},\
  \bibinfo {pages} {5628} (\bibinfo {year} {2017})}\BibitemShut {NoStop}%
\bibitem [{\citenamefont {{Huang}}\ \emph {et~al.}(2020)\citenamefont
  {{Huang}}, \citenamefont {{Kueng}},\ and\ \citenamefont
  {{Preskill}}}]{Huang2020Shadow}%
  \BibitemOpen
  \bibfield  {author} {\bibinfo {author} {\bibfnamefont {H.-Y.}\ \bibnamefont
  {{Huang}}}, \bibinfo {author} {\bibfnamefont {R.}~\bibnamefont {{Kueng}}},\
  and\ \bibinfo {author} {\bibfnamefont {J.}~\bibnamefont {{Preskill}}},\
  }\bibfield  {title} {\bibinfo {title} {{Predicting many properties of a
  quantum system from very few measurements}},\ }\href
  {https://doi.org/10.1038/s41567-020-0932-7} {\bibfield  {journal} {\bibinfo
  {journal} {Nat. Phys.}\ }\textbf {\bibinfo {volume} {16}},\ \bibinfo {pages}
  {1050} (\bibinfo {year} {2020})}\BibitemShut {NoStop}%
\bibitem [{\citenamefont {Elben}\ \emph
  {et~al.}(2020{\natexlab{a}})\citenamefont {Elben}, \citenamefont {Vermersch},
  \citenamefont {van Bijnen}, \citenamefont {Kokail}, \citenamefont {Brydges},
  \citenamefont {Maier}, \citenamefont {Joshi}, \citenamefont {Blatt},
  \citenamefont {Roos},\ and\ \citenamefont {Zoller}}]{Elben2020CrossPlatform}%
  \BibitemOpen
  \bibfield  {author} {\bibinfo {author} {\bibfnamefont {A.}~\bibnamefont
  {Elben}}, \bibinfo {author} {\bibfnamefont {B.}~\bibnamefont {Vermersch}},
  \bibinfo {author} {\bibfnamefont {R.}~\bibnamefont {van Bijnen}}, \bibinfo
  {author} {\bibfnamefont {C.}~\bibnamefont {Kokail}}, \bibinfo {author}
  {\bibfnamefont {T.}~\bibnamefont {Brydges}}, \bibinfo {author} {\bibfnamefont
  {C.}~\bibnamefont {Maier}}, \bibinfo {author} {\bibfnamefont {M.~K.}\
  \bibnamefont {Joshi}}, \bibinfo {author} {\bibfnamefont {R.}~\bibnamefont
  {Blatt}}, \bibinfo {author} {\bibfnamefont {C.~F.}\ \bibnamefont {Roos}},\
  and\ \bibinfo {author} {\bibfnamefont {P.}~\bibnamefont {Zoller}},\
  }\bibfield  {title} {\bibinfo {title} {{Cross-Platform Verification of
  Intermediate Scale Quantum Devices}},\ }\href
  {https://doi.org/10.1103/PhysRevLett.124.010504} {\bibfield  {journal}
  {\bibinfo  {journal} {Phys. Rev. Lett.}\ }\textbf {\bibinfo {volume} {124}},\
  \bibinfo {pages} {010504} (\bibinfo {year} {2020}{\natexlab{a}})}\BibitemShut
  {NoStop}%
\bibitem [{\citenamefont {Brydges}\ \emph {et~al.}(2019)\citenamefont
  {Brydges}, \citenamefont {Elben}, \citenamefont {Jurcevic}, \citenamefont
  {Vermersch}, \citenamefont {Maier}, \citenamefont {Lanyon}, \citenamefont
  {Zoller}, \citenamefont {Blatt},\ and\ \citenamefont
  {Roos}}]{Brydges2019Entanglement}%
  \BibitemOpen
  \bibfield  {author} {\bibinfo {author} {\bibfnamefont {T.}~\bibnamefont
  {Brydges}}, \bibinfo {author} {\bibfnamefont {A.}~\bibnamefont {Elben}},
  \bibinfo {author} {\bibfnamefont {P.}~\bibnamefont {Jurcevic}}, \bibinfo
  {author} {\bibfnamefont {B.}~\bibnamefont {Vermersch}}, \bibinfo {author}
  {\bibfnamefont {C.}~\bibnamefont {Maier}}, \bibinfo {author} {\bibfnamefont
  {B.~P.}\ \bibnamefont {Lanyon}}, \bibinfo {author} {\bibfnamefont
  {P.}~\bibnamefont {Zoller}}, \bibinfo {author} {\bibfnamefont
  {R.}~\bibnamefont {Blatt}},\ and\ \bibinfo {author} {\bibfnamefont {C.~F.}\
  \bibnamefont {Roos}},\ }\bibfield  {title} {\bibinfo {title} {{Probing
  R\'enyi entanglement entropy via randomized measurements}},\ }\href
  {https://doi.org/10.1126/science.aau4963} {\bibfield  {journal} {\bibinfo
  {journal} {Science}\ }\textbf {\bibinfo {volume} {364}},\ \bibinfo {pages}
  {260} (\bibinfo {year} {2019})}\BibitemShut {NoStop}%
\bibitem [{\citenamefont {Elben}\ \emph
  {et~al.}(2020{\natexlab{b}})\citenamefont {Elben}, \citenamefont {Kueng},
  \citenamefont {Huang}, \citenamefont {van Bijnen}, \citenamefont {Kokail},
  \citenamefont {Dalmonte}, \citenamefont {Calabrese}, \citenamefont {Kraus},
  \citenamefont {Preskill}, \citenamefont {Zoller},\ and\ \citenamefont
  {Vermersch}}]{Elben2020Entanglement}%
  \BibitemOpen
  \bibfield  {author} {\bibinfo {author} {\bibfnamefont {A.}~\bibnamefont
  {Elben}}, \bibinfo {author} {\bibfnamefont {R.}~\bibnamefont {Kueng}},
  \bibinfo {author} {\bibfnamefont {H.-Y.}\ \bibnamefont {Huang}}, \bibinfo
  {author} {\bibfnamefont {R.}~\bibnamefont {van Bijnen}}, \bibinfo {author}
  {\bibfnamefont {C.}~\bibnamefont {Kokail}}, \bibinfo {author} {\bibfnamefont
  {M.}~\bibnamefont {Dalmonte}}, \bibinfo {author} {\bibfnamefont
  {P.}~\bibnamefont {Calabrese}}, \bibinfo {author} {\bibfnamefont
  {B.}~\bibnamefont {Kraus}}, \bibinfo {author} {\bibfnamefont
  {J.}~\bibnamefont {Preskill}}, \bibinfo {author} {\bibfnamefont
  {P.}~\bibnamefont {Zoller}},\ and\ \bibinfo {author} {\bibfnamefont
  {B.}~\bibnamefont {Vermersch}},\ }\bibfield  {title} {\bibinfo {title}
  {{Mixed-State Entanglement from Local Randomized Measurements}},\ }\href
  {https://doi.org/10.1103/PhysRevLett.125.200501} {\bibfield  {journal}
  {\bibinfo  {journal} {Phys. Rev. Lett.}\ }\textbf {\bibinfo {volume} {125}},\
  \bibinfo {pages} {200501} (\bibinfo {year} {2020}{\natexlab{b}})}\BibitemShut
  {NoStop}%
\bibitem [{\citenamefont {Zhou}\ \emph {et~al.}(2020)\citenamefont {Zhou},
  \citenamefont {Zeng},\ and\ \citenamefont {Liu}}]{Zhou2020Entanglement}%
  \BibitemOpen
  \bibfield  {author} {\bibinfo {author} {\bibfnamefont {Y.}~\bibnamefont
  {Zhou}}, \bibinfo {author} {\bibfnamefont {P.}~\bibnamefont {Zeng}},\ and\
  \bibinfo {author} {\bibfnamefont {Z.}~\bibnamefont {Liu}},\ }\bibfield
  {title} {\bibinfo {title} {{Single-Copies Estimation of Entanglement
  Negativity}},\ }\href {https://doi.org/10.1103/PhysRevLett.125.200502}
  {\bibfield  {journal} {\bibinfo  {journal} {Phys. Rev. Lett.}\ }\textbf
  {\bibinfo {volume} {125}},\ \bibinfo {pages} {200502} (\bibinfo {year}
  {2020})}\BibitemShut {NoStop}%
\bibitem [{\citenamefont {Anshu}\ \emph {et~al.}(2020)\citenamefont {Anshu},
  \citenamefont {Arunachalam}, \citenamefont {Kuwahara},\ and\ \citenamefont
  {Soleimanifar}}]{Anshu2020HamiltonianLearning}%
  \BibitemOpen
  \bibfield  {author} {\bibinfo {author} {\bibfnamefont {A.}~\bibnamefont
  {Anshu}}, \bibinfo {author} {\bibfnamefont {S.}~\bibnamefont {Arunachalam}},
  \bibinfo {author} {\bibfnamefont {T.}~\bibnamefont {Kuwahara}},\ and\
  \bibinfo {author} {\bibfnamefont {M.}~\bibnamefont {Soleimanifar}},\
  }\bibfield  {title} {\bibinfo {title} {Sample-efficient learning of quantum
  many-body systems},\ }in\ \href
  {https://doi.org/10.1109/FOCS46700.2020.00069} {\emph {\bibinfo {booktitle}
  {2020 IEEE 61st Annual Symposium on Foundations of Computer Science
  (FOCS)}}}\ (\bibinfo {year} {2020})\ pp.\ \bibinfo {pages}
  {685--691}\BibitemShut {NoStop}%
\bibitem [{\citenamefont {{Hadfield}}\ \emph {et~al.}(2022)\citenamefont
  {{Hadfield}}, \citenamefont {{Bravyi}}, \citenamefont {{Raymond}},\ and\
  \citenamefont {{Mezzacapo}}}]{Hadfield2022HamiltonianLearning}%
  \BibitemOpen
  \bibfield  {author} {\bibinfo {author} {\bibfnamefont {C.}~\bibnamefont
  {{Hadfield}}}, \bibinfo {author} {\bibfnamefont {S.}~\bibnamefont
  {{Bravyi}}}, \bibinfo {author} {\bibfnamefont {R.}~\bibnamefont
  {{Raymond}}},\ and\ \bibinfo {author} {\bibfnamefont {A.}~\bibnamefont
  {{Mezzacapo}}},\ }\bibfield  {title} {\bibinfo {title} {{Measurements of
  Quantum Hamiltonians with Locally-Biased Classical Shadows}},\ }\href
  {https://doi.org/10.1007/s00220-022-04343-8} {\bibfield  {journal} {\bibinfo
  {journal} {Commun. Math. Phys.}\ }\textbf {\bibinfo {volume} {391}},\
  \bibinfo {pages} {951} (\bibinfo {year} {2022})}\BibitemShut {NoStop}%
\bibitem [{\citenamefont {Huang}\ \emph {et~al.}(2023)\citenamefont {Huang},
  \citenamefont {Tong}, \citenamefont {Fang},\ and\ \citenamefont
  {Su}}]{Huang2023HamiltonianLearning}%
  \BibitemOpen
  \bibfield  {author} {\bibinfo {author} {\bibfnamefont {H.-Y.}\ \bibnamefont
  {Huang}}, \bibinfo {author} {\bibfnamefont {Y.}~\bibnamefont {Tong}},
  \bibinfo {author} {\bibfnamefont {D.}~\bibnamefont {Fang}},\ and\ \bibinfo
  {author} {\bibfnamefont {Y.}~\bibnamefont {Su}},\ }\bibfield  {title}
  {\bibinfo {title} {{Learning Many-Body Hamiltonians with Heisenberg-Limited
  Scaling}},\ }\href {https://doi.org/10.1103/PhysRevLett.130.200403}
  {\bibfield  {journal} {\bibinfo  {journal} {Phys. Rev. Lett.}\ }\textbf
  {\bibinfo {volume} {130}},\ \bibinfo {pages} {200403} (\bibinfo {year}
  {2023})}\BibitemShut {NoStop}%
\bibitem [{\citenamefont {{Huang}}\ \emph {et~al.}(2021)\citenamefont
  {{Huang}}, \citenamefont {{Broughton}}, \citenamefont {{Mohseni}},
  \citenamefont {{Babbush}}, \citenamefont {{Boixo}}, \citenamefont {{Neven}},\
  and\ \citenamefont {{McClean}}}]{Huang2021ML}%
  \BibitemOpen
  \bibfield  {author} {\bibinfo {author} {\bibfnamefont {H.-Y.}\ \bibnamefont
  {{Huang}}}, \bibinfo {author} {\bibfnamefont {M.}~\bibnamefont
  {{Broughton}}}, \bibinfo {author} {\bibfnamefont {M.}~\bibnamefont
  {{Mohseni}}}, \bibinfo {author} {\bibfnamefont {R.}~\bibnamefont
  {{Babbush}}}, \bibinfo {author} {\bibfnamefont {S.}~\bibnamefont {{Boixo}}},
  \bibinfo {author} {\bibfnamefont {H.}~\bibnamefont {{Neven}}},\ and\ \bibinfo
  {author} {\bibfnamefont {J.~R.}\ \bibnamefont {{McClean}}},\ }\bibfield
  {title} {\bibinfo {title} {{Power of data in quantum machine learning}},\
  }\href {https://doi.org/10.1038/s41467-021-22539-9} {\bibfield  {journal}
  {\bibinfo  {journal} {Nat. Commun.}\ }\textbf {\bibinfo {volume} {12}},\
  \bibinfo {eid} {2631} (\bibinfo {year} {2021})}\BibitemShut {NoStop}%
\bibitem [{\citenamefont {Jerbi}\ \emph {et~al.}(2024)\citenamefont {Jerbi},
  \citenamefont {Gyurik}, \citenamefont {Marshall}, \citenamefont {Molteni},\
  and\ \citenamefont {Dunjko}}]{Jerbi2023ML}%
  \BibitemOpen
  \bibfield  {author} {\bibinfo {author} {\bibfnamefont {S.}~\bibnamefont
  {Jerbi}}, \bibinfo {author} {\bibfnamefont {C.}~\bibnamefont {Gyurik}},
  \bibinfo {author} {\bibfnamefont {S.~C.}\ \bibnamefont {Marshall}}, \bibinfo
  {author} {\bibfnamefont {R.}~\bibnamefont {Molteni}},\ and\ \bibinfo {author}
  {\bibfnamefont {V.}~\bibnamefont {Dunjko}},\ }\bibfield  {title} {\bibinfo
  {title} {{Shadows of quantum machine learning}},\ }\href
  {https://doi.org/10.1038/s41467-024-49877-8} {\bibfield  {journal} {\bibinfo
  {journal} {Nat. Commun.}\ }\textbf {\bibinfo {volume} {15}},\ \bibinfo
  {pages} {5676} (\bibinfo {year} {2024})}\BibitemShut {NoStop}%
\bibitem [{\citenamefont {Struchalin}\ \emph {et~al.}(2021)\citenamefont
  {Struchalin}, \citenamefont {Zagorovskii}, \citenamefont {Kovlakov},
  \citenamefont {Straupe},\ and\ \citenamefont
  {Kulik}}]{Struchalin2021Experiment}%
  \BibitemOpen
  \bibfield  {author} {\bibinfo {author} {\bibfnamefont {G.}~\bibnamefont
  {Struchalin}}, \bibinfo {author} {\bibfnamefont {Y.~A.}\ \bibnamefont
  {Zagorovskii}}, \bibinfo {author} {\bibfnamefont {E.}~\bibnamefont
  {Kovlakov}}, \bibinfo {author} {\bibfnamefont {S.}~\bibnamefont {Straupe}},\
  and\ \bibinfo {author} {\bibfnamefont {S.}~\bibnamefont {Kulik}},\ }\bibfield
   {title} {\bibinfo {title} {{Experimental Estimation of Quantum State
  Properties from Classical Shadows}},\ }\href
  {https://doi.org/10.1103/PRXQuantum.2.010307} {\bibfield  {journal} {\bibinfo
   {journal} {PRX Quantum}\ }\textbf {\bibinfo {volume} {2}},\ \bibinfo {pages}
  {010307} (\bibinfo {year} {2021})}\BibitemShut {NoStop}%
\bibitem [{\citenamefont {Zhang}\ \emph {et~al.}(2021)\citenamefont {Zhang},
  \citenamefont {Sun}, \citenamefont {Fang}, \citenamefont {Zhang},
  \citenamefont {Yuan},\ and\ \citenamefont {Lu}}]{Zhang2021Experiment}%
  \BibitemOpen
  \bibfield  {author} {\bibinfo {author} {\bibfnamefont {T.}~\bibnamefont
  {Zhang}}, \bibinfo {author} {\bibfnamefont {J.}~\bibnamefont {Sun}}, \bibinfo
  {author} {\bibfnamefont {X.-X.}\ \bibnamefont {Fang}}, \bibinfo {author}
  {\bibfnamefont {X.-M.}\ \bibnamefont {Zhang}}, \bibinfo {author}
  {\bibfnamefont {X.}~\bibnamefont {Yuan}},\ and\ \bibinfo {author}
  {\bibfnamefont {H.}~\bibnamefont {Lu}},\ }\bibfield  {title} {\bibinfo
  {title} {{Experimental Quantum State Measurement with Classical Shadows}},\
  }\href {https://doi.org/10.1103/PhysRevLett.127.200501} {\bibfield  {journal}
  {\bibinfo  {journal} {Phys. Rev. Lett.}\ }\textbf {\bibinfo {volume} {127}},\
  \bibinfo {pages} {200501} (\bibinfo {year} {2021})}\BibitemShut {NoStop}%
\bibitem [{\citenamefont {Stricker}\ \emph {et~al.}(2022)\citenamefont
  {Stricker}, \citenamefont {Meth}, \citenamefont {Postler}, \citenamefont
  {Edmunds}, \citenamefont {Ferrie}, \citenamefont {Blatt}, \citenamefont
  {Schindler}, \citenamefont {Monz}, \citenamefont {Kueng},\ and\ \citenamefont
  {Ringbauer}}]{Stricker2022Experiment}%
  \BibitemOpen
  \bibfield  {author} {\bibinfo {author} {\bibfnamefont {R.}~\bibnamefont
  {Stricker}}, \bibinfo {author} {\bibfnamefont {M.}~\bibnamefont {Meth}},
  \bibinfo {author} {\bibfnamefont {L.}~\bibnamefont {Postler}}, \bibinfo
  {author} {\bibfnamefont {C.}~\bibnamefont {Edmunds}}, \bibinfo {author}
  {\bibfnamefont {C.}~\bibnamefont {Ferrie}}, \bibinfo {author} {\bibfnamefont
  {R.}~\bibnamefont {Blatt}}, \bibinfo {author} {\bibfnamefont
  {P.}~\bibnamefont {Schindler}}, \bibinfo {author} {\bibfnamefont
  {T.}~\bibnamefont {Monz}}, \bibinfo {author} {\bibfnamefont {R.}~\bibnamefont
  {Kueng}},\ and\ \bibinfo {author} {\bibfnamefont {M.}~\bibnamefont
  {Ringbauer}},\ }\bibfield  {title} {\bibinfo {title} {{Experimental
  Single-Setting Quantum State Tomography}},\ }\href
  {https://doi.org/10.1103/PRXQuantum.3.040310} {\bibfield  {journal} {\bibinfo
   {journal} {PRX Quantum}\ }\textbf {\bibinfo {volume} {3}},\ \bibinfo {pages}
  {040310} (\bibinfo {year} {2022})}\BibitemShut {NoStop}%
\bibitem [{\citenamefont {Huggins}\ \emph {et~al.}(2022)\citenamefont
  {Huggins}, \citenamefont {O'Gorman}, \citenamefont {Rubin}, \citenamefont
  {Reichman}, \citenamefont {Babbush},\ and\ \citenamefont
  {Lee}}]{William2022Experiment}%
  \BibitemOpen
  \bibfield  {author} {\bibinfo {author} {\bibfnamefont {W.~J.}\ \bibnamefont
  {Huggins}}, \bibinfo {author} {\bibfnamefont {B.~A.}\ \bibnamefont
  {O'Gorman}}, \bibinfo {author} {\bibfnamefont {N.~C.}\ \bibnamefont {Rubin}},
  \bibinfo {author} {\bibfnamefont {D.~R.}\ \bibnamefont {Reichman}}, \bibinfo
  {author} {\bibfnamefont {R.}~\bibnamefont {Babbush}},\ and\ \bibinfo {author}
  {\bibfnamefont {J.}~\bibnamefont {Lee}},\ }\bibfield  {title} {\bibinfo
  {title} {{Unbiasing fermionic quantum Monte Carlo with a quantum computer}},\
  }\href {https://doi.org/10.1038/s41586-021-04351-z} {\bibfield  {journal}
  {\bibinfo  {journal} {Nature}\ }\textbf {\bibinfo {volume} {603}},\ \bibinfo
  {pages} {416} (\bibinfo {year} {2022})}\BibitemShut {NoStop}%
\bibitem [{\citenamefont {Aaronson}\ and\ \citenamefont
  {Gottesman}(2004)}]{Aaronson2004StabFormalism}%
  \BibitemOpen
  \bibfield  {author} {\bibinfo {author} {\bibfnamefont {S.}~\bibnamefont
  {Aaronson}}\ and\ \bibinfo {author} {\bibfnamefont {D.}~\bibnamefont
  {Gottesman}},\ }\bibfield  {title} {\bibinfo {title} {{Improved simulation of
  stabilizer circuits}},\ }\href {https://doi.org/10.1103/PhysRevA.70.052328}
  {\bibfield  {journal} {\bibinfo  {journal} {Phys. Rev. A}\ }\textbf {\bibinfo
  {volume} {70}},\ \bibinfo {pages} {052328} (\bibinfo {year}
  {2004})}\BibitemShut {NoStop}%
\bibitem [{\citenamefont {Kueng}\ and\ \citenamefont
  {Gross}(2015)}]{Kueng20153design}%
  \BibitemOpen
  \bibfield  {author} {\bibinfo {author} {\bibfnamefont {R.}~\bibnamefont
  {Kueng}}\ and\ \bibinfo {author} {\bibfnamefont {D.}~\bibnamefont {Gross}},\
  }\bibfield  {title} {\bibinfo {title} {{Qubit stabilizer states are complex
  projective 3-designs}},\ }\href {https://arxiv.org/abs/1510.02767} {\bibfield
   {journal} {\bibinfo  {journal} {arXiv:1510.02767}\ } (\bibinfo {year}
  {2015})}\BibitemShut {NoStop}%
\bibitem [{\citenamefont {Webb}(2016)}]{Webb20163Design}%
  \BibitemOpen
  \bibfield  {author} {\bibinfo {author} {\bibfnamefont {Z.}~\bibnamefont
  {Webb}},\ }\bibfield  {title} {\bibinfo {title} {{The Clifford group forms a
  unitary 3-design}},\ }\href {https://arxiv.org/abs/1510.02769} {\bibfield
  {journal} {\bibinfo  {journal} {Quant. Inf. Comp.}\ }\textbf {\bibinfo
  {volume} {16}},\ \bibinfo {pages} {1379} (\bibinfo {year}
  {2016})}\BibitemShut {NoStop}%
\bibitem [{\citenamefont {Zhu}(2017)}]{Zhu20173Design}%
  \BibitemOpen
  \bibfield  {author} {\bibinfo {author} {\bibfnamefont {H.}~\bibnamefont
  {Zhu}},\ }\bibfield  {title} {\bibinfo {title} {{Multiqubit {C}lifford groups
  are unitary 3-designs}},\ }\href {https://doi.org/10.1103/PhysRevA.96.062336}
  {\bibfield  {journal} {\bibinfo  {journal} {Phys. Rev. A}\ }\textbf {\bibinfo
  {volume} {96}},\ \bibinfo {pages} {062336} (\bibinfo {year}
  {2017})}\BibitemShut {NoStop}%
\bibitem [{\citenamefont {Bravyi}\ and\ \citenamefont
  {Kitaev}(2005)}]{Bravyi2005MagicState}%
  \BibitemOpen
  \bibfield  {author} {\bibinfo {author} {\bibfnamefont {S.}~\bibnamefont
  {Bravyi}}\ and\ \bibinfo {author} {\bibfnamefont {A.}~\bibnamefont
  {Kitaev}},\ }\bibfield  {title} {\bibinfo {title} {{Universal quantum
  computation with ideal Clifford gates and noisy ancillas}},\ }\href
  {https://doi.org/10.1103/PhysRevA.71.022316} {\bibfield  {journal} {\bibinfo
  {journal} {Phys. Rev. A}\ }\textbf {\bibinfo {volume} {71}},\ \bibinfo
  {pages} {022316} (\bibinfo {year} {2005})}\BibitemShut {NoStop}%
\bibitem [{\citenamefont {{Briegel}}\ \emph {et~al.}(2009)\citenamefont
  {{Briegel}}, \citenamefont {{Browne}}, \citenamefont {{D{\"u}r}},
  \citenamefont {{Raussendorf}},\ and\ \citenamefont {{Van den
  Nest}}}]{Briegel2009QComputation}%
  \BibitemOpen
  \bibfield  {author} {\bibinfo {author} {\bibfnamefont {H.~J.}\ \bibnamefont
  {{Briegel}}}, \bibinfo {author} {\bibfnamefont {D.~E.}\ \bibnamefont
  {{Browne}}}, \bibinfo {author} {\bibfnamefont {W.}~\bibnamefont {{D{\"u}r}}},
  \bibinfo {author} {\bibfnamefont {R.}~\bibnamefont {{Raussendorf}}},\ and\
  \bibinfo {author} {\bibfnamefont {M.}~\bibnamefont {{Van den Nest}}},\
  }\bibfield  {title} {\bibinfo {title} {{Measurement-based quantum
  computation}},\ }\href {https://doi.org/10.1038/nphys1157} {\bibfield
  {journal} {\bibinfo  {journal} {Nat. Phys.}\ }\textbf {\bibinfo {volume}
  {5}},\ \bibinfo {pages} {19} (\bibinfo {year} {2009})}\BibitemShut {NoStop}%
\bibitem [{\citenamefont {{Gottesman}}(1997)}]{Gottesman1997StabilizerCode}%
  \BibitemOpen
  \bibfield  {author} {\bibinfo {author} {\bibfnamefont {D.}~\bibnamefont
  {{Gottesman}}},\ }\emph {\bibinfo {title} {{Stabilizer codes and quantum
  error correction}}},\ \href@noop {} {Ph.D. thesis},\ \bibinfo  {school}
  {California Institute of Technology} (\bibinfo {year} {1997})\BibitemShut
  {NoStop}%
\bibitem [{\citenamefont {Campbell}\ \emph {et~al.}(2017)\citenamefont
  {Campbell}, \citenamefont {Terhal},\ and\ \citenamefont
  {Vuillot}}]{Campbell2017FaultTolerant}%
  \BibitemOpen
  \bibfield  {author} {\bibinfo {author} {\bibfnamefont {E.~T.}\ \bibnamefont
  {Campbell}}, \bibinfo {author} {\bibfnamefont {B.~M.}\ \bibnamefont
  {Terhal}},\ and\ \bibinfo {author} {\bibfnamefont {C.}~\bibnamefont
  {Vuillot}},\ }\bibfield  {title} {\bibinfo {title} {{Roads towards
  fault-tolerant universal quantum computation}},\ }\href
  {https://doi.org/https://doi.org/10.1038/nature23460} {\bibfield  {journal}
  {\bibinfo  {journal} {Nature}\ }\textbf {\bibinfo {volume} {549}},\ \bibinfo
  {pages} {172} (\bibinfo {year} {2017})}\BibitemShut {NoStop}%
\bibitem [{\citenamefont {Helsen}\ and\ \citenamefont
  {Walter}(2023)}]{Helsen2023MultiShot}%
  \BibitemOpen
  \bibfield  {author} {\bibinfo {author} {\bibfnamefont {J.}~\bibnamefont
  {Helsen}}\ and\ \bibinfo {author} {\bibfnamefont {M.}~\bibnamefont
  {Walter}},\ }\bibfield  {title} {\bibinfo {title} {{Thrifty Shadow
  Estimation: Reusing Quantum Circuits and Bounding Tails}},\ }\href
  {https://doi.org/10.1103/PhysRevLett.131.240602} {\bibfield  {journal}
  {\bibinfo  {journal} {Phys. Rev. Lett.}\ }\textbf {\bibinfo {volume} {131}},\
  \bibinfo {pages} {240602} (\bibinfo {year} {2023})}\BibitemShut {NoStop}%
\bibitem [{\citenamefont {Zhou}\ and\ \citenamefont
  {Liu}(2023)}]{Zhou2023MultiShot}%
  \BibitemOpen
  \bibfield  {author} {\bibinfo {author} {\bibfnamefont {Y.}~\bibnamefont
  {Zhou}}\ and\ \bibinfo {author} {\bibfnamefont {Q.}~\bibnamefont {Liu}},\
  }\bibfield  {title} {\bibinfo {title} {{Performance analysis of multi-shot
  shadow estimation}},\ }\href {https://doi.org/10.22331/q-2023-06-29-1044}
  {\bibfield  {journal} {\bibinfo  {journal} {{Quantum}}\ }\textbf {\bibinfo
  {volume} {7}},\ \bibinfo {pages} {1044} (\bibinfo {year} {2023})}\BibitemShut
  {NoStop}%
\bibitem [{\citenamefont {Seif}\ \emph {et~al.}(2023)\citenamefont {Seif},
  \citenamefont {Cian}, \citenamefont {Zhou}, \citenamefont {Chen},\ and\
  \citenamefont {Jiang}}]{Seif2023MultiShot}%
  \BibitemOpen
  \bibfield  {author} {\bibinfo {author} {\bibfnamefont {A.}~\bibnamefont
  {Seif}}, \bibinfo {author} {\bibfnamefont {Z.-P.}\ \bibnamefont {Cian}},
  \bibinfo {author} {\bibfnamefont {S.}~\bibnamefont {Zhou}}, \bibinfo {author}
  {\bibfnamefont {S.}~\bibnamefont {Chen}},\ and\ \bibinfo {author}
  {\bibfnamefont {L.}~\bibnamefont {Jiang}},\ }\bibfield  {title} {\bibinfo
  {title} {{Shadow Distillation: Quantum Error Mitigation with Classical
  Shadows for Near-Term Quantum Processors}},\ }\href
  {https://doi.org/10.1103/PRXQuantum.4.010303} {\bibfield  {journal} {\bibinfo
   {journal} {PRX Quantum}\ }\textbf {\bibinfo {volume} {4}},\ \bibinfo {pages}
  {010303} (\bibinfo {year} {2023})}\BibitemShut {NoStop}%
\bibitem [{\citenamefont {Vermersch}\ \emph {et~al.}(2024)\citenamefont
  {Vermersch}, \citenamefont {Rath}, \citenamefont {Sundar}, \citenamefont
  {Branciard}, \citenamefont {Preskill},\ and\ \citenamefont
  {Elben}}]{Vermersch2024MultiShot}%
  \BibitemOpen
  \bibfield  {author} {\bibinfo {author} {\bibfnamefont {B.}~\bibnamefont
  {Vermersch}}, \bibinfo {author} {\bibfnamefont {A.}~\bibnamefont {Rath}},
  \bibinfo {author} {\bibfnamefont {B.}~\bibnamefont {Sundar}}, \bibinfo
  {author} {\bibfnamefont {C.}~\bibnamefont {Branciard}}, \bibinfo {author}
  {\bibfnamefont {J.}~\bibnamefont {Preskill}},\ and\ \bibinfo {author}
  {\bibfnamefont {A.}~\bibnamefont {Elben}},\ }\bibfield  {title} {\bibinfo
  {title} {{Enhanced Estimation of Quantum Properties with Common Randomized
  Measurements}},\ }\href {https://doi.org/10.1103/PRXQuantum.5.010352}
  {\bibfield  {journal} {\bibinfo  {journal} {PRX Quantum}\ }\textbf {\bibinfo
  {volume} {5}},\ \bibinfo {pages} {010352} (\bibinfo {year}
  {2024})}\BibitemShut {NoStop}%
\bibitem [{\citenamefont {Zhu}\ \emph {et~al.}(2016)\citenamefont {Zhu},
  \citenamefont {Kueng}, \citenamefont {Grassl},\ and\ \citenamefont
  {Gross}}]{Zhu2016Fail4Design}%
  \BibitemOpen
  \bibfield  {author} {\bibinfo {author} {\bibfnamefont {H.}~\bibnamefont
  {Zhu}}, \bibinfo {author} {\bibfnamefont {R.}~\bibnamefont {Kueng}}, \bibinfo
  {author} {\bibfnamefont {M.}~\bibnamefont {Grassl}},\ and\ \bibinfo {author}
  {\bibfnamefont {D.}~\bibnamefont {Gross}},\ }\bibfield  {title} {\bibinfo
  {title} {{The {C}lifford group fails gracefully to be a unitary 4-design}},\
  }\href {https://arxiv.org/abs/1609.08172} {\bibfield  {journal} {\bibinfo
  {journal} {arXiv:1609.08172}\ } (\bibinfo {year} {2016})}\BibitemShut
  {NoStop}%
\bibitem [{\citenamefont {Gross}\ \emph {et~al.}(2021)\citenamefont {Gross},
  \citenamefont {Nezami},\ and\ \citenamefont {Walter}}]{Gross2021Duality}%
  \BibitemOpen
  \bibfield  {author} {\bibinfo {author} {\bibfnamefont {D.}~\bibnamefont
  {Gross}}, \bibinfo {author} {\bibfnamefont {S.}~\bibnamefont {Nezami}},\ and\
  \bibinfo {author} {\bibfnamefont {M.}~\bibnamefont {Walter}},\ }\bibfield
  {title} {\bibinfo {title} {{Schur--Weyl duality for the Clifford group with
  applications: Property testing, a robust Hudson theorem, and de Finetti
  representations}},\ }\href
  {https://doi.org/https://doi.org/10.1007/s00220-021-04118-7} {\bibfield
  {journal} {\bibinfo  {journal} {Commun. Math. Phys.}\ }\textbf {\bibinfo
  {volume} {385}},\ \bibinfo {pages} {1325} (\bibinfo {year}
  {2021})}\BibitemShut {NoStop}%
\bibitem [{\citenamefont {Veitch}\ \emph {et~al.}(2014)\citenamefont {Veitch},
  \citenamefont {Mousavian}, \citenamefont {Gottesman},\ and\ \citenamefont
  {Emerson}}]{Veitch2014StabResource}%
  \BibitemOpen
  \bibfield  {author} {\bibinfo {author} {\bibfnamefont {V.}~\bibnamefont
  {Veitch}}, \bibinfo {author} {\bibfnamefont {S.~A.~H.}\ \bibnamefont
  {Mousavian}}, \bibinfo {author} {\bibfnamefont {D.}~\bibnamefont
  {Gottesman}},\ and\ \bibinfo {author} {\bibfnamefont {J.}~\bibnamefont
  {Emerson}},\ }\bibfield  {title} {\bibinfo {title} {{The resource theory of
  stabilizer quantum computation}},\ }\href
  {https://doi.org/10.1088/1367-2630/16/1/013009} {\bibfield  {journal}
  {\bibinfo  {journal} {New J. Phys.}\ }\textbf {\bibinfo {volume} {16}},\
  \bibinfo {pages} {013009} (\bibinfo {year} {2014})}\BibitemShut {NoStop}%
\bibitem [{\citenamefont {Bravyi}\ and\ \citenamefont
  {Gosset}(2016)}]{Bravyi2016TCount}%
  \BibitemOpen
  \bibfield  {author} {\bibinfo {author} {\bibfnamefont {S.}~\bibnamefont
  {Bravyi}}\ and\ \bibinfo {author} {\bibfnamefont {D.}~\bibnamefont
  {Gosset}},\ }\bibfield  {title} {\bibinfo {title} {{Improved classical
  simulation of quantum circuits dominated by Clifford gates}},\ }\href
  {https://doi.org/10.1103/PhysRevLett.116.250501} {\bibfield  {journal}
  {\bibinfo  {journal} {Phys. Rev. Lett.}\ }\textbf {\bibinfo {volume} {116}},\
  \bibinfo {pages} {250501} (\bibinfo {year} {2016})}\BibitemShut {NoStop}%
\bibitem [{\citenamefont {Howard}\ and\ \citenamefont
  {Campbell}(2017)}]{Howard2017RoM}%
  \BibitemOpen
  \bibfield  {author} {\bibinfo {author} {\bibfnamefont {M.}~\bibnamefont
  {Howard}}\ and\ \bibinfo {author} {\bibfnamefont {E.}~\bibnamefont
  {Campbell}},\ }\bibfield  {title} {\bibinfo {title} {{Application of a
  Resource Theory for Magic States to Fault-Tolerant Quantum Computing}},\
  }\href {https://doi.org/10.1103/PhysRevLett.118.090501} {\bibfield  {journal}
  {\bibinfo  {journal} {Phys. Rev. Lett.}\ }\textbf {\bibinfo {volume} {118}},\
  \bibinfo {pages} {090501} (\bibinfo {year} {2017})}\BibitemShut {NoStop}%
\bibitem [{\citenamefont {Hinsche}\ \emph {et~al.}(2023)\citenamefont
  {Hinsche}, \citenamefont {Ioannou}, \citenamefont {Nietner}, \citenamefont
  {Haferkamp}, \citenamefont {Quek}, \citenamefont {Hangleiter}, \citenamefont
  {Seifert}, \citenamefont {Eisert},\ and\ \citenamefont
  {Sweke}}]{Hinsche2023THard}%
  \BibitemOpen
  \bibfield  {author} {\bibinfo {author} {\bibfnamefont {M.}~\bibnamefont
  {Hinsche}}, \bibinfo {author} {\bibfnamefont {M.}~\bibnamefont {Ioannou}},
  \bibinfo {author} {\bibfnamefont {A.}~\bibnamefont {Nietner}}, \bibinfo
  {author} {\bibfnamefont {J.}~\bibnamefont {Haferkamp}}, \bibinfo {author}
  {\bibfnamefont {Y.}~\bibnamefont {Quek}}, \bibinfo {author} {\bibfnamefont
  {D.}~\bibnamefont {Hangleiter}}, \bibinfo {author} {\bibfnamefont {J.-P.}\
  \bibnamefont {Seifert}}, \bibinfo {author} {\bibfnamefont {J.}~\bibnamefont
  {Eisert}},\ and\ \bibinfo {author} {\bibfnamefont {R.}~\bibnamefont
  {Sweke}},\ }\bibfield  {title} {\bibinfo {title} {{One $T$ Gate Makes
  Distribution Learning Hard}},\ }\href
  {https://doi.org/10.1103/PhysRevLett.130.240602} {\bibfield  {journal}
  {\bibinfo  {journal} {Phys. Rev. Lett.}\ }\textbf {\bibinfo {volume} {130}},\
  \bibinfo {pages} {240602} (\bibinfo {year} {2023})}\BibitemShut {NoStop}%
\bibitem [{\citenamefont {Leone}\ \emph {et~al.}(2022)\citenamefont {Leone},
  \citenamefont {Oliviero},\ and\ \citenamefont {Hamma}}]{Lorenzo2022SRE}%
  \BibitemOpen
  \bibfield  {author} {\bibinfo {author} {\bibfnamefont {L.}~\bibnamefont
  {Leone}}, \bibinfo {author} {\bibfnamefont {S.~F.~E.}\ \bibnamefont
  {Oliviero}},\ and\ \bibinfo {author} {\bibfnamefont {A.}~\bibnamefont
  {Hamma}},\ }\bibfield  {title} {\bibinfo {title} {{Stabilizer R\'enyi
  Entropy}},\ }\href {https://doi.org/10.1103/PhysRevLett.128.050402}
  {\bibfield  {journal} {\bibinfo  {journal} {Phys. Rev. Lett.}\ }\textbf
  {\bibinfo {volume} {128}},\ \bibinfo {pages} {050402} (\bibinfo {year}
  {2022})}\BibitemShut {NoStop}%
\bibitem [{\citenamefont {Leone}\ \emph {et~al.}(2021)\citenamefont {Leone},
  \citenamefont {Oliviero}, \citenamefont {Zhou},\ and\ \citenamefont
  {Hamma}}]{Leone2021Interleaved}%
  \BibitemOpen
  \bibfield  {author} {\bibinfo {author} {\bibfnamefont {L.}~\bibnamefont
  {Leone}}, \bibinfo {author} {\bibfnamefont {S.~F.~E.}\ \bibnamefont
  {Oliviero}}, \bibinfo {author} {\bibfnamefont {Y.}~\bibnamefont {Zhou}},\
  and\ \bibinfo {author} {\bibfnamefont {A.}~\bibnamefont {Hamma}},\ }\bibfield
   {title} {\bibinfo {title} {Quantum {C}haos is {Q}uantum},\ }\href
  {https://doi.org/10.22331/q-2021-05-04-453} {\bibfield  {journal} {\bibinfo
  {journal} {{Quantum}}\ }\textbf {\bibinfo {volume} {5}},\ \bibinfo {pages}
  {453} (\bibinfo {year} {2021})}\BibitemShut {NoStop}%
\bibitem [{\citenamefont {Haferkamp}(2022)}]{Haferkamp2022Interleaved}%
  \BibitemOpen
  \bibfield  {author} {\bibinfo {author} {\bibfnamefont {J.}~\bibnamefont
  {Haferkamp}},\ }\bibfield  {title} {\bibinfo {title} {{Random quantum
  circuits are approximate unitary {$t$}-designs in depth
  {$O\left(nt^{5+o(1)}\right)$}}},\ }\href
  {https://doi.org/10.22331/q-2022-09-08-795} {\bibfield  {journal} {\bibinfo
  {journal} {{Quantum}}\ }\textbf {\bibinfo {volume} {6}},\ \bibinfo {pages}
  {795} (\bibinfo {year} {2022})}\BibitemShut {NoStop}%
\bibitem [{\citenamefont {Haferkamp}\ \emph {et~al.}(2023)\citenamefont
  {Haferkamp}, \citenamefont {Montealegre-Mora}, \citenamefont {Heinrich},
  \citenamefont {Eisert}, \citenamefont {Gross},\ and\ \citenamefont
  {Roth}}]{Haferkamp2023Interleaved}%
  \BibitemOpen
  \bibfield  {author} {\bibinfo {author} {\bibfnamefont {J.}~\bibnamefont
  {Haferkamp}}, \bibinfo {author} {\bibfnamefont {F.}~\bibnamefont
  {Montealegre-Mora}}, \bibinfo {author} {\bibfnamefont {M.}~\bibnamefont
  {Heinrich}}, \bibinfo {author} {\bibfnamefont {J.}~\bibnamefont {Eisert}},
  \bibinfo {author} {\bibfnamefont {D.}~\bibnamefont {Gross}},\ and\ \bibinfo
  {author} {\bibfnamefont {I.}~\bibnamefont {Roth}},\ }\bibfield  {title}
  {\bibinfo {title} {{Efficient unitary designs with a system-size independent
  number of non-Clifford gates}},\ }\href
  {https://doi.org/https://doi.org/10.1007/s00220-022-04507-6} {\bibfield
  {journal} {\bibinfo  {journal} {Commun. Math. Phys.}\ }\textbf {\bibinfo
  {volume} {397}},\ \bibinfo {pages} {995} (\bibinfo {year}
  {2023})}\BibitemShut {NoStop}%
\bibitem [{\citenamefont {Bravyi}\ \emph {et~al.}(2016)\citenamefont {Bravyi},
  \citenamefont {Smith},\ and\ \citenamefont {Smolin}}]{Bravyi2016TCount2}%
  \BibitemOpen
  \bibfield  {author} {\bibinfo {author} {\bibfnamefont {S.}~\bibnamefont
  {Bravyi}}, \bibinfo {author} {\bibfnamefont {G.}~\bibnamefont {Smith}},\ and\
  \bibinfo {author} {\bibfnamefont {J.~A.}\ \bibnamefont {Smolin}},\ }\bibfield
   {title} {\bibinfo {title} {{Trading Classical and Quantum Computational
  Resources}},\ }\href {https://doi.org/10.1103/PhysRevX.6.021043} {\bibfield
  {journal} {\bibinfo  {journal} {Phys. Rev. X}\ }\textbf {\bibinfo {volume}
  {6}},\ \bibinfo {pages} {021043} (\bibinfo {year} {2016})}\BibitemShut
  {NoStop}%
\bibitem [{\citenamefont {Gross}(2006)}]{Gross2006CharFunction}%
  \BibitemOpen
  \bibfield  {author} {\bibinfo {author} {\bibfnamefont {D.}~\bibnamefont
  {Gross}},\ }\bibfield  {title} {\bibinfo {title} {{Hudson's theorem for
  finite-dimensional quantum systems}},\ }\href
  {https://doi.org/10.1063/1.2393152} {\bibfield  {journal} {\bibinfo
  {journal} {J. Math. Phys.}\ }\textbf {\bibinfo {volume} {47}},\ \bibinfo
  {pages} {122107} (\bibinfo {year} {2006})}\BibitemShut {NoStop}%
\bibitem [{\citenamefont {Dai}\ \emph {et~al.}(2022)\citenamefont {Dai},
  \citenamefont {Fu},\ and\ \citenamefont {Luo}}]{Dai2022CharFunction}%
  \BibitemOpen
  \bibfield  {author} {\bibinfo {author} {\bibfnamefont {H.}~\bibnamefont
  {Dai}}, \bibinfo {author} {\bibfnamefont {S.}~\bibnamefont {Fu}},\ and\
  \bibinfo {author} {\bibfnamefont {S.}~\bibnamefont {Luo}},\ }\bibfield
  {title} {\bibinfo {title} {{Detecting magic states via characteristic
  functions}},\ }\href
  {https://doi.org/https://doi.org/10.1007/s10773-022-05027-8} {\bibfield
  {journal} {\bibinfo  {journal} {Int. J. Theor. Phys.}\ }\textbf {\bibinfo
  {volume} {61}},\ \bibinfo {pages} {35} (\bibinfo {year} {2022})}\BibitemShut
  {NoStop}%
\bibitem [{\citenamefont {Gu}\ \emph {et~al.}(2024)\citenamefont {Gu},
  \citenamefont {Leone}, \citenamefont {Ghosh}, \citenamefont {Eisert},
  \citenamefont {Yelin},\ and\ \citenamefont {Quek}}]{Gu2024Pseudomagic}%
  \BibitemOpen
  \bibfield  {author} {\bibinfo {author} {\bibfnamefont {A.}~\bibnamefont
  {Gu}}, \bibinfo {author} {\bibfnamefont {L.}~\bibnamefont {Leone}}, \bibinfo
  {author} {\bibfnamefont {S.}~\bibnamefont {Ghosh}}, \bibinfo {author}
  {\bibfnamefont {J.}~\bibnamefont {Eisert}}, \bibinfo {author} {\bibfnamefont
  {S.~F.}\ \bibnamefont {Yelin}},\ and\ \bibinfo {author} {\bibfnamefont
  {Y.}~\bibnamefont {Quek}},\ }\bibfield  {title} {\bibinfo {title}
  {{Pseudomagic Quantum States}},\ }\href
  {https://doi.org/10.1103/PhysRevLett.132.210602} {\bibfield  {journal}
  {\bibinfo  {journal} {Phys. Rev. Lett.}\ }\textbf {\bibinfo {volume} {132}},\
  \bibinfo {pages} {210602} (\bibinfo {year} {2024})}\BibitemShut {NoStop}%
\bibitem [{\citenamefont {Oliviero}\ \emph {et~al.}(2022)\citenamefont
  {Oliviero}, \citenamefont {Leone},\ and\ \citenamefont
  {Hamma}}]{Oliviero2022IsingSRE}%
  \BibitemOpen
  \bibfield  {author} {\bibinfo {author} {\bibfnamefont {S.~F.~E.}\
  \bibnamefont {Oliviero}}, \bibinfo {author} {\bibfnamefont {L.}~\bibnamefont
  {Leone}},\ and\ \bibinfo {author} {\bibfnamefont {A.}~\bibnamefont {Hamma}},\
  }\bibfield  {title} {\bibinfo {title} {{Magic-state resource theory for the
  ground state of the transverse-field Ising model}},\ }\href
  {https://doi.org/10.1103/PhysRevA.106.042426} {\bibfield  {journal} {\bibinfo
   {journal} {Phys. Rev. A}\ }\textbf {\bibinfo {volume} {106}},\ \bibinfo
  {pages} {042426} (\bibinfo {year} {2022})}\BibitemShut {NoStop}%
\bibitem [{\citenamefont {Haug}\ and\ \citenamefont
  {Piroli}(2023)}]{Haug2023IsingSRE}%
  \BibitemOpen
  \bibfield  {author} {\bibinfo {author} {\bibfnamefont {T.}~\bibnamefont
  {Haug}}\ and\ \bibinfo {author} {\bibfnamefont {L.}~\bibnamefont {Piroli}},\
  }\bibfield  {title} {\bibinfo {title} {{Quantifying nonstabilizerness of
  matrix product states}},\ }\href
  {https://doi.org/10.1103/PhysRevB.107.035148} {\bibfield  {journal} {\bibinfo
   {journal} {Phys. Rev. B}\ }\textbf {\bibinfo {volume} {107}},\ \bibinfo
  {pages} {035148} (\bibinfo {year} {2023})}\BibitemShut {NoStop}%
\bibitem [{\citenamefont {Tarabunga}\ \emph {et~al.}(2023)\citenamefont
  {Tarabunga}, \citenamefont {Tirrito}, \citenamefont {Chanda},\ and\
  \citenamefont {Dalmonte}}]{Tarabunga2023IsingSRE}%
  \BibitemOpen
  \bibfield  {author} {\bibinfo {author} {\bibfnamefont {P.~S.}\ \bibnamefont
  {Tarabunga}}, \bibinfo {author} {\bibfnamefont {E.}~\bibnamefont {Tirrito}},
  \bibinfo {author} {\bibfnamefont {T.}~\bibnamefont {Chanda}},\ and\ \bibinfo
  {author} {\bibfnamefont {M.}~\bibnamefont {Dalmonte}},\ }\bibfield  {title}
  {\bibinfo {title} {{Many-Body Magic Via Pauli-Markov Chains---From
  Criticality to Gauge Theories}},\ }\href
  {https://doi.org/10.1103/PRXQuantum.4.040317} {\bibfield  {journal} {\bibinfo
   {journal} {PRX Quantum}\ }\textbf {\bibinfo {volume} {4}},\ \bibinfo {pages}
  {040317} (\bibinfo {year} {2023})}\BibitemShut {NoStop}%
\bibitem [{\citenamefont {Liu}\ and\ \citenamefont
  {Winter}(2022)}]{Liu2022HypergraphSRE}%
  \BibitemOpen
  \bibfield  {author} {\bibinfo {author} {\bibfnamefont {Z.-W.}\ \bibnamefont
  {Liu}}\ and\ \bibinfo {author} {\bibfnamefont {A.}~\bibnamefont {Winter}},\
  }\bibfield  {title} {\bibinfo {title} {{Many-Body Quantum Magic}},\ }\href
  {https://doi.org/10.1103/PRXQuantum.3.020333} {\bibfield  {journal} {\bibinfo
   {journal} {PRX Quantum}\ }\textbf {\bibinfo {volume} {3}},\ \bibinfo {pages}
  {020333} (\bibinfo {year} {2022})}\BibitemShut {NoStop}%
\bibitem [{\citenamefont {Chen}\ \emph {et~al.}(2024)\citenamefont {Chen},
  \citenamefont {Yan},\ and\ \citenamefont {Zhou}}]{Chen2024HypergraphSRE}%
  \BibitemOpen
  \bibfield  {author} {\bibinfo {author} {\bibfnamefont {J.}~\bibnamefont
  {Chen}}, \bibinfo {author} {\bibfnamefont {Y.}~\bibnamefont {Yan}},\ and\
  \bibinfo {author} {\bibfnamefont {Y.}~\bibnamefont {Zhou}},\ }\bibfield
  {title} {\bibinfo {title} {{Magic of quantum hypergraph states}},\ }\href
  {https://doi.org/10.22331/q-2024-05-21-1351} {\bibfield  {journal} {\bibinfo
  {journal} {{Quantum}}\ }\textbf {\bibinfo {volume} {8}},\ \bibinfo {pages}
  {1351} (\bibinfo {year} {2024})}\BibitemShut {NoStop}%
\bibitem [{\citenamefont {Leone}\ \emph {et~al.}(2023)\citenamefont {Leone},
  \citenamefont {Oliviero},\ and\ \citenamefont
  {Hamma}}]{Leone2023NonstabilizerHard}%
  \BibitemOpen
  \bibfield  {author} {\bibinfo {author} {\bibfnamefont {L.}~\bibnamefont
  {Leone}}, \bibinfo {author} {\bibfnamefont {S.~F.~E.}\ \bibnamefont
  {Oliviero}},\ and\ \bibinfo {author} {\bibfnamefont {A.}~\bibnamefont
  {Hamma}},\ }\bibfield  {title} {\bibinfo {title} {{Nonstabilizerness
  determining the hardness of direct fidelity estimation}},\ }\href
  {https://doi.org/10.1103/PhysRevA.107.022429} {\bibfield  {journal} {\bibinfo
   {journal} {Phys. Rev. A}\ }\textbf {\bibinfo {volume} {107}},\ \bibinfo
  {pages} {022429} (\bibinfo {year} {2023})}\BibitemShut {NoStop}%
\bibitem [{\citenamefont {Hayashi}\ and\ \citenamefont
  {Morimae}(2015)}]{Hayashi2015GraphVerify}%
  \BibitemOpen
  \bibfield  {author} {\bibinfo {author} {\bibfnamefont {M.}~\bibnamefont
  {Hayashi}}\ and\ \bibinfo {author} {\bibfnamefont {T.}~\bibnamefont
  {Morimae}},\ }\bibfield  {title} {\bibinfo {title} {{Verifiable
  Measurement-Only Blind Quantum Computing with Stabilizer Testing}},\ }\href
  {https://doi.org/10.1103/PhysRevLett.115.220502} {\bibfield  {journal}
  {\bibinfo  {journal} {Phys. Rev. Lett.}\ }\textbf {\bibinfo {volume} {115}},\
  \bibinfo {pages} {220502} (\bibinfo {year} {2015})}\BibitemShut {NoStop}%
\bibitem [{\citenamefont {Pallister}\ \emph {et~al.}(2018)\citenamefont
  {Pallister}, \citenamefont {Linden},\ and\ \citenamefont
  {Montanaro}}]{Pallister2018Verification}%
  \BibitemOpen
  \bibfield  {author} {\bibinfo {author} {\bibfnamefont {S.}~\bibnamefont
  {Pallister}}, \bibinfo {author} {\bibfnamefont {N.}~\bibnamefont {Linden}},\
  and\ \bibinfo {author} {\bibfnamefont {A.}~\bibnamefont {Montanaro}},\
  }\bibfield  {title} {\bibinfo {title} {Optimal verification of entangled
  states with local measurements},\ }\href
  {https://doi.org/10.1103/PhysRevLett.120.170502} {\bibfield  {journal}
  {\bibinfo  {journal} {Phys. Rev. Lett.}\ }\textbf {\bibinfo {volume} {120}},\
  \bibinfo {pages} {170502} (\bibinfo {year} {2018})}\BibitemShut {NoStop}%
\bibitem [{\citenamefont {Zhu}\ and\ \citenamefont
  {Hayashi}(2019)}]{Zhu2019Verification}%
  \BibitemOpen
  \bibfield  {author} {\bibinfo {author} {\bibfnamefont {H.}~\bibnamefont
  {Zhu}}\ and\ \bibinfo {author} {\bibfnamefont {M.}~\bibnamefont {Hayashi}},\
  }\bibfield  {title} {\bibinfo {title} {General framework for verifying pure
  quantum states in the adversarial scenario},\ }\href
  {https://doi.org/10.1103/PhysRevA.100.062335} {\bibfield  {journal} {\bibinfo
   {journal} {Phys. Rev. A}\ }\textbf {\bibinfo {volume} {100}},\ \bibinfo
  {pages} {062335} (\bibinfo {year} {2019})}\BibitemShut {NoStop}%
\bibitem [{\citenamefont {Dangniam}\ \emph {et~al.}(2020)\citenamefont
  {Dangniam}, \citenamefont {Han},\ and\ \citenamefont
  {Zhu}}]{Zhu2020StabVerify}%
  \BibitemOpen
  \bibfield  {author} {\bibinfo {author} {\bibfnamefont {N.}~\bibnamefont
  {Dangniam}}, \bibinfo {author} {\bibfnamefont {Y.-G.}\ \bibnamefont {Han}},\
  and\ \bibinfo {author} {\bibfnamefont {H.}~\bibnamefont {Zhu}},\ }\bibfield
  {title} {\bibinfo {title} {{Optimal verification of stabilizer states}},\
  }\href {https://doi.org/10.1103/PhysRevResearch.2.043323} {\bibfield
  {journal} {\bibinfo  {journal} {Phys. Rev. Res.}\ }\textbf {\bibinfo {volume}
  {2}},\ \bibinfo {pages} {043323} (\bibinfo {year} {2020})}\BibitemShut
  {NoStop}%
\bibitem [{\citenamefont {Odavi{\'c}}\ \emph {et~al.}(2023)\citenamefont
  {Odavi{\'c}}, \citenamefont {Haug}, \citenamefont {Torre}, \citenamefont
  {Hamma}, \citenamefont {Franchini},\ and\ \citenamefont
  {Giampaolo}}]{Jovan2023WSRE}%
  \BibitemOpen
  \bibfield  {author} {\bibinfo {author} {\bibfnamefont {J.}~\bibnamefont
  {Odavi{\'c}}}, \bibinfo {author} {\bibfnamefont {T.}~\bibnamefont {Haug}},
  \bibinfo {author} {\bibfnamefont {G.}~\bibnamefont {Torre}}, \bibinfo
  {author} {\bibfnamefont {A.}~\bibnamefont {Hamma}}, \bibinfo {author}
  {\bibfnamefont {F.}~\bibnamefont {Franchini}},\ and\ \bibinfo {author}
  {\bibfnamefont {S.~M.}\ \bibnamefont {Giampaolo}},\ }\bibfield  {title}
  {\bibinfo {title} {{{Complexity of frustration: A new source of non-local
  non-stabilizerness}}},\ }\href
  {https://doi.org/10.21468/SciPostPhys.15.4.131} {\bibfield  {journal}
  {\bibinfo  {journal} {SciPost Phys.}\ }\textbf {\bibinfo {volume} {15}},\
  \bibinfo {pages} {131} (\bibinfo {year} {2023})}\BibitemShut {NoStop}%
\bibitem [{\citenamefont {Catalano}\ \emph {et~al.}(2024)\citenamefont
  {Catalano}, \citenamefont {Odavi{\'c}}, \citenamefont {Torre}, \citenamefont
  {Hamma}, \citenamefont {Franchini},\ and\ \citenamefont
  {Giampaolo}}]{Catalano2024GeneralW}%
  \BibitemOpen
  \bibfield  {author} {\bibinfo {author} {\bibfnamefont {A.}~\bibnamefont
  {Catalano}}, \bibinfo {author} {\bibfnamefont {J.}~\bibnamefont
  {Odavi{\'c}}}, \bibinfo {author} {\bibfnamefont {G.}~\bibnamefont {Torre}},
  \bibinfo {author} {\bibfnamefont {A.}~\bibnamefont {Hamma}}, \bibinfo
  {author} {\bibfnamefont {F.}~\bibnamefont {Franchini}},\ and\ \bibinfo
  {author} {\bibfnamefont {S.}~\bibnamefont {Giampaolo}},\ }\bibfield  {title}
  {\bibinfo {title} {{Magic phase transition and non-local complexity in
  generalized $ W $ State}},\ }\href {10.48550/arXiv.2406.19457} {\bibfield
  {journal} {\bibinfo  {journal} {arXiv:2406.19457}\ } (\bibinfo {year}
  {2024})}\BibitemShut {NoStop}%
\bibitem [{\citenamefont {Yi}\ \emph {et~al.}(2024)\citenamefont {Yi},
  \citenamefont {Ye}, \citenamefont {Gottesman},\ and\ \citenamefont
  {Liu}}]{Yi2023GeneralW}%
  \BibitemOpen
  \bibfield  {author} {\bibinfo {author} {\bibfnamefont {J.}~\bibnamefont
  {Yi}}, \bibinfo {author} {\bibfnamefont {W.}~\bibnamefont {Ye}}, \bibinfo
  {author} {\bibfnamefont {D.}~\bibnamefont {Gottesman}},\ and\ \bibinfo
  {author} {\bibfnamefont {Z.-W.}\ \bibnamefont {Liu}},\ }\bibfield  {title}
  {\bibinfo {title} {{Complexity and order in approximate quantum
  error-correcting codes}},\ }\bibfield  {journal} {\bibinfo  {journal} {Nat.
  Phys.}\ }\href {https://doi.org/https://doi.org/10.1038/s41567-024-02621-x}
  {https://doi.org/10.1038/s41567-024-02621-x} (\bibinfo {year}
  {2024})\BibitemShut {NoStop}%
\bibitem [{\citenamefont {Anshu}\ \emph {et~al.}(2022)\citenamefont {Anshu},
  \citenamefont {Landau},\ and\ \citenamefont {Liu}}]{Anshu2022CrossPlatform}%
  \BibitemOpen
  \bibfield  {author} {\bibinfo {author} {\bibfnamefont {A.}~\bibnamefont
  {Anshu}}, \bibinfo {author} {\bibfnamefont {Z.}~\bibnamefont {Landau}},\ and\
  \bibinfo {author} {\bibfnamefont {Y.}~\bibnamefont {Liu}},\ }\bibfield
  {title} {\bibinfo {title} {{Distributed quantum inner product estimation}},\
  }in\ \href {https://doi.org/10.1145/3519935.3519974} {\emph {\bibinfo
  {booktitle} {Proceedings of the 54th Annual ACM SIGACT Symposium on Theory of
  Computing}}}\ (\bibinfo {year} {2022})\ pp.\ \bibinfo {pages}
  {44--51}\BibitemShut {NoStop}%
\bibitem [{\citenamefont {Roe}\ and\ \citenamefont
  {Wallach}(2009)}]{Goodman2009SchurWeyl}%
  \BibitemOpen
  \bibfield  {author} {\bibinfo {author} {\bibfnamefont {G.}~\bibnamefont
  {Roe}}\ and\ \bibinfo {author} {\bibfnamefont {N.~R.}\ \bibnamefont
  {Wallach}},\ }\href {https://doi.org/10.1007/978-0-387-79852-3} {\emph
  {\bibinfo {title} {Symmetry, Representations, and Invariants}}},\ Vol.\
  \bibinfo {volume} {255}\ (\bibinfo  {publisher} {Springer New York},\
  \bibinfo {year} {2009})\BibitemShut {NoStop}%
\bibitem [{\citenamefont {Etingof}\ \emph {et~al.}(2009)\citenamefont
  {Etingof}, \citenamefont {Golberg}, \citenamefont {Hensel}, \citenamefont
  {Liu}, \citenamefont {Schwendner}, \citenamefont {Vaintrob},\ and\
  \citenamefont {Yudovina}}]{Pavel2009SchurWeyl}%
  \BibitemOpen
  \bibfield  {author} {\bibinfo {author} {\bibfnamefont {P.}~\bibnamefont
  {Etingof}}, \bibinfo {author} {\bibfnamefont {O.}~\bibnamefont {Golberg}},
  \bibinfo {author} {\bibfnamefont {S.}~\bibnamefont {Hensel}}, \bibinfo
  {author} {\bibfnamefont {T.}~\bibnamefont {Liu}}, \bibinfo {author}
  {\bibfnamefont {A.}~\bibnamefont {Schwendner}}, \bibinfo {author}
  {\bibfnamefont {D.}~\bibnamefont {Vaintrob}},\ and\ \bibinfo {author}
  {\bibfnamefont {E.}~\bibnamefont {Yudovina}},\ }\bibfield  {title} {\bibinfo
  {title} {{Introduction to Representation Theory}},\ }\href
  {https://arxiv.org/abs/0901.0827} {\bibfield  {journal} {\bibinfo  {journal}
  {arXiv:0901.0827}\ } (\bibinfo {year} {2009})}\BibitemShut {NoStop}%
\end{thebibliography}%

\let\addcontentsline\oldaddcontentsline

\onecolumngrid

\newpage
\appendix

\tableofcontents

\bigskip

In this appendix, we prove the results presented in the main text, including \Thsref{thm:AverageShadow}-\ref{thm:VartUk},  \pref{pro:VarNoiseShadow}, and \coref{cor:VarCl}. In addition, we provide some additional results on the characteristic functions, cross characteristic functions, and  the variances $V(O,\rho)$, $V_*(O,\rho)$. To this end, we clarify the basic properties of the cross moment operator, which is of interest beyond the focus of this work, including shadow estimation based on common randomized measurements and cross-platform verification. Throughout this appendix, we assume that $\caH$ is an $n$-qubit Hilbert space of dimension $d=2^n$ and $O\in\caL_0^{\rmH}(\caH)$  as in the main text. In addition, we assume that $\caU$ is a unitary 2-design.

\section{Proofs of \thref{thm:AverageShadow} and \pref{pro:VarNoiseShadow}}

By assumption  $\caU$ is a unitary 2-design on $\caH$  and $O\in \caL_0^{\rmH}(\caH)$, so   $\caM^{-1}(O) = (d+1)O-\tr(O)\bbone=(d+1)O$ and  the variance $V(O,\rho)$ in shadow estimation based on $\caU$ 
can be expressed as follows \cite{Huang2020Shadow}, 
\begin{equation}\label{eq:VarShadow}
V(O,\rho) =  \bbE_{U\sim\caU} \sum_{\bfb} (d+1)^2\tr\left[\Tensor{\bigl(U^\dag |\bfb\>\<\bfb|U\bigr)}{3} \bigl(\rho\otimes\Tensor{O}{2}\bigr)\right]-[\tr(O\rho)]^2.
\end{equation}
In addition, the formula for 
$V_*(O,\rho)$ in \eref{eq:DefV*} can be simplified as 
\begin{equation}\label{eq:DefV*2Design}
V_*(O,\rho) = (d+1)^2\tr\bigl[\Omega(\caU)(O\otimes\rho)^{\otimes 2}\bigr]-[\tr(O\rho)]^2.
\end{equation}

\subsection{Two auxiliary lemmas}
Before proving \thref{thm:AverageShadow}, we need to introduce two auxiliary lemmas. Let $\SWAP$ be the swap operator on $\caH^{\otimes 2}$.

\begin{lemma}\label{lem:UnitaryAvg}
	Suppose $\tcaU$ is a unitary 2-design on  $\caH$ and $A,B\in\caL(\caH)$. Then 
	\begin{gather}
	\bbE_{U\sim\tcaU}U A U^\dag=\frac{\tr(A)}{d}\bbone,  \label{eq:TwirlAvg} \\
	\bbE_{U\sim\tcaU} \left[\left(U A U^\dag\right)\otimes \left(U A U^\dag\right)^\dag\right] = \frac{d|\tr A|^2 -\|A\|_2^2}{d(d^2-1)}\Tensor{\bbone}{2} + \frac{d\|A\|_2^2-|\tr A|^2}{d(d^2-1)}\SWAP,\label{eq:TensorAvg}\\
	\bbE_{U\sim\tcaU} \left|\tr(UAU^\dag B)\right|^2=\frac{d|\tr A|^2|\tr B|^2+d\|A\|_2^2\|B\|_2^2-|\tr A|^2\|B\|_2^2-\|A\|_2^2|\tr B|^2}{d(d^2-1)}. \label{eq:TraceAvg}
	\end{gather}
\end{lemma}

\begin{lemma}\label{lem:TrOrhoAvg}
	Suppose $\tcaU$ is a unitary 2-design on  $\caH$, $\rho\in\caD(\caH)$, and $O\in\caL_0^{\rmH}(\caH)$. Then
	\begin{equation}\label{eq:TraceAvg2}
	\bbE_{U\sim\tcaU}\; \tr\bigl(UOU^\dag \rho\bigr)^2 =\bbE_{U\sim\tcaU}\; \tr \bigl(O U\rho U^\dag\bigr)^2 = \frac{d\wp(\rho)-1}{d(d^2-1)}\|O\|_2^2,
	\end{equation}
	where  $\wp(\rho) = \tr(\rho^2)$ is the purity of $\rho$.
\end{lemma}
\Eqsref{eq:TwirlAvg}{eq:TensorAvg} in \lref{lem:UnitaryAvg} follow from Schur-Weyl duality and the assumption that $\tcaU$ is a unitary 2-design.  \Eref{eq:TraceAvg}  is a simple corollary of \eref{eq:TensorAvg}. \Eref{eq:TraceAvg2} in \lref{lem:TrOrhoAvg} is in turn a simple corollary of \eref{eq:TraceAvg}.

\subsection{Main proofs}

\begin{proof}[Proof of \thref{thm:AverageShadow}]
	By virtue of \eref{eq:VarShadow} and \lsref{lem:UnitaryAvg}, \ref{lem:TrOrhoAvg} we can deduce that
	\begin{align}
	V(O,\caE(\rho)) &=\bbE_{W\sim \haar}\left\{ \bbE_{U\sim\caU}  \sum_{\bfb} (d+1)^2\tr\left[\Tensor{\bigl(U^\dag |\bfb\>\<\bfb|U\bigr)}{3} \bigl(W\rho W^\dag\bigr)\otimes\Tensor{O}{2}\right] - \bigl[\tr\bigl(OW\rho W^\dag\bigl)\bigr]^2\right\} \nonumber\\	
	&= \bbE_{U\sim\caU} \sum_{\bfb} \frac{(d+1)^2}{d}\tr\left[\Tensor{\bigl(U^\dag |\bfb\>\<\bfb|U\bigr)}{2}\Tensor{O}{2}\right]-\frac{d\tr(\rho^2)-1}{d(d^2-1)}\|O\|_2^2 = \left[1+\frac{d-\tr(\rho^2)}{d^2-1}\right]\|O\|_2^2,\\	
	V(\caE(O),\rho) &=\bbE_{W\sim \haar}\left\{ \bbE_{U\sim\caU}  \sum_{\bfb} (d+1)^2 \tr\left[\Tensor{\bigl(U^\dag |\bfb\>\<\bfb|U\bigr)}{3} \rho\otimes\Tensor{\bigl(WOW^\dag\bigr)}{2}\right] - \bigl[\tr\bigl(WOW^\dag\rho \bigr)\bigr]^2\right\} \nonumber\\
	&= \bbE_{U\sim\caU}  \sum_{\bfb} \frac{d+1}{d}\|O\|_2^2 \tr\left(U^\dag |\bfb\>\<\bfb|U \rho\right) -\frac{d\tr(\rho^2)-1}{d(d^2-1)}\|O\|_2^2 = \left[1+\frac{d-\tr(\rho^2)}{d^2-1}\right]\|O\|_2^2,
	\end{align}
	which confirm \eref{eq:VOrhoAvg}.

	Next, we consider the average of $V_*(O,\rho)$. By virtue of \eqsref{eq:OmegaU}{eq:DefV*2Design} we can deduce that
	\begin{align}
	V_*(O,\caE(\rho)) &=\bbE_{W\sim \haar}\left\{ (d+1)^2\tr\left[\Omega(\caU)\left(O\otimes W\rho W^\dag\otimes O \otimes W\rho W^\dag\right)\right]-\bigl[\tr\bigl(OW\rho W^\dag\bigl)\bigr]^2\right\}\nonumber\\
	&=\frac{(d+1)[d-\tr(\rho^2)]}{d(d-1)} \bbE_{U\sim\caU}  \sum_{\bfa,\bfb}\tr\left\{\left[\dagtensor{U}{2}\left(|\bfa\>\<\bfa|\otimes|\bfb\>\<\bfb|\right)\Tensor{U}{2}\right] \Tensor{O}{2} \right\}\nonumber\\
	&\quad + \frac{(d+1)[d\tr(\rho^2)-1]}{d(d-1)} \bbE_{U\sim\caU}  \sum_{\bfb}\tr\left[\Tensor{\bigl(U^\dag|\bfb\>\<\bfb|U\bigr)}{2} \Tensor{O}{2} \right] -\frac{d\tr(\rho^2)-1}{d(d^2-1)}\|O\|_2^2\nonumber\\
	&= \frac{(d+1)[d-\tr(\rho^2)]}{d(d-1)}[\tr (O)]^2+ \frac{(d+1)[d\tr(\rho^2)-1]}{d(d-1)}\times \frac{d}{d(d+1)}\|O\|_2^2-\frac{d\tr(\rho^2)-1}{d(d^2-1)}\|O\|_2^2\nonumber\\
	&=\frac{d\tr(\rho^2)-1}{d^2-1}\|O\|_2^2.
	\end{align}
	Here the second equality follows from \lsref{lem:UnitaryAvg} and \ref{lem:TrOrhoAvg}, the third equality follows from \lref{lem:TrOrhoAvg} and the fact that $\sum_\bfa |\bfa\>\<\bfa|=\sum_\bfb|\bfb\>\<\bfb|=\bbone$, and the last equality holds because $O$ is traceless by assumption. 
	By a similar reasoning we can deduce that
	\begin{align}
	V_*(\caE(O),\rho) &= \bbE_{W\sim \haar}\left\{ (d+1)^2\tr\left[\Omega(\caU) \bigl(WOW^\dag\otimes \rho\otimes WOW^\dag \otimes \rho\bigr)\right]-\bigl[\tr\bigl(WOW^\dag\rho \bigr)\bigr]^2\right\}\nonumber\\
	&=-\frac{(d+1)\|O\|_2^2}{d(d-1)} \bbE_{U\sim\caU}  \sum_{\bfa,\bfb}\tr\left\{\left[\dagtensor{U}{2}\left(|\bfa\>\<\bfa|\otimes|\bfb\>\<\bfb|\right)\Tensor{U}{2}\right] \Tensor{\rho}{2} \right\}\nonumber\\
	&\quad + \frac{(d+1)\|O\|_2^2}{d-1} \bbE_{U\sim\caU}  \sum_{\bfb}\tr\left[\Tensor{\bigl(U^\dag|\bfb\>\<\bfb|U\bigr)}{2} \Tensor{\rho}{2} \right] -\frac{d\tr(\rho^2)-1}{d(d^2-1)}\|O\|_2^2\nonumber\\
	&=-\frac{(d+1)\|O\|_2^2}{d(d-1)}
	+\frac{(d+1)\|O\|_2^2}{d-1}\times\frac{d[\tr(\rho^2)+1]}{d(d+1)}-\frac{d\tr(\rho^2)-1}{d(d^2-1)}\|O\|_2^2 = \frac{d\tr(\rho^2)-1}{d^2-1}\|O\|_2^2.
	\end{align}
	The above two equations together confirm  \eref{eq:V*OrhoAvg} and complete the proof of \thref{thm:AverageShadow}. 
\end{proof}

\begin{proof}[Proof of \pref{pro:VarNoiseShadow}]
Since $\caU$ is a unitary 3-design by assumption, by virtue of \eref{eq:VarShadow}  we can deduce that

\begin{equation}\label{eq:VarShadow3design}
V(O,\rho) = \frac{6(d+1)}{d+2}\tr\bigl[P_{[3]} \bigl(\rho\otimes\Tensor{O}{2}\bigr)\bigr]-[\tr(O\rho)]^2= \frac{d+1}{d+2}\left[\tr(O^2)+2\tr(\rho O^2)\right]-\left[\tr(\rho O)\right]^2\leq 3\|O\|_2^2,
\end{equation}
where $P_{[3]}$ is the projector onto the symmetric subspace in $\caH^{\otimes 3}$, and the inequality  holds because $\tr(\rho O^2) \le \tr(O^2)=\|O\|_2^2$. This equation confirms the upper bound in \eref{eq:VarShadowBound}, which was originally proved in \rcite{Huang2020Shadow}. The lower bound of $V(O,\rho)$ in \eref{eq:VarShadowBound} holds because $\tr(\rho O^2)\ge0$ and
\begin{align}
\left[\tr(\rho O)\right]^2=\left[\tr(\sqrt{\rho}\sqrt{\rho} O)\right]^2\leq  \tr(\rho) \tr (\rho O^2)= \tr (\rho O^2),
\end{align}
where we have applied the Cauchy-Schwarz inequality.

When $O = |\phi\>\<\phi|-\bbone/d$, \eref{eq:VarFShadow} follows from \eref{eq:VarShadow3design},  which completes the proof of \pref{pro:VarNoiseShadow}.
\end{proof}

\section{Pauli group and Clifford group}\label{app:Clifford}

The $n$-qubit  \emph{Pauli group} $\caP_n$ is the group generated by all $n$-fold tensor products of Pauli operators, that is
\begin{equation}
\caP_n:=\langle \{I, X, Y, Z\}^{\otimes n}\rangle =\bigl\{\rmi^k Z_{\bfp} X_{\bfq} \;| \;k\in\bbZ_4,\bfp,\bfq\in\bbZ_2^n\bigr\},
\end{equation}
where 
\begin{align}\label{eq:ZpXq}
Z_{\bfp}:=Z^{p_1}\otimes Z^{p_2}\otimes\cdots\otimes Z^{p_n},\quad  X_{\bfq}:= X^{q_1}\otimes X^{q_2}\otimes \cdots \otimes X^{q_n}. 
\end{align}
The projective Pauli group $\bcaP_n$ is the quotient of $\caP_n$ over the phase factors 
and can be identified with the set $\{I, X, Y,Z\}^{\otimes n}$. The \emph{Clifford group} is  the unitary normalizer of the Pauli group, that is,
\begin{equation}
\Cl_n = \{U\in\rmU(d)\;|\;U\caP_n U^\dag = \caP_n\},
\end{equation}
where $\rmU(d)$ is the unitary group acting on $\caH$. The projective Clifford group $\overline{\Cl}_n$ is the quotient  of $\Cl_n$ over the phase factors. It is known that $\overline{\Cl}_n/\bcaP_n$ is isomorphic to the symplectic group $\mathrm{Sp}(2n, \bbF_2)$, where $\bbF_2\simeq \bbZ_2$ is the  finite field composed of two elements. Up to overall phase factors,  the Clifford group $\Cl_n$ can be generated by  the Hadamard gate $H$, phase gate $S$, and  controlled-not gate $\mathrm{CNOT}$, where 
\begin{equation}
H = \frac{1}{\sqrt{2}}\begin{pmatrix}
1 & 1\\
1 & -1
\end{pmatrix},
\quad S = \begin{pmatrix}
1 & 0\\
0 & \rmi
\end{pmatrix},
\quad \mathrm{CNOT} = \begin{pmatrix}
1 & 0 & 0 & 0 \\
0 & 1 & 0 & 0 \\
0 & 0 & 0 & 1 \\
0 & 0 & 1 & 0
\end{pmatrix}.
\end{equation}

\section{Properties of characteristic and cross characteristic functions}
Characteristic functions \cite{Gross2006CharFunction,Dai2022CharFunction} and cross characteristic functions play crucial roles in understanding the performance of thrifty shadow estimation, as demonstrated in \thsref{thm:VarCl}, \ref{thm:VarUkl}, \ref{thm:VartUk} and their corollaries. Here we clarify the basic properties of these functions, in preparation for the proofs of our main results.

Let $O,O_1, O_2$ be three linear operators on $\caH$. Recall that the 
characteristic function $\Xi_O$ and cross characteristic functions $\Xi_{O_1, O_2}$, $\tXi_{O_1, O_2}$ are defined as follows:
\begin{align}\label{eq:CharCrossChar}
\Xi_O(P) := \tr(OP),\;\; \Xi_{O_1, O_2}(P) := \tr(O_1P)\tr(O_2 P)=\Xi_{O_1}(P)\Xi_{O_2}(P),\;\; \tXi_{O_1, O_2}(P) := \tr(O_1 P O_2 P),\;\; P\in \bcaP_n.
\end{align}
They can be regarded as vectors with $d^2$ entries. 
 Since the set of  Pauli operators in  $\bcaP_n$ forms an orthogonal unitary basis in  $\caL(\caH)$, by definition it is straightforward to derive the following relations: 
\begin{equation}\label{eq:CharOrelation}
O = \frac{1}{d} \sum_{P\in\bcaP_n} \Xi_{O}(P) P,\quad 
\tr(O^2)=\frac{1}{d}\|\Xi_{O}\|_2^2,\quad \tr(O_1 O_2)=\frac{1}{d}\sum_{P\in\bcaP_n}\Xi_{O_1, O_2}(P).
\end{equation}
Additional properties of characteristic functions and cross characteristic functions are summarized in \lsref{lem:CrossChar}-\ref{lem:CharAvg} below, which are proved in \aref{app:CharProofs}.

\begin{lemma}\label{lem:CrossChar}
Suppose  $O_1,O_2\in \caL^{\rmH}(\caH)$. Then
\begin{equation}\label{eq:CrossCharNorm}
|\tXi_{O_1, O_2}\cdot\Xi_{O_1, O_2}|\leq \|\tXi_{O_1, O_2}\|_2^2 =  \|\Xi_{O_1, O_2}\|_2^2.
\end{equation}
\end{lemma}

\begin{lemma}\label{lem:CharOrhoNorm}
	Suppose $\rho\in\caD(\caH)$ and $O\in\caL_0^{\rmH}(\caH)$. Then
	\begin{equation}\label{eq:CharOrhoNorm}
|\tXi_{\rho, O}\cdot\Xi_{\rho, O}|\leq \|\tXi_{\rho, O}\|_2^2=	\|\Xi_{\rho, O}\|_2^2 \le \|\Xi_O^2\|_{[d]}\le\|\Xi_O\|_2^2= d\|O\|_2^2. 
	\end{equation}
\end{lemma}

\begin{lemma}\label{lem:CharOrhoNormF}
    Suppose $\rho\in\caD(\caH)$ and $O=|\phi\>\<\phi|-\bbone/d$ with $|\phi\>\in\caH$. Then
    \begin{align}
    \tXi_{\rho, O}\cdot\Xi_{\rho, O} &= \tXi_{\rho, \phi}\cdot\Xi_{\rho, \phi} - 2F+\frac{1}{d}< \|\Xi_{\rho, O}\|_2^2=\|\Xi_{\rho, \phi}\|_2^2 - 1< \|\Xi_\phi^2\|_{[d]},\label{eq:CharOrhoNormF}\\
    \tXi_{\phi, \phi}\cdot\Xi_{\phi, \phi} &= \|\Xi_{\phi, \phi}\|_2^2 = 2^{-M_2(\phi)}d, \label{eq:CharOrhoNormFIdeal}
    \end{align} 
    where $F=\<\phi|\rho|\phi\>$ is the fidelity between $\rho$ and $|\phi\>$ and $M_2(\phi)$ is the 2-SRE of $|\phi\>$.
\end{lemma}
Here   $ \Xi_\phi$, $\Xi_{\rho,\phi}$, and  $\Xi_{\phi,\phi}$ are the abbreviations of $\Xi_{|\phi\>\<\phi|}$, $\Xi_{\rho, |\phi\>\<\phi|}$, and $\Xi_{ |\phi\>\<\phi|, |\phi\>\<\phi|}$, respectively; a similar convention applies to $\tXi_{\rho,\phi}$, and  $\tXi_{\phi,\phi}$.

\begin{figure}[t]
	\centering
	\includegraphics[width=0.4\textwidth]{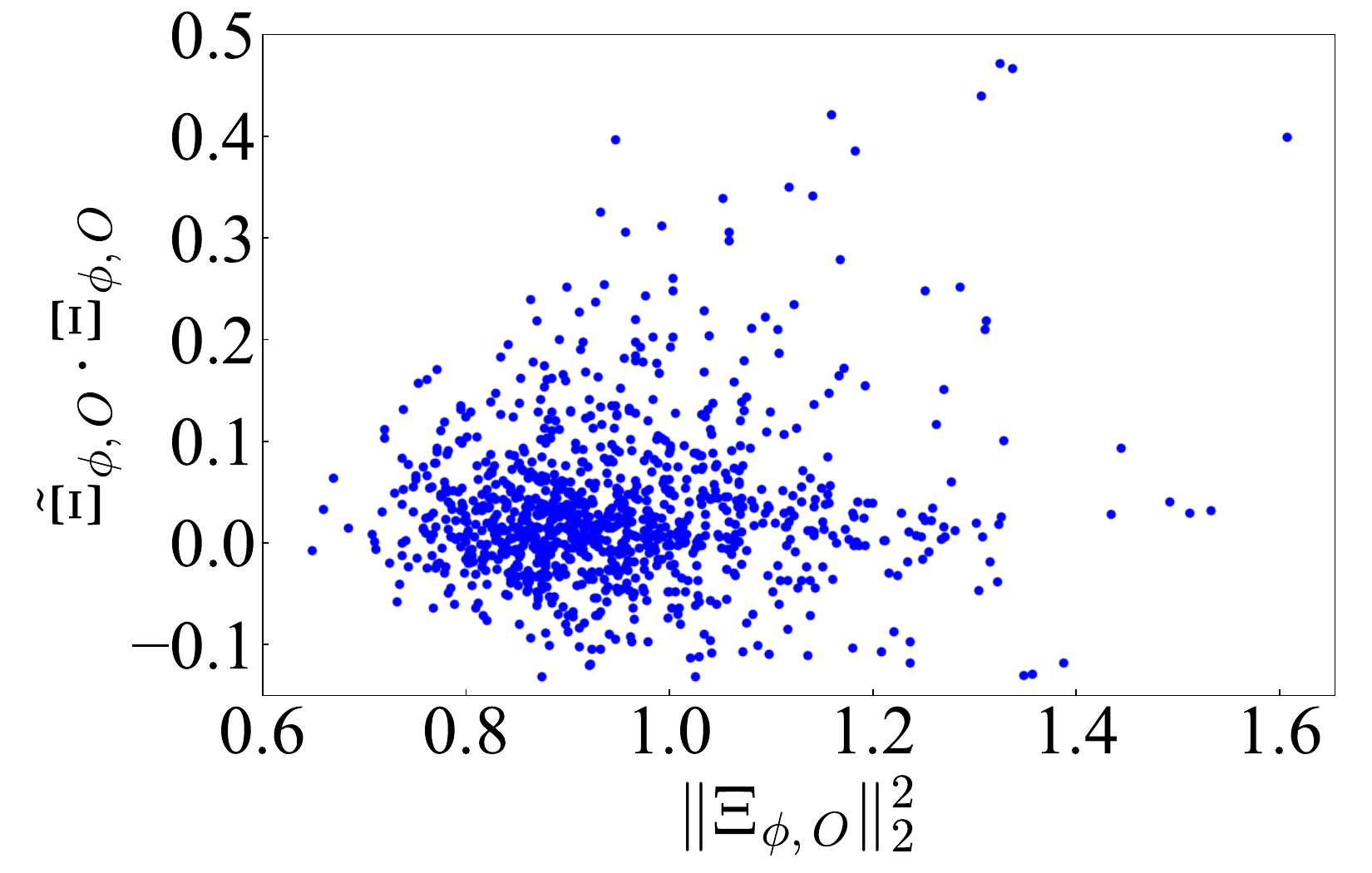}
	\caption{A scatter plot about the relation between $\|\Xi_{\phi, O}\|_2^2$ and $\tXi_{\phi, O}\cdot\Xi_{\phi, O}$ for a five-qubit quantum system. Here $\phi$ is a fixed random pure state, and $O$ is sampled (1000 times) from a unitary-invariant ensemble generated from a 
fixed Hermitian operator that is normalized with respect to the Hilbert-Schmidt norm. 
		}
	\label{fig:ratio}
\end{figure}

\begin{lemma}\label{lem:CharAvg}
Suppose $\rho\in\caD(\caH)$ and $O\in\caL_0^{\rmH}(\caH)$. Then 
\begin{gather}
\bbE_{\rho'\sim\caE(\rho)} [\tr(\rho' O)]^2=\bbE_{O'\sim\caE(O)} [\tr(\rho O')]^2=\frac{d\wp(\rho)-1}{d(d^2-1)}\|O\|_2^2, \label{eq:rhoOtrAvg}\\
\bbE_{\rho'\sim\caE(\rho)} \|\Xi_{\rho', O}\|_2^2 =\bbE_{O'\sim\caE(O)} \|\Xi_{\rho, O'}\|_2^2= \frac{d^2\wp(\rho)-d}{d^2-1}\|O\|_2^2,  \label{eq:CharNormAvg} \\
\bbE_{\rho'\sim\caE(\rho)} \; \tXi_{\rho', O}\cdot\Xi_{\rho', O} =\bbE_{O'\sim\caE(O)} \; \tXi_{\rho, O'}\cdot\Xi_{\rho, O'} = \frac{d\wp(\rho)-1}{d^2-1}\|O\|_2^2. \label{eq:CharProdAvg}
\end{gather}
\end{lemma}

\Lref{lem:CharAvg} shows that on average $\tXi_{O_1, O_2}\cdot\Xi_{O_1, O_2}$ is exponentially smaller than $\|\Xi_{O_1, O_2}\|_2^2$ with respect to  the number $n$ of qubits, as illustrated in  \fref{fig:ratio}.
However, the two terms are comparable in certain special situations that are relevant to fidelity estimation (see \thref{thm:VarFCl}). 
When $|\phi\>$ is a stabilizer state, $\rho = |\phi\>\<\phi|$, and $O = |\phi\>\<\phi| - \bbone/d$ for example,
we have $\|\Xi_{\rho, O}\|_2^2-\tXi_{\rho, O}\cdot\Xi_{\rho, O}=\caO(1)$ and $\|\Xi_{\rho, O}\|_2^2$, $\tXi_{\rho, O}\cdot\Xi_{\rho, O}=\caO(d)$ by \lref{lem:CharOrhoNormF}.

Incidentally, by virtue of \thref{thm:VarCl} and \lref{lem:CharAvg}, it is straightforward to determine $V_*(\caE(O),\rho)$
and provide an alternative proof of \eref{eq:V*OrhoAvg} in \thref{thm:AverageShadow} for the ensemble $\caU=\Cl_n$. 
Note that the averages in Eqs.~\eqref{eq:rhoOtrAvg}-\eqref{eq:CharProdAvg}
are all proportional to  $V_*(\caE(O),\rho)$.  This fact is not a mere coincidence.

\subsection{\label{app:CharProofs}Proofs of \lsref{lem:CrossChar}-\ref{lem:CharAvg}}

\begin{proof}[Proof of \lref{lem:CrossChar}]
	By virtue of \eqsref{eq:CharCrossChar}{eq:CharOrelation}	
	we can deduce that 
	\begin{align}
	\|\tXi_{O_1, O_2}\|_2^2&=\sum_{P\in \bcaP_n} [\tr(O_1 P O_2 P)]^2=\frac{1}{d^2}\sum_{P,P', P''\in \bcaP_n} \Xi_{O_1} (P') \Xi_{O_1}(P'')\tr(P'PO_2P)\tr(P''PO_2P)\nonumber\\
	&=\frac{1}{d^2}\sum_{P', P''\in \bcaP_n} \Xi_{O_1} (P') \Xi_{O_1}(P'')\sum_{P\in \bcaP_n} \tr[(PP'P\otimes PP''P)(O_2\otimes O_2)]
	\nonumber\\
	&=\sum_{P' \in \bcaP_n} \Xi^2_{O_1} (P')[\tr(O_2P')]^2
	=\sum_{P \in \bcaP_n} \Xi^2_{O_1,O_2} (P) = \|\Xi_{O_1, O_2}\|_2^2,
	\end{align}	
	which confirms the equality in \eref{eq:CrossCharNorm}. Here the fourth equality holds because
	\begin{align}
	\sum_{P\in \bcaP_n} PP'P\otimes PP''P=\begin{cases}
	d^2 P'\otimes P' & \mbox{if}\quad  P''=P',\\
	0 &\mbox{otherwise}. 
	\end{cases}
	\end{align}
	
	The inequality in \eref{eq:CrossCharNorm} follows from the equality proved above and  the Cauchy-Schwarz inequality. 
\end{proof}

\begin{proof}[Proof of \lref{lem:CharOrhoNorm}]
	The first inequality and equality in \eref{eq:CharOrhoNorm} follow from \lref{lem:CrossChar}, while the last inequality and equality hold by definition. The second inequality in \eref{eq:CharOrhoNorm} can be proved as follows,
	\begin{equation}
	\|\Xi_{\rho, O}\|_2^2=\sum_{P\in \bcaP_n}\Xi_\rho^2(P) \Xi_O^2(P)  \le \|\Xi_O^2\|_{[d]},
	\end{equation}
	where the inequality holds because  $\Xi_\rho^2(P) \le 1$ for each $P\in\bcaP_n$ and $\|\Xi_\rho^2\|_1=\|\Xi_\rho\|_2^2 = d\wp(\rho) \le d$. 
	This observation completes the proof of \lref{lem:CharOrhoNorm}. 
\end{proof}

\begin{proof}[Proof of \lref{lem:CharOrhoNormF}]
Based on the definition in \eref{eq:CharCrossChar} we can deduce that
\begin{align} 
\|\Xi_{\rho, O}\|_2^2 &= \sum_{P\in\bcaP_n}[\tr(\rho P)]^2\left[\tr(|\phi\>\<\phi|P) -\frac{1}{d}\tr(P) \right]^2=\sum_{P\in\bcaP_n, P\neq \bbone}[\tr(\rho P)]^2\left[\tr(|\phi\>\<\phi|P) \right]^2=\|\Xi_{\rho, \phi}\|_2^2 - 1,\\
\tXi_{\rho, O}\cdot\Xi_{\rho, O} &= \sum_{P\in\bcaP_n} \left[\tr(\rho P |\phi\>\<\phi| P)-\frac{1}{d}\tr(\rho)\right] \tr(\rho P)\left[\tr(|\phi\>\<\phi|P) -\frac{1}{d}\tr(P) \right]\nonumber\\
&=\sum_{P\in\bcaP_n, P\neq \bbone}\tr(\rho P |\phi\>\<\phi| P)\tr(\rho P)\tr(|\phi\>\<\phi|P)-\frac{1}{d}\sum_{P\in\bcaP_n, P\neq \bbone} \tr(\rho P)\tr(|\phi\>\<\phi|P)\nonumber\\
&=\tXi_{\rho, \phi}\cdot\Xi_{\rho, \phi} - 2F+\frac{1}{d},
\end{align}
which imply the first and third equalities in \eref{eq:CharOrhoNormF}. In deriving the above results we have employed the facts that $\tr(P)=d\delta_{P,\bbone}$ and $\sum_{P\in\bcaP_n}\Xi_{\rho,\phi}(P)/d = F$. The second and fourth inequalities follow directly from \lref{lem:CharOrhoNorm}. \Eref{eq:CharOrhoNormFIdeal} is a simple corollary of the definitions in \eqsref{eq:DefM2}{eq:CharCrossChar}.
\end{proof}

\begin{proof}[Proof of \lref{lem:CharAvg}]
\Eref{eq:rhoOtrAvg} is a simple corollary of \lref{lem:TrOrhoAvg}. 

By virtue of \eref{eq:TraceAvg} in \lref{lem:UnitaryAvg} we can deduce that
\begin{align}
\bbE_{\rho'\sim\caE(\rho)} \|\Xi_{\rho', O}\|_2^2
&=\bbE_{\rho'\sim\caE(\rho)}\sum_{P\in\bcaP_n}\left\{[\tr(\rho'P)]^2[\tr(O P)]^2 \right\}\nonumber\\
&=\frac{1}{d(d^2-1)}\sum_{P\in\bcaP_n} 
\left\{\left[d^2\wp(\rho)-d+d(\tr P)^2-\wp(\rho)(\tr P)^2\right][\tr(OP)]^2
\right\}= \frac{d^2\wp(\rho)-d}{d^2-1}\|O\|_2^2,\\
\bbE_{O'\sim\caE(O)} \|\Xi_{\rho, O'}\|_2^2&=\bbE_{O'\sim\caE(O)}\sum_{P\in\bcaP_n}\left\{[\tr(\rho P)]^2 [\tr(O'P)]^2\right\}\nonumber\\
&=\frac{\|O\|_2^2}{d(d^2-1)}\sum_{P\in\bcaP_n} \left\{[d^2-(\tr P)^2][\tr(\rho P)]^2\right\}= \frac{d^2\wp(\rho)-d}{d^2-1}\|O\|_2^2,
\end{align}
which confirm \eref{eq:CharNormAvg}. In deriving the above results we have taken into account the facts that  $\tr(P) = d\delta_{P, \bbone}$ and $\sum_{P\in\bcaP_n}[\tr(\rho P)]^2 =d\tr(\rho^2) =d\wp(\rho)$.

By virtue of \eref{eq:TensorAvg} in \lref{lem:UnitaryAvg} we can  deduce that
\begin{align}
&\bbE_{\rho'\sim\caE(\rho)} \; \tXi_{\rho', O}\cdot\Xi_{\rho', O} =
\bbE_{\rho'\sim\caE(\rho)}\sum_{P\in\bcaP_n} \tr(\rho' P)\tr(OP)\tr(\rho' PO P)=\bbE_{\rho'\sim\caE(\rho)}\sum_{P\in\bcaP_n} \tr(O P)\tr[(\rho'\otimes \rho') (P\otimes PO P)]\nonumber\\
&=\frac{1}{d(d^2-1)}\sum_{P\in\bcaP_n}\left\{\left[d-\wp(\rho) \right]\tr(OP)\tr(P)\tr(POP)+\left[d\wp(\rho)-1\right]
\tr(OP)\tr(P^2OP)\right\}=\frac{d\wp(\rho)-1}{d^2-1}\|O\|_2^2,\\
&\bbE_{O'\sim\caE(O)} \; \tXi_{\rho, O'}\cdot\Xi_{\rho, O'} =\bbE_{O'\sim\caE(O)}\sum_{P\in\bcaP_n} \tr(\rho P)\tr(O'P)\tr(O' P\rho P)=\bbE_{O'\sim\caE(O)}\sum_{P\in\bcaP_n} \tr(\rho P)\tr[(O'\otimes O') (P\otimes P\rho P)]\nonumber\\
& =\frac{\|O\|_2^2}{d(d^2-1)}\sum_{P\in\bcaP_n} 
\left[d\tr(\rho P)\tr(P^2\rho P)-\tr(\rho P)\tr(P)\tr(P\rho P)\right]= \frac{d\wp(\rho)-1}{d^2-1}\|O\|_2^2,
\end{align}
which confirm \eref{eq:CharProdAvg}. 
\end{proof}

\section{Characteristic functions and stabilizer 2-R\'{e}nyi entropies of some pure states}\label{app:Stab2Entropy}
\subsection{Phased W states}
Let  $\theta_1,\theta_2,\ldots,\theta_n$ be $n$ real phases and $\bmtheta=(\theta_1,\theta_2,\ldots,\theta_n)$. Here we shall determine the 2-SRE of the following phased W state,
\begin{equation}
|\rmW_n(\bmtheta)\> = \frac{1}{\sqrt{n}} \sum_{j=1}^n e^{\rmi \theta_j} X_j \Tensor{|0\>}{n}. 
\end{equation}
When $\theta_j=j\theta$ for $j=1,2,\ldots, n$, $|\rmW_n(\bmtheta)\>$ reduces to the generalized W state in \eref{eq:WState}, which was studied in \rscite{Yi2023GeneralW,Catalano2024GeneralW}. If in addition $\theta=0$,  then $|\rmW_n(\bmtheta)\>$ reduces to the W state. 
The 2-SRE of  $|\rmW_n\>$  was  calculated in 
\rcite{Jovan2023WSRE}; the 2-SRE of  $|\rmW_n(\theta)\>$ was  calculated in \rcite{Catalano2024GeneralW} under some assumption. 

Let $P\in \bcaP_n$ be an $n$-qubit Pauli operator. Denote by $n_X(P)$ the number of the  Pauli operator $X$ in the tensor decomposition of  $P$; define  $n_Y(P)$ and  $n_Z(P)$  in a similar way. 
If  $n_X(P)+ n_Y(P)=2$, let  $l_{1,2}(P)$ be the first (second) index for which the tensor factor of $P$ is equal to $X$ or $Y$ and let $l(P)=l_2(P)-l_1(P)$; if $n_X(P)+ n_Y(P)\neq 2$,  let $l(P)=l_1(P)=l_2(P)=0$. Based on these definitions, the characteristic function $\Xi_{\rmW_n(\bmtheta)}(P)$ can be expressed as
\begin{align}
\Xi_{\rmW_n(\bmtheta)}(P) &= \frac{1}{n}\bigl\{\delta_{n_X(P),0}\delta_{n_Y(P),0}[n-2n_Z(P)]+2\delta_{n_X(P)+n_Y(P),2}[\delta_{n_X(P),2}+\delta_{n_Y(P),2}]\cos\left(\theta_{l_2}-\theta_{l_1}\right)\nonumber\\
&\quad  - 2\delta_{n_X(P),1}\delta_{n_Y(P),1}\sin\left(\theta_{l_2}-\theta_{l_1}\right)\bigr\},
\end{align}
where $l_1$ and $l_2$ are shorthands for $l_1(P)$ and $l_2(P)$, respectively.  When $P=\Tensor{Z}{n}$ for example,  $\Xi_{\rmW_n(\bmtheta)}(P)=-1$.  As a simple corollary, we can reproduce the  characteristic functions  
$\Xi_{\rmW_n}(P)$ and $\Xi_{\rmW_n(\theta)}(P)$
determined in Refs.~\cite{Jovan2023WSRE,Catalano2024GeneralW},
\begin{align}
\Xi_{\rmW_n}(P) &= \frac{1}{n}\left\{\delta_{n_X(P),0}\delta_{n_Y(P),0}[n-2n_Z(P)]+2\delta_{n_X(P)+n_Y(P),2}[\delta_{n_X(P),2}+\delta_{n_Y(P),2}]\right\},\label{eq:CharFW}\\
\Xi_{\rmW_n(\theta)}(P) &= \frac{1}{n}\left\{\delta_{n_X(P),0}\delta_{n_Y(P),0}[n-2n_Z(P)]+2\delta_{n_X(P)+n_Y(P),2}[\delta_{n_X(P),2}+\delta_{n_Y(P),2}]\cos(l\theta)\right.\nonumber\\
&\quad \left.-2\delta_{n_X(P),1}\delta_{n_Y(P),1}\sin(l\theta)\right\},\label{eq:CharFGeneralW}
\end{align}
where $l$ is a shorthand for $l(P)$.

\begin{proposition}\label{pro:CharNormPhasedW}
	Suppose $\theta,\theta_1,\theta_2,\ldots,\theta_n$ are $n+1$ real phases with $n\geq 1$ and
	$\bmtheta=(\theta_1,\theta_2,\ldots,\theta_n)$. Then 
	\begin{align}
	\|\Xi_{\rmW_n(\bmtheta)}\|_4^4&=\frac{d}{n^4} \left[6n(n-1)+\Biggl|\sum_{j=1}^n e^{4\rmi\theta_j}\Biggr|^2\right], & M_2(\rmW_n(\bmtheta))&=\log_2\frac{n^4}{\Bigl[6n(n-1)+\Bigl|\sum_{j=1}^n e^{4\rmi\theta_j}\bigr|^2\Bigr]},  \label{eq:CharNormPhasedW}\\
	\|\Xi_{\rmW_n}\|_4^4 &=\frac{d(7n-6)}{n^3},& \quad 
	M_2(\rmW_n) &= \log_2\frac{n^3}{7n-6}, \label{eq:CharNormW}\\
	\|\Xi_{\rmW_n(\theta)}\|_4^4
	&= \frac{d\left[6n^2-6n + \sin^2(2n\theta)/\sin^2(2\theta)\right]}{n^4},& \quad M_2(\rmW_n(\theta))& =\log_2 \frac{n^4}{6n^2-6n + \sin^2(2n\theta)/\sin^2(2\theta)}. \label{eq:CharNormGeneralW}
	\end{align}
\end{proposition}
The 2-SRE of  $|\rmW_n\>$ in \eref{eq:CharNormW}  agrees with the result derived in 
\rcite{Jovan2023WSRE}. \Pref{pro:CharNormPhasedW} also implies that
\begin{equation}\label{eq:gWstateChar4normLUB}
\frac{6d(n-1)}{n^3}\leq \|\Xi_{\rmW_n(\bmtheta)}\|_4^4\leq \frac{d(7n-6)}{n^3},\quad 
\log_2 \frac{n^3}{7n-6} \leq  M_2(\rmW_n(\bmtheta))\leq \log_2 \frac{n^3}{6(n-1)}.
\end{equation}
Here the upper bound for $\|\Xi_{\rmW_n(\bmtheta)}\|_4^4$ [lower bound for  $M_2(\rmW_n(\bmtheta))$] is saturated iff $\theta_1=\theta_2=\cdots =\theta_n$, in which case 
$|\rmW_n(\bmtheta)\> $ reduces to the W state up to an overall phase factor;
the lower bound for $\|\Xi_{\rmW_n(\bmtheta)}\|_4^4$ [upper bound for  $M_2(\rmW_n(\bmtheta))$] is saturated iff $\sum_{j=1}^n e^{4\rmi\theta_j}=0$. According to \eref{eq:gWstateChar4normLUB}, all $n$-qubit phased W states have comparable 2-SREs irrespective of the phases $\theta_1, \theta_2,\ldots, \theta_n$.

\begin{proof}[Proof of \pref{pro:CharNormPhasedW}]For the convenience of the following proof, we need to introduce a stabilizer projector following \rcite{Zhu2016Fail4Design}, which will be discussed further in \aref{app:SWDClifford}. Note that $\{P^{\otimes 4}\}_{P\in \bcaP_n}$ is a stabilizer group. Let $P_n$ be the stabilizer projector onto the corresponding stabilizer code. Let $\scrS=\{0000,0011,0101,0110\}$ and
	let $\tbfu$ be the bitwise "NOT" of $\bfu$ for each $\bfu\in \scrS$. Define
	\begin{align}
	|\varphi_\bfu\>:=\frac{|\bfu\>+|\tbfu\>}{\sqrt{2}},\quad 
	|\varphi(\bfu_1,\bfu_2,\ldots,\bfu_n)\>:=|\varphi_{\bfu_1}\>\otimes |\varphi_{\bfu_2}\> \otimes \cdots \otimes |\varphi_{\bfu_n}\>, \quad \bfu, \bfu_1,\bfu_2, \ldots, \bfu_n\in \scrS. 
	\end{align}
	Then the projectors $P_1$ and $P_n$ can be expressed as follows [cf. \eref{eq:FormRT4} in \aref{app:SWDClifford}],
	\begin{align}
	P_1=\sum_{\bfu\in\scrS } |\varphi_\bfu\>\<\varphi_\bfu|,\quad P_n=P_1^{\otimes n}=\sum_{\bfu_1,\bfu_2,\ldots,\bfu_n\in \scrS}|\varphi(\bfu_1,\bfu_2,\ldots,\bfu_n)\>\< \varphi(\bfu_1,\bfu_2,\ldots,\bfu_n)|.
	\end{align}
	Accordingly, $\|\Xi_{\rmW_n(\bmtheta)}\|_4^4$ can be expressed as follows,
	\begin{align}
	\|\Xi_{\rmW_n(\bmtheta)}\|_4^4 &=\sum_{P\in \bcaP_n} [\tr(P |\rmW_n(\bmtheta)\>\<\rmW_n(\bmtheta)|)]^4= d^2\tr \left[P_n\Tensor{(|\rmW_n(\bmtheta)\>\<\rmW_n(\bmtheta)|)}{4}\right]\nonumber\\
	&=d^2\sum_{\bfu_1,\bfu_2,\ldots, \bfu_n\in \scrS}\left|\<\varphi(\bfu_1,\bfu_2,\ldots,\bfu_n) |\rmW_n(\bmtheta)\>^{\otimes 4}\right|^2.\label{eq:CharNormPhasedWproof}
	\end{align}

	Next, let $\scrW=\{\bfw\in \{0,1\}^n\;:\; |\bfw|=1\}$,
	where $|\bfw|$ denotes the weight of $\bfw$, that is,  the number of 1 in the string. Then $\Tensor{|\rmW_n(\bmtheta)\>}{4}$ can be expressed as follows,
	\begin{equation}
	\Tensor{|\rmW_n(\bmtheta)\>}{4} = \frac{1}{n^2} \sum_{\bfa,\bfb,\bfc,\bfw\in\scrW}e^{\rmi(\bfa+\bfb+\bfc+\bfw)\cdot\bmtheta} |a_1 b_1 c_1 w_1\> \otimes  |a_2 b_2 c_2 w_2\>\otimes \dots \otimes|a_n b_n c_n w_n\>, 
	\end{equation}

	According to the properties of $|\varphi(\bfu_1,\bfu_2,\ldots,\bfu_n)\>$ and $|\rmW_n(\bmtheta)\>$, only the following two cases  contribute to the summation in \eref{eq:CharNormPhasedWproof}: 
	\begin{gather}
	\bfu_1=\bfu_2=\cdots=\bfu_n=0000,\\
	\bfu_j=\bfu_k\neq 0000, \quad j<k, \quad \bfu_i=0000\quad \forall i\neq j,  k.
	\end{gather}
	The first case gives the contribution
	\begin{align}
	d\left|(\<0000|+\<1111|)^{\otimes n}|\rmW_n(\bmtheta)\>^{\otimes 4}\right|^2=\frac{d}{n^4}\Biggl|\sum_{j=1}^n e^{4\rmi\theta_j}\Biggr|^2.
	\end{align}
	The second case has $3n(n-1)/2$ subcases, all of which give the same contribution. When $\bfu_1=\bfu_2=0011$ and $\bfu_3=\bfu_4=\cdots=\bfu_n=0000$ for example, we have
	\begin{align}
	&d^2\left|\<\varphi(\bfu_1,\bfu_2,\ldots,\bfu_n) |\rmW_n(\bmtheta)\>^{\otimes 4}\right|^2
	=d\left|\bigl[(\<0011|+\<1100|)^{\otimes 2}\otimes (\<0000|+\<1111|)^{\otimes (n-2)}\bigr]|
	\rmW_n(\bmtheta)\>^{\otimes 4}\right|^2\nonumber\\
	&=d\left|\bigl[(\<0011|\otimes \<1100|+\<1100|\otimes \<0011|)
	\otimes 
	\<0000|^{\otimes (n-2)}
	\bigr]|
	\rmW_n(\bmtheta)\>^{\otimes 4}\right|^2=\frac{4d}{n^4}. 
	\end{align}
	Combining the above results we can deduce that
	\begin{equation}
	\|\Xi_{\rmW_n(\theta)}\|_4^4=\frac{d}{n^4} \left[6n(n-1)+\Biggl|\sum_{j=1}^n e^{4\rmi\theta_j}\Biggr|^2\right],
	\end{equation}
	which confirms the first equality in \eref{eq:CharNormPhasedW}. The second equality in \eref{eq:CharNormPhasedW} follows from  the first equality and the definition of 2-SRE. \Eqsref{eq:CharNormW}{eq:CharNormGeneralW} are simple corollaries of \eref{eq:CharNormPhasedW}. 
\end{proof}

\subsection{$|S_{n,k}(\theta)\>$}
According to the definition in \eref{eq:MagicState}, the characteristic function of  $|S_{1, 1}(\theta)\>=(|0\>+e^{\rmi\theta}|1\>)/\sqrt{2}$ reads
\begin{equation}
\Xi_{S_{1, 1}(\theta)}(P) =
\begin{cases}
1 &   P = I,\\
\cos\theta & P = X,\\
\sin\theta & P = Y,\\
0  & P = Z,
\end{cases}
\end{equation}
and its  Schatten $p$-norm can be expressed as follows,
\begin{equation}
\|\Xi_{S_{1, 1}(\theta)}\|_p = \left(1+\left|\cos\theta\right|^p + \left|\sin\theta\right|^p\right)^{1/p}.
\end{equation}
By definition  $|S_{n, k}(\theta)\>$ is a tensor product of single-qubit states, so the Schatten $p$-norm of its characteristic function can be calculated as follows,
\begin{equation}
\|\Xi_{S_{n, k}(\theta)}\|_p =\|\Xi_{|0\>}\|_p^{n-k} \|\Xi_{S_{1, 1}(\theta)}\|_p^{k}= 2^{(n-k)/p}\left(1+\left|\cos\theta\right|^p + \left|\sin\theta\right|^p\right)^{k/p}.
\end{equation}
Now, it is straightforward to 
calculate the 2-SRE of $|S_{n,k}(\theta)\>$, with the result
\begin{equation}\label{eq:SRESnk}
M_2(|S_{n,k}(\theta)\>) = k M_2(|S_{1,1}(\theta)\>) 
=k\log_2\frac{2}{\|\Xi_{S_{1, 1}(\theta)}\|_4^4}
= -k\log_2 \frac{\cos (4\theta) + 7}{8}.
\end{equation}

\section{\label{app:BoundsAdd}Additional results on the variances $V(O,\rho)$ and $V_*(O,\rho)$}% 

In this appendix, we provide some additional results on the variances $V(O,\rho)$ and $V_*(O,\rho)$ in thrifty shadow estimation.

\subsection{Impact of depolarizing noise on  $V(O,\rho)$ and $V_*(O,\rho)$}

\begin{lemma}\label{lem:VDepolarization}
	Suppose $\caU$ is a unitary 2-design on $\caH$,  $\rho\in \caD(\caH)$, $O\in \caL^\rmH_0(\caH)$, and $\rho_p=(1-p)\rho+p\bbone/d$ with $0\leq p\leq 1$. Then 
	\begin{align}\label{eq:V*Depolarization}
	V_*(O,\rho_p) =(1-p)^2V_*(O,\rho).
	\end{align}
\end{lemma}

\Lref{lem:VDepolarization} is proved in \aref{app:VDepolarizeProof}; it shows that the variance $V_*(O,\rho)$ decreases monotonically with the strength of the  depolarizing noise acting on the state $\rho$. 
\Psref{pro:VarHaarDepolarize} and \ref{pro:VarClDepolarize} below are simple corollaries of \thsref{thm:VarHaar}, \ref{thm:VarFCl},~\pref{pro:VarNoiseShadow}, and \lref{lem:VDepolarization}.

\begin{proposition}\label{pro:VarHaarDepolarize}
Suppose $\caU$ is a unitary 3-design on $\caH$, $|\phi\>\in\caH$, $\rho=(1-p)|\phi\>\<\phi|+p \bbone / d$, and $O=|\phi\>\<\phi|- \bbone/d$. Then 
\begin{equation}
V(O,\rho)= \left(\frac{d-1}{d}\right)^2\left[-p^2+\frac{4dp}{(d-1)(d+2)}+\frac{d^2-3d-2}{(d-1)(d+2)}\right]+\frac{d^2-1}{d^2} \le \frac{2d(d+1)}{(d+2)^2},
\end{equation}
where the upper bound   is saturated when $p = 2d/[(d-1)(d+2)]$. If in addition $\caU$ is a unitary 4-design, then 
\begin{equation}
V_*(O,\rho)= \frac{4(1-p)^2(d-1)}{(d+2)(d+3)} \le \frac{4(d-1)}{(d+2)(d+3)},
\end{equation}
where the upper bound is saturated when $p=0$.
\end{proposition}

\begin{proposition}\label{pro:VarClDepolarize}
Suppose $\caU=\Cl_n$, $|\phi\>\in\caH$, $\rho=(1-p)|\phi\>\<\phi|+p \bbone / d$, and $O=|\phi\>\<\phi|- \bbone/d$. Then
\begin{equation}
V_*(O,\rho)= \frac{(1-p)^2\left[2^{1-M_2(\phi)}(d+1)-4\right]}{d+2} \le \frac{2^{1-M_2(\phi)}(d+1)-4}{d+2},
\end{equation}
where the upper bound is saturated when $p=0$.
\end{proposition}

\subsection{Mean variances in fidelity estimation based on thrifty shadow}

Here we  clarify the mean variances  in fidelity estimation based on thrifty shadow when the underlying unitary ensemble forms a unitary 2-design. The following proposition is proved in \aref{app:AverageVFproof}.

\begin{proposition}\label{pro:AverageVF}
	Suppose $\caU$ is a unitary 2-design on $\caH$ and $O=|\phi\>\<\phi|-\bbone/d$ with $|\phi\>\in\caH$. Then
	\begin{align}	
\bbE_{|\phi\>\sim\haar} V(O,\phi) &= \frac{2(d-1)}{d+2},  \label{eq:AverageVF}\\
	\bbE_{|\phi\>\sim\haar} V_*(O,\phi) &= \frac{4(d-1)}{(d+2)(d+3)}. \label{eq:AverageV*F}
	\end{align}
\end{proposition}
Note that the averages of $ V(O,\phi)$ and $V_*(O,\phi)$ are equal to the counterparts  in thrifty shadow estimation based on the Haar random unitary ensemble  (cf. \pref{pro:VarNoiseShadow} and  \thref{thm:VarHaar}).

\subsection{The variances $V(O,\rho)$ and $V_*(O,\rho)$ associated with the ensembles $\Cl_n$, $\bbU_{k,l}$, and  $\tbbU_k$ }
To start with we clarify the basic properties of the function $V_\triangle(O,\rho)$ defined in \eref{eq:VTriangleDef} and the derived function $V_\triangle (O) := \max_\rho V_\triangle(O,\rho)$, which will be useful to proving \coref{cor:VarCl} and a number of related results. 
\begin{proposition}\label{pro:VTriangle}
Suppose $\rho\in \caD(\caH)$ and $O\in \caL^\rmH_0(\caH)$. Then
\begin{align}
V_\triangle(O,\rho) &\le \frac{2(d+1)}{d(d+2)} \sqrt{2^{-\tM_2(\rho)}d\wp\|\Xi_O\|_4^4},\label{eq:VTriangleSRE}\\
V_\triangle(O,\rho)&\leq \frac{2(d+1)}{d+2}\|\Xi_{\mathring{\rho}}\|_\infty^2\|O\|_2^2,\label{eq:VTriangleNorminf}\\
V_\triangle(O) &\le \frac{2(d+1)}{d(d+2)}
\|\Xi_O^2\|_{[d]}\le\frac{2(d+1)}{d+2}\|O\|_2^2,\label{eq:VTriangleNorm2}
\end{align}
where $\wp$ is the purity of $\rho$, $\mathring{\rho}$ is the traceless part of $\rho$, and $\|\Xi_O^2\|_{[d]}$ is the sum of the $d$ largest entries of the squared characteristic function $\Xi_O^2$.
If $O=|\phi\>\<\phi|-\bbone/d$ with $|\phi\>\in\caH$, then
\begin{gather}
V_\triangle(O) < \frac{2^{1-M_2(\phi)/2}(d+1)}{d+2},\label{eq:VTriangleFUpper}\\
2^{1-M_2(\phi)}-\frac{5}{d}<V_\triangle(O,\phi) = \frac{\left[2^{1-M_2(\phi)}d^2-3d+1\right](d+1)}{d^2(d+2)} < 2^{1-M_2(\phi)}-\frac{2}{d}.\label{eq:VTriangleFExact}
\end{gather}
\end{proposition}

\begin{figure}[bt]
	\centering
	\includegraphics[width=0.35\linewidth]{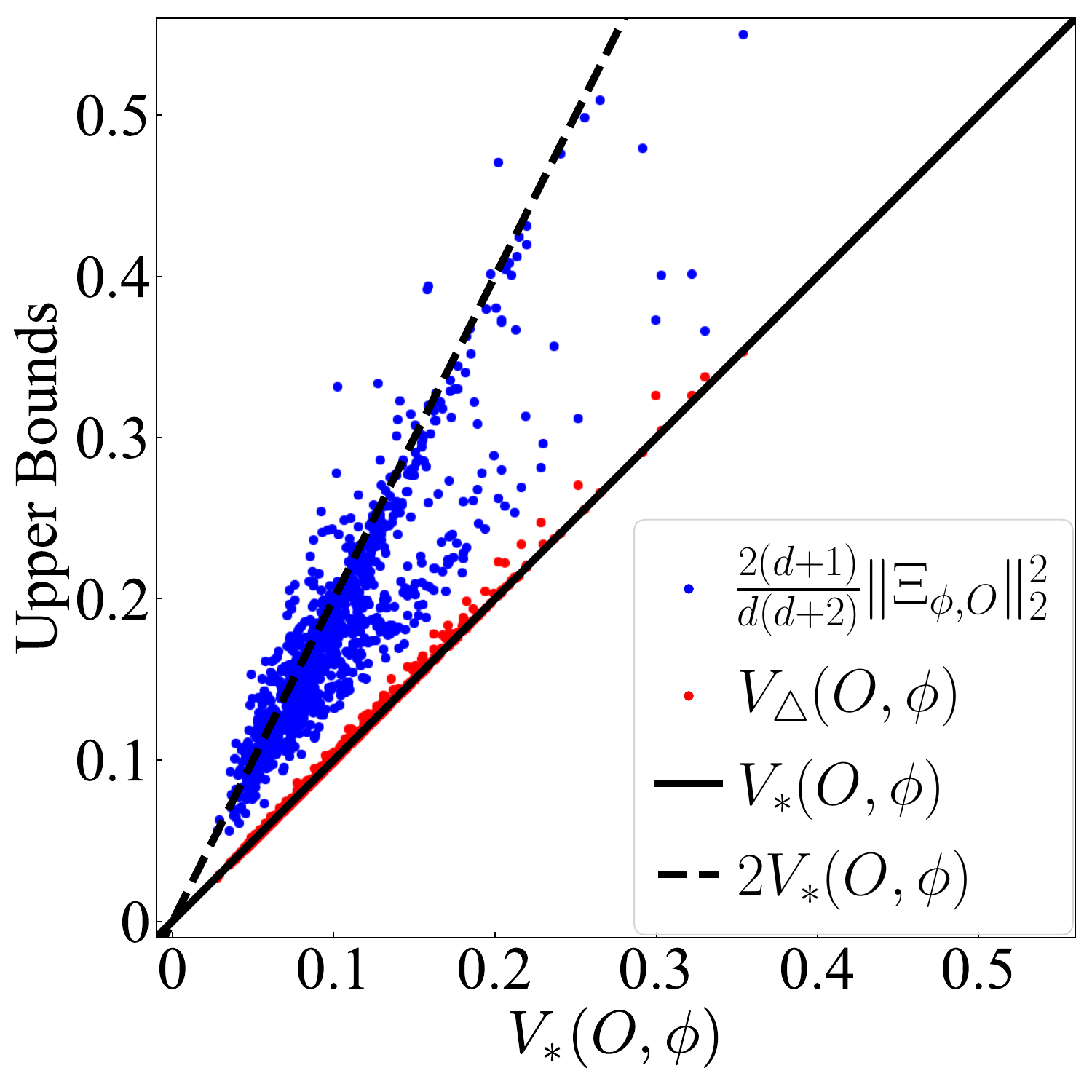}
	\caption{A scatter plot about  $V_*(O,\phi)$ and its upper bounds
		$V_\triangle(O,\phi)$ and $2(d+1)\|\Xi_{\phi,O}\|_2^2/[d(d+2)]$. Here  $O=|\rmW_3\>\<\rmW_3|-\bbone/d$, where $|\rmW_3\>$ is the three-qubit $\rmW$ state, and $\phi$ is sampled from the ensemble of  Haar random pure states 1000 times.}
	\label{fig:CompareUpper}
\end{figure}

Now, we consider the estimation protocol based on the Clifford group and show that the upper bounds for $V_*(O,\rho)$ presented in  \thref{thm:VarCl} are pretty good. \Fref{fig:CompareUpper} illustrates some numerical results on fidelity estimation in which  $O=|\rmW_3\>\<\rmW_3|-\bbone/d$ and $\rho=|\phi\>\<\phi|$ is a Haar random pure state. 
Note that the upper bound $V_\triangle(O,\rho)$ is
nearly tight in most cases, while the upper bound $2(d+1)\|\Xi_{\rho, O}\|_2^2/[d(d+2)]$ is around two times of $V_*(O,\rho)$.

\begin{figure}[t]
\centering
\includegraphics[width=0.45\textwidth]{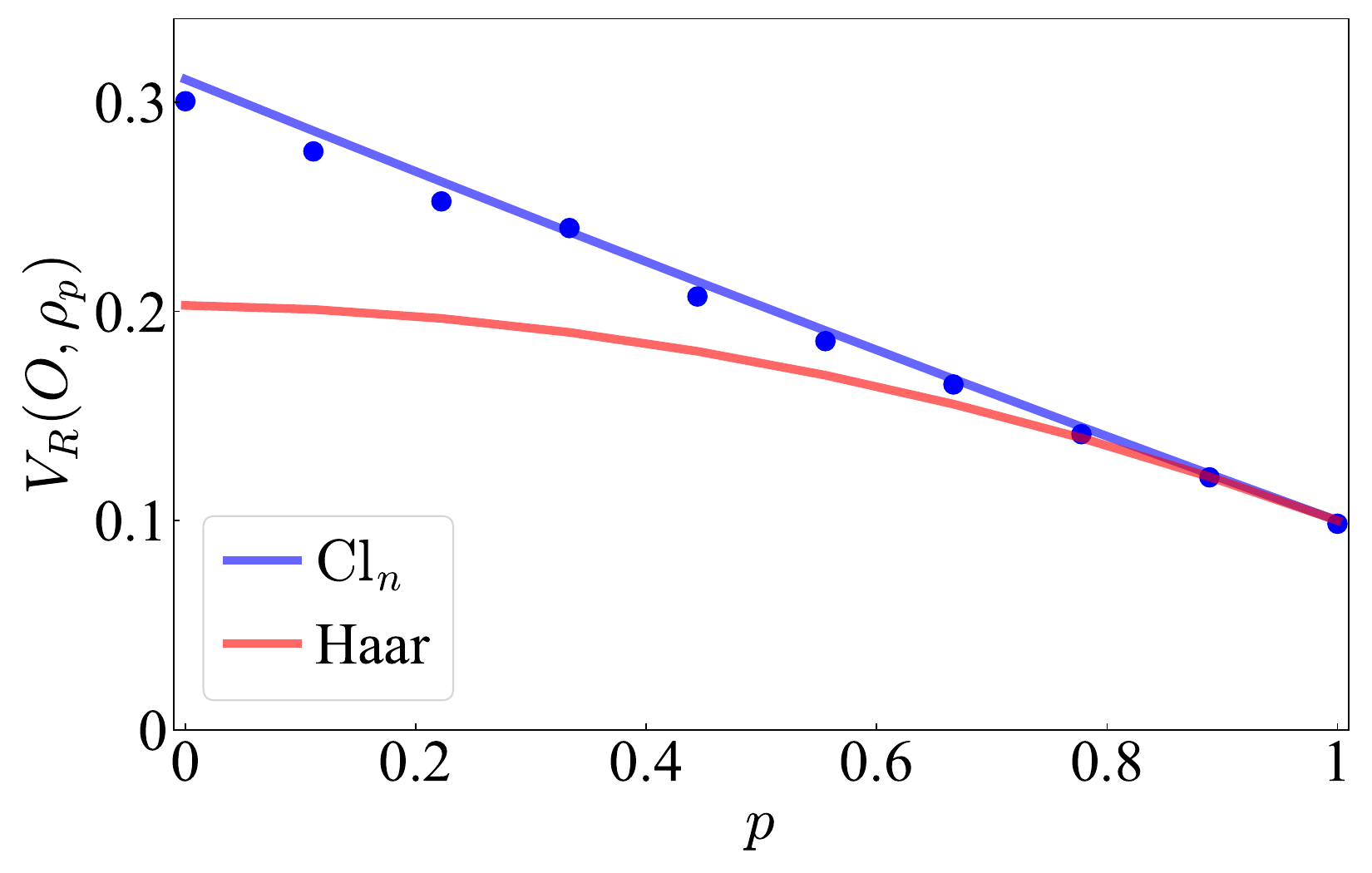}
\caption{The variance $V_R(O,\rho)$ in fidelity estimation based on thrifty shadow with Clifford measurements and Haar random measurements.  Here $R=10$, $O =|\rmW_{10}\>\<\rmW_{10}|- \bbone/d$, and 
 $\rho_p = (1-p)|\rmW_{10}\>\<\rmW_{10}|+p \bbone/d$. The two lines are based on analytical formulas in \eref{eq:VarMultiShot} and \psref{pro:VarHaarDepolarize}, \ref{pro:VarClDepolarize}, where $M_2(\rmW_{10})$ is determined in \pref{pro:CharNormPhasedW}.
 The blue dots are based on numerical simulation in which 20,000 random Clifford unitaries are sampled and each one is reused $R$ times for each data point. 
}
\label{fig:error}
\end{figure}

Next, we turn to thrifty shadow estimation based on the ensembles $\bbU_{k,l}$ and  $\tbbU_k$ (see \fref{fig:ModelCircuits}) and clarify the impact of $T$ gates on the variance $V_*(O,\rho)$. To simplify many formulas, it is convenient to introduce the two constants,
\begin{align}\label{eq:gammanu}
\gamma=\frac{3}{4},\quad \nu=\frac{1}{2}. 
\end{align}
When $\caU=\bbU_{k,l}$, the following result is a simple corollary of \Thref{thm:VarUkl},
\begin{align}
V_*(O,\rho) = \gamma^{kl}V_\triangle (O,\rho) + \caO\bigl(d^{-1}\bigr) \|O\|_2^2.\label{eq:VarUklAppro}
\end{align}
When $l=1$, this formula also applies to the ensemble $\tbbU_k$ thanks to \thref{thm:VartUk}. As analogs of \coref{cor:VarCl}, the two corollaries below follow from \thsref{thm:VarUkl},~\ref{thm:VartUk}, and \pref{pro:VTriangle}. 

\begin{corollary}
Suppose $\caU=\bbU_{k, l}$ with $0\leq k\leq n$ and $l\geq 0$, $\rho\in \caD(\caH)$, and $O\in \caL^\rmH_0(\caH)$. Then
\begin{align}
V_*(O,\rho) &\le \frac{2\gamma^{kl}}{d} \sqrt{2^{-\tM_2(\rho)}d\wp\|\Xi_O\|_4^4}+ \frac{6}{d}\|O\|_2^2,\\
V_*(O,\rho) &\le \frac{2\gamma^{kl}}{d} \|\Xi_{\mathring{\rho}}\|_\infty^2\|O\|_2^2+ \frac{6}{d}\|O\|_2^2,\\
V_*(O) &\le \frac{2\gamma^{kl}}{d} \|\Xi_O^2\|_{[d]}+ \frac{6}{d}\|O\|_2^2\le\left(2\gamma^{kl}+\frac{6}{d}\right)\|O\|_2^2.\label{eq:V*UklUB3}
\end{align}
\end{corollary}
The second upper bound for $V_*(O)$ in \eref{eq:V*UklUB3} improves over the upper bound in Theorem~4 of \rcite{Helsen2023MultiShot} (only applicable to $\bbU_{1,k}$) by about 15 times; the first upper bound is even much better whenever $\|\Xi_O^2\|_{[d]}\ll d\|O\|_2^2$, which is the case for a generic observable $O$.

\begin{corollary}
Suppose $\caU=\tbbU_k$ with $0\leq k\leq n$,  $\rho\in \caD(\caH)$, and $O\in\caL_0^{\rmH}(\caH)$. Then
\begin{align}
V_*(O,\rho) &\le \frac{2\gamma^{k}}{d} \sqrt{2^{-\tM_2(\rho)}d\wp\|\Xi_O\|_4^4}+ \frac{6}{d}\|O\|_2^2,\\
V_*(O,\rho) &\le \frac{2\gamma^{k}}{d} \|\Xi_{\mathring{\rho}}\|_\infty^2\|O\|_2^2+ \frac{6}{d}\|O\|_2^2,\\
V_*(O) &\le \frac{2\gamma^{k}}{d}\|\Xi_O^2\|_{[d]}+ \frac{6}{d}\|O\|_2^2\le\left(2\gamma^{k}+\frac{6}{d}\right)\|O\|_2^2.
\end{align}
\end{corollary}

\begin{figure}
	\centering
	\includegraphics[width=0.7\linewidth]{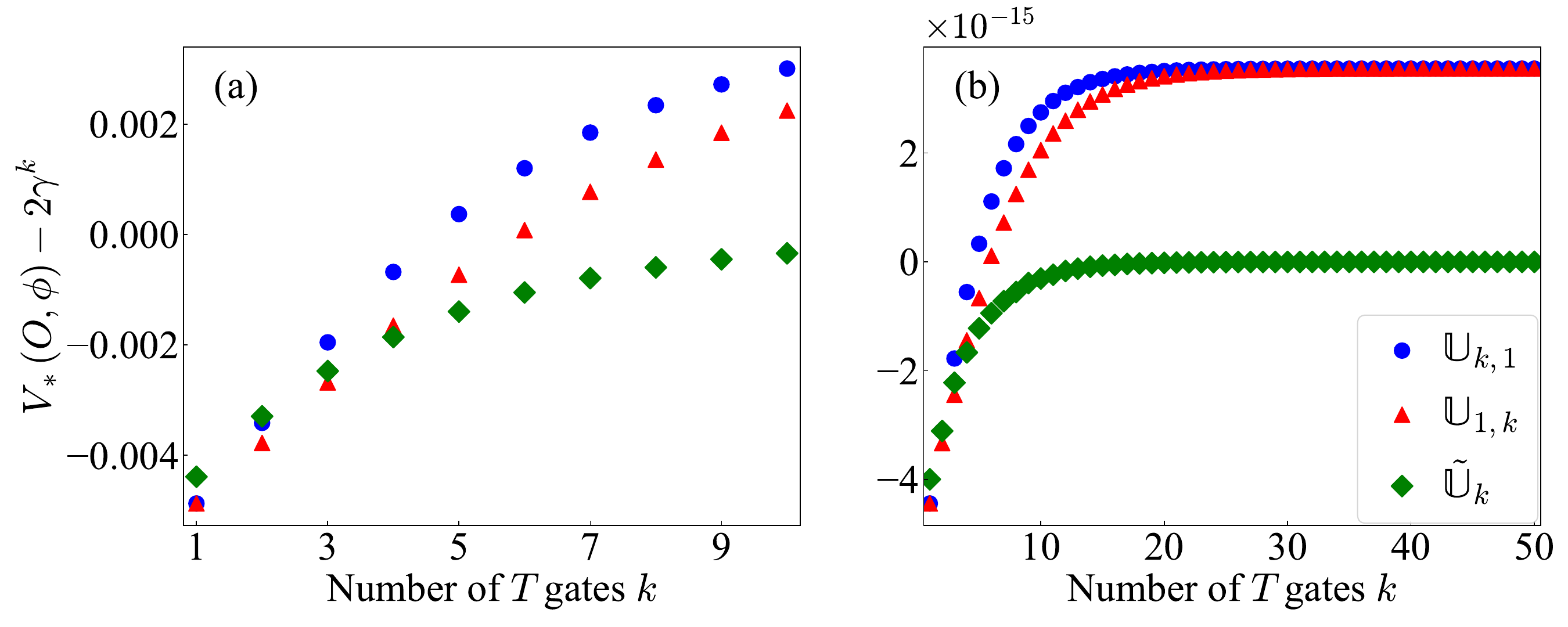}
	\caption{
		The variance $V_*(O,\phi)$ in fidelity estimation based on thrifty shadow with the three unitary ensembles  $\bbU_{k,1}$, $\bbU_{1,k}$, and $\tbbU_k$. Here  $O =|\phi\>\<\phi|- \bbone/d$,  $|\phi\>$ is an $n$-qubit stabilizer state with (a) $n=10$ and (b) $n=50$, and  $V_*(O,\phi)$ is determined by \psref{pro:VarUklFphi} and \ref{pro:VarFtUk} with $M_2(\phi)=0$.}
	\label{fig:CompareEnsemble}
\end{figure}

Next, we focus on the variance $V_*(O,\rho)$ in fidelity estimation based on the ensembles $\bbU_{k, l}$ and  $\tbbU_k$. As complements to \thsref{thm:VarUkl}-\ref{thm:VartUk}, the following two propositions are  proved in \asref{app:ProofUkl} and \ref{app:ProoftUk}, respectively.  

\begin{proposition}\label{pro:VarUklFphi}
Suppose $\caU=\bbU_{k, l}$ with $1\leq k\leq n$ and $l\geq 0$, and $O=|\phi\>\<\phi|-\bbone/d$ with $|\phi\>\in\caH$. Then
\begin{equation}\label{eq:VarUklFphi}
V_*(O,\phi)=\frac{4 (d-1)}{(d+2) (d+3)}+ \frac{ 2^{1-M_2(\phi)}(d+1)(d+3)-8(d+1)}{(d+2)(d+3)}\alpha_k^l,
\end{equation}
where $\alpha_k$ is defined in \eref{eq:alphakbetak} in \aref{app:Omega}. 
\end{proposition}

\begin{proposition}\label{pro:VarFtUk}
Suppose $\caU=\tbbU_k$ with $0\leq k\leq n$ and $O=|\phi\>\<\phi|-\bbone/d$ with $|\phi\>\in\caH$. Then
\begin{align}
V_*(O) &\le 2^{1-M_2(\phi)/2} \gamma^k + \frac{6}{d},\label{eq:VarFtUkBound}\\
V_*(O,\phi) &= \frac{ 2^{1-M_2(\phi)}\left[(d^3+4 d^2+3 d) \gamma ^k+(2 d^2+8 d +6)\nu^k-2 d^2 -12 d-10\right]}{(d-1) (d+2) (d+4)}\nonumber \\
&\quad+\frac{2 \left[-(4 d^2+4d) \gamma ^k-(8 d+8) \nu ^k+2 d^2+6 d+16\right]}{(d-1) (d+2) (d+4)},\label{eq:VarFtUkExact}\\
& -\frac{6}{d}<V_*(O,\phi)-2^{1-M_2(\phi)}\gamma^k<\frac{4}{d}. \label{eq:VarFtUkIdealBound}
\end{align}
\end{proposition}

If $|\phi\>$ is a stabilizer state, then $M_2(\phi)=0$ and \eref{eq:VarFtUkExact} implies that
\begin{align}
2\gamma^k-\frac{6}{d+2}\le V_*(O,\phi) = \frac{2 \left[(d^2 +d)\gamma ^k+(2 d+2) \nu ^k-6\right]}{(d+2) (d+4)} <2\gamma^k,
\end{align}
given that $0<\gamma^k,\nu^k\le 1$. Here the lower bound is saturated when $k=0$, so the upper bound in \eref{eq:VarFtUkIdealBound} is nearly tight.

Thanks to \thref{thm:VarFUkl} and
\pref{pro:VarFtUk}, $V_*(O,\phi) =2^{1-M_2(\phi)}\gamma^k +\caO(d^{-1})$
for shadow estimation based on all three ensembles $\bbU_{k,1}$,  $\bbU_{1,k}$, and $\tbbU_k$. So  the three ensembles  have similar performances in thrifty shadow estimation, and their distinctions vanish quickly as the system size increases, as illustrated in  \Fref{fig:CompareEnsemble}. 
Nevertheless, the ensemble $\tbbU_k$ is exceptional when $k=n=1$, in which case \eref{eq:VarFtUkExact} implies that
\begin{equation}
V_*(O,\phi) = \frac{1}{2}-\frac{3}{8}\times 2^{-M_2(\phi)}. 
\end{equation}
Here the variance $V_*(O,\phi)$ increases monotonically with $M_2(\phi)$, in sharp contrast with the general behavior for $n\geq 2$.  For example, $V_*(O,\phi) = 1/8$ when $|\phi\>$ is a stabilizer state, 
 while $V_*(O,\phi)=7/32$ when $|\phi\> = |S_{1,1}(\pi/4)\>$.

\subsection{\label{app:VDepolarizeProof}Proof of \lref{lem:VDepolarization}}
Before proving \lref{lem:VDepolarization}, we need to introduce an auxiliary lemma.
\begin{lemma}\label{lem:OmegaUPT}
	Suppose $\caU$ is a unitary 2-design on $\caH$. Then 
	\begin{align}\label{eq:OmegaUPT}
	\tr_4 \Omega(\caU)=\frac{(\Tensor{\bbone}{2}+\SWAP)\otimes \bbone}{d+1}. 
	\end{align}
\end{lemma}
\begin{proof}[Proof of \lref{lem:OmegaUPT}]
	By virtue of the definition in \eref{eq:OmegaU} we can deduce that
	\begin{align}
	\tr_4 \Omega(\caU)&=\sum_{\bfa,\bfb} \bbE_{U\sim\caU} \dagtensor{U}{3} \bigl[\Tensor{(|\bfa\>\<\bfa|)}{2}\otimes |\bfb\>\<\bfb|\bigr] \Tensor{U}{3}=\sum_{\bfa} \bbE_{U\sim\caU} \Tensor{\left(U^\dag|\bfa\>\<\bfa|U\right)}{2}\otimes \bbone=\frac{(\bbone^{\otimes 2}+\SWAP)\otimes \bbone}{d+1},
	\end{align}
	which confirms \eref{eq:OmegaUPT}. 
	Here the second equality holds because $\sum_{\bfb}|\bfb\>\<\bfb|=\bbone$, and the last equality follows from \lref{lem:UnitaryAvg}. 
\end{proof}

\begin{proof}[Proof of \lref{lem:VDepolarization}] \lref{lem:VDepolarization} is a simple corollary of  \eref{eq:DefV*2Design} and the following equations,
\begin{align}
[\tr(O\rho_p)]^2& =(1-p)^2 [\tr(O\rho)]^2,\\
\tr[\Omega(\caU)\Tensor{(O\otimes\rho_p)}{2}] &=  (1-p)^2\tr[\Omega(\caU)\Tensor{(O\otimes \rho)}{2}]+\frac{ p(1-p)}{d}\tr[\Omega(\caU)(O\otimes \rho\otimes O\otimes\bbone)]\nonumber\\
&\quad +\frac{ p(1-p)}{d}\tr[\Omega(\caU)(O\otimes \bbone\otimes O\otimes\rho)] + \frac{ p^2}{d^2}\tr[\Omega(\caU)(O\otimes \bbone\otimes O\otimes\bbone)]\nonumber\\
&=  (1-p)^2\tr[\Omega(\caU)\Tensor{(O\otimes \rho)}{2}],
\end{align}
given that $O$ is traceless by assumption. Here the second equality follows from \lref{lem:OmegaUPT}, which means
\begin{align}
\tr[\Omega(\caU)(O\otimes \bbone\otimes O\otimes\rho)]=\tr[\Omega(\caU)(O\otimes \rho\otimes O\otimes\bbone)] =\tr[\Omega(\caU)(O\otimes \bbone\otimes O\otimes\bbone)]=0.
\end{align}
\end{proof}

\subsection{\label{app:AverageVFproof}Proof of \pref{pro:AverageVF}}
\begin{proof}[Proof of \pref{pro:AverageVF}]
Let $O_0=|0\>\<0|-\bbone/d$. By virtue of \eref{eq:VarShadow} we can deduce that (cf. the proof of \pref{pro:VarNoiseShadow})
\begin{align}
\bbE_{|\phi\>\sim\haar}V(O,\phi) &=  \bbE_{W\sim\haar} V(WO_0W^\dagger, W|0\>)\nonumber\\
&=\bbE_{U\sim\caU}\bbE_{W\sim \haar} \sum_{\bfb} (d+1)^2\tr\left[\Tensor{\bigl(U^\dag |\bfb\>\<\bfb|U\bigr)}{3} W^{\otimes 3}\bigl(|0\>\<0|\otimes\Tensor{O_0}{2}\bigr)W^{\dag\otimes 3}\right]-\frac{(d-1)^2}{d^2}\nonumber\\
&=\frac{6(d+1)}{d+2}\tr\bigl[P_{[3]}\bigl(|0\>\<0|\otimes\Tensor{O_0}{2}\bigr)\bigr]-\frac{(d-1)^2}{d^2}=\frac{d+1}{d+2}\bigl[\tr\bigl(O_0^2\bigr)+2\tr\bigl(|0\>\<0| O_0^2\bigr)\bigr]-\frac{(d-1)^2}{d^2}\nonumber\\
&=\frac{d+1}{d+2}\biggl[\frac{d-1}{d}+\frac{2(d-1)^2}{d^2} \biggr]-\frac{(d-1)^2}{d^2}
=\frac{2(d-1)}{d+2},
\end{align}
where $P_{[3]}$ is the projector onto the symmetric subspace in $\caH^{\otimes 3}$, and "Haar" denotes the Haar measure on the set of pure states or the unitary group on $\caH$ depending on the context. This equation confirms \eref{eq:AverageVF}; note that the third expression above  is exactly the variance in  shadow  estimation based on the Haar random unitary ensemble.

By virtue of \eref{eq:DefV*2Design} we can deduce that
	\begin{align}
	\bbE_{|\phi\>\in\haar} V_*(O,\phi) &= \bbE_{W\sim\haar} V_*(WO_0W^\dagger, W|0\> )\nonumber\\
	&=(d+1)^2 \bbE_{W\sim\haar}\tr\left[\Omega(\caU)\Tensor{W}{4}\Tensor{(O_0\otimes|0\>\<0|)}{2}\dagtensor{W}{4} \right]-\frac{(d-1)^2}{d^2} \nonumber\\
	&=(d+1)^2 \tr\left[\Omega(\haar)\Tensor{(O_0\otimes|0\>\<0|)}{2} \right]-\frac{(d-1)^2}{d^2}=\frac{4(d-1)}{(d+2)(d+3)},
	\end{align}
which confirms \eref{eq:AverageV*F}. Here	
$\Omega(\haar)$ is the cross moment operator associated  with the Haar random unitary ensemble and the last equality follows from 
\thref{thm:VarHaar} (proved in \aref{app:ProofHaar}); note that the third expression above  is exactly the variance in thrifty shadow estimation based on the Haar random unitary ensemble.
\end{proof}

\subsection{Proof of \pref{pro:VTriangle}}
\begin{proof}[Proof of \pref{pro:VTriangle}]
According to the definition of $V_\triangle(O,\rho)$ in \eref{eq:VTriangleDef}, we have
\begin{equation}\label{eq:VTriangleProof}
V_\triangle(O,\rho) \le \frac{2(d+1)}{d(d+2)}\|\Xi_{\rho, O}\|_2^2 \le \frac{2(d+1)}{d(d+2)}\|\Xi_O\|_4^2 \|\Xi_\rho\|_4^2=\frac{2(d+1)}{d(d+2)} \sqrt{2^{-\tM_2(\rho)}d\wp\|\Xi_O\|_4^4},
\end{equation}
which confirms \eref{eq:VTriangleSRE}.
Here the first inequality follows from \lref{lem:CharOrhoNorm}, the second inequality follows from the Cauchy-Schwarz inequality, and the equality follows from the definition of $\tM_2(\rho)$ in \eref{eq:DefM2}.  In addition,
\begin{align}
V_\triangle(O,\rho) &\le \frac{2(d+1)}{d(d+2)}\|\Xi_{\rho, O}\|_2^2 = \frac{2(d+1)}{d(d+2)}\Xi_\rho^2 \cdot \Xi_O^2 = \frac{2(d+1)}{d(d+2)}\Xi_{\mathring{\rho}}^2 \cdot \Xi_O^2\nonumber\\
&\le \frac{2(d+1)}{d(d+2)}\|\Xi_{\mathring{\rho}}^2\|_\infty\|\Xi_O^2\|_1 = \frac{2(d+1)}{d+2}\|\Xi_{\mathring{\rho}}\|_\infty^2\|O\|_2^2,
\end{align}
which confirms \eref{eq:VTriangleNorminf}. Here the second equality holds because $\tr(O)=0$, the second inequality follows from H\"{o}lder's inequality, and the last equality holds because $\|\Xi_O^2\|_1=\|\Xi_O\|_2^2=d\|O\|_2^2$ by \eref{eq:CharOrelation}.  \Eref{eq:VTriangleNorm2} follows  from the first inequality in \eref{eq:VTriangleProof} and \lref{lem:CharOrhoNorm}.

%thereby completing the proof of the first statement.

When $O=|\phi\>\<\phi|-\bbone/d$, by virtue of \eref{eq:VTriangleProof} and \lref{lem:CharOrhoNormF} we can deduce that
\begin{align}
V_\triangle(O)\le \frac{2(d+1)}{d(d+2)}\|\Xi^2_{\phi}\|_{[d]} \le \frac{2(d+1)}{d^{1/2}(d+2)}\|\Xi^4_{\phi}\|_{[d]}^{1/2}\le \frac{2(d+1)}{d^{1/2}(d+2)} \|\Xi_{\phi}\|_4^2 =\frac{2^{1-M_2(\phi)/2}(d+1)}{d+2},
\end{align}
which confirms \eref{eq:VTriangleFUpper}. Here, the second inequality holds because the arithmetic mean is no more than the quadratic mean, and the third inequality holds because $\Xi_{\phi}^4(P)\ge 0$ for each $P\in \bcaP_n$. 
\Eref{eq:VTriangleFExact} follows from the definition of $V_\triangle(O,\rho)$ in \eref{eq:VTriangleDef}  and \lref{lem:CharOrhoNormF} with $\rho=|\phi\>\<\phi|$. 
\end{proof}

\section{The commutant of the $t$th Clifford tensor power}\label{app:SWDClifford}
To better understand the mathematical foundation of thrifty shadow estimation based on the Clifford group, here we need to recapitulate the main results on the commutant of the $t$th Clifford tensor power following  \rscite{Zhu2016Fail4Design,Gross2021Duality}.

\subsection{Schur-Weyl duality of Clifford tensor powers}\label{app:LagrangianSub}
The famous  Schur-Weyl duality states that the commutant of the $t$th (diagonal) tensor power of the unitary group $\rmU(d)$ on $\caH$ is spanned by 
permutation operators \cite{Goodman2009SchurWeyl,Pavel2009SchurWeyl}. 
Since the Clifford group $\Cl_n$ forms a unitary 3-design, but not a 4-design \cite{Zhu2016Fail4Design,Gross2021Duality}, the commutant of its $t$th tensor power is also generated by permutation operators when $t\leq 3$, but is strictly larger when $t\geq 4$. To describe the commutant, we need to introduce additional concepts.

A subspace $\caT\le \bbZ_2^{2t}$ is a \emph{stochastic Lagrangian subspace} \cite{Gross2021Duality} if it satisfies the following three conditions:
\begin{enumerate}
\item $\bfx\cdot\bfx=\bfy\cdot\bfy \mod 4$ for all $(\bfx,\bfy)\in \caT$.
\item $\caT$ has dimension $t$. 
\item $\mathbf{1}_{2t} = (11\cdots1)^\top\in \caT$.
\end{enumerate}
Here $\bfx$ and $\bfy$ are vectors in $\bbZ_2^{t}$, and $(11\cdots1)^\top$ is a shorthand for $(1,1,\dots,1)^\top$. The first condition means $\caT$ is totally isotropic with respect to the quadratic form $\frq(\bfx,\bfy)=\bfx\cdot\bfx-\bfy\cdot\bfy \mod 4$. The first two conditions together means $\caT$ is a Lagrangian subspace. The set of all stochastic Lagrangian subspaces in $ \bbZ_2^{2t}$ is denoted by  $\Sigma_{t,t}$ henceforth.

A $t\times t$ matrix $O$ over $\bbZ_2$ is a  \emph{stochastic isometry}  if $O\bfx\cdot O\bfx = \bfx\cdot \bfx \mod 4$ for all $\bfx\in\bbZ_2^t$ \cite{Gross2021Duality}. 
All $t\times t$ stochastic isometries over $\bbZ_2$  form the stochastic orthogonal group, which is denoted by $O_t$ henceforth. Note that every  $t\times t$ permutation matrix over $\bbZ_2$ is a stochastic isometry. The group of all such permutation matrices is a subgroup of $O_t$ and is denoted by $S_t$ henceforth. Alternatively, permutations in $S_t$ can be represented by cycles or their products.
Each stochastic isometry $O$ in $O_t$ defines a stochastic Lagrangian subspace as follows, 
\begin{align}
\caT_O := \{(O\bfx,\bfx)\;|\;\bfx\in\bbZ_2^t\}. 
\end{align}
In this way, $O_t$ and  $S_t$ can be identified as subsets of the set  $\Sigma_{t,t}$ of stochastic Lagrangian subspaces. In addition, $O\in O_t$ has  natural left and right actions on $\Sigma_{t,t}$ defined as follows,
\begin{equation}
O\caT := \left\{(O\bfx,\bfy)\;|\;(\bfx,\bfy)\in \caT \right\},\quad TO := \left\{\bigl(\bfx,O^\top\bfy\bigr)\;|\;(\bfx,\bfy)\in \caT \right\}.
\end{equation}
Under these actions, $\Sigma_{t,t}(d)$ is decomposed into a disjoint union of double cosets:
\begin{equation}\label{eq:doublecosets}
\Sigma_{t,t}(d) = O_t(d)\caT_1 O_t(d) \cup O_t(d)\caT_2 O_t(d) \cup \cdots \cup O_t(d)\caT_k O_t(d),
\end{equation}
where $\caT_1, \caT_2,\dots, \caT_k\in \Sigma_{t,t}(d)$ are  coset representatives. Note that one of the double cosets can be identified with $O_t(d)$.

To further understand the structure of  stochastic Lagrangian subspaces, we need to introduce two additional concepts. A subspace $\caN\le \bbZ_2^t$ is a \emph{defect subspace} if $\bfx\cdot\bfx=0\mod 4$ for all $\bfx\in \caN$, which means $\caN$ is $q$-isotropic with respect to the quadratic form $q(\bfx) = \bfx\cdot\bfx \mod 4$ \cite{Gross2021Duality}. In this case, we have $\caN\leq \caN^\perp$, where $\caN^\perp$ denotes the orthogonal complement of $\caN$ with respect to the inner product modulo 2. Given two defect subspaces $\caN, \caN'$ in $\bbZ_2^t$, a linear map $J: \caN^\perp/\caN\rightarrow \caN'^\perp/\caN'$ is called a \emph{defect isomorphism} if it satisfies the following two conditions,
\begin{enumerate}
\item $q\left(J[\bfx]\right)=q\left([\bfx]\right)$ for all $[\bfx]\in \caN^\perp/\caN$.
\item $J[\mathbf{1}_t] = [\mathbf{1}_t]$.
\end{enumerate}
In that case, a stochastic Lagrangian subspace in $\Sigma_{t,t}$ can be constructed as follows,
\begin{equation}
\caT = \{(\bfx+\bfz,\bfy+\bfw)\; |\; [\bfy]\in \caN^\perp/\caN, [\bfx]=J[\bfy] ,\bfw\in \caN,\bfz\in \caN'\}.
\end{equation}
Moreover, any  stochastic Lagrangian subspace in $\Sigma_{t,t}$ can be induced from two defect subspaces and a defect isomorphism
in this way \cite{Gross2021Duality}.  If $\caN'=\caN$ and the defect isomorphism $J$ is the identity map, then $\caT$ is called  a  stochastic Lagrangian subspace of Calderbank-Shor-Steane (CSS) type and can be expressed as follows,
\begin{equation}\label{eq:TCSS}
\caT = \{(\bfx+\bfz,\bfx+\bfw)\; |\; [\bfx]\in \caN^\perp/\caN,\bfw, \bfz\in \caN
\}.
\end{equation}

For each stochastic Lagrangian subspace $\caT$ in $\Sigma_{t,t}$, we can define an operator $r(\caT)$ on $\bbC_2^{\otimes t}$ and an operator $R(\caT)$ on $\Tensor{\bigl(\bbC_2^{\otimes t}\bigr)}{n}$ as follows:
\begin{equation}\label{eq:DefRT}
r(\caT) := \sum_{(\bfx,\bfy)\in \caT} |\bfx\>\<\bfy|,\quad R(\caT) := \Tensor{r(\caT)}{n}.
\end{equation}
This definition is compatible with the action of the  stochastic orthogonal group $O_t$ in the following sense
\begin{align}\label{eq:RORT}
 r(O\caT)= r(O)r(\caT),\quad  R(O\caT)= R(O)R(\caT)\quad \forall O\in O_t,\quad \caT\in \Sigma_{t,t}.
\end{align}
Note that $\Tensor{\bigl(\bbC_2^{\otimes t}\bigr)}{n}$ is isomorphic to $\caH^{\otimes t}$, and $R(\caT)$ can be regarded as an operator on $\caH^{\otimes t}$ thanks to this  isomorphism. Moreover, $R(\caT)$ belongs to the commutant of $\Cl_n^{\tilde{\otimes}t}:=\{U^{\otimes t}\;|\;U\in \Cl_n\}$ \cite{Gross2021Duality}, where the notation $\tilde{\otimes}$ is used to denote the diagonal tensor power. Actually, $\{R(\caT)\}_{\caT\in \Sigma_{t,t}}$  spans the commutant  when  $n\ge t-1$ \cite{Gross2021Duality}. If $\caT$ is a stochastic Lagrangian subspace of CSS type, then it is uniquely determined by a defect subspace $\caN$ via \eref{eq:TCSS}. In addition, $\caN$ can be used to  define a CSS stabilizer group as follows,
\begin{equation}\label{eq:CSS}
\CSS(\caN):=\{Z_{\bfp}X_{\bfq}\;|\; \bfp,\bfq\in \caN\},
\end{equation}
where $Z_\bfp$ and $X_\bfq$ are defined in \eref{eq:ZpXq}. Denote by $|\caN|$  the cardinality of $\caN$ and let 
\begin{equation}\label{eq:DefCSSProjector}
P_{\caN}=\frac{1}{|\caN|^2} \sum_{\bfp,\bfq\in \caN} Z_{\bfp}X_{\bfq}
\end{equation}
be the projector onto the corresponding code space. Then the operators  $r(\caT)$ and $R(\caT)$ can be expressed as follows,
\begin{equation}\label{eq:CSSProjector}
r(\caT) = |\caN|P_{\caN},\quad  R(\caT)=r(\caT)^{\otimes n} = |\caN|^n P_{\caN}^{\otimes n}.
\end{equation}
This observation explains why $\caT$ is called a stochastic Lagrangian subspace of CSS type.

\subsection{The commutant of the fourth Clifford tensor power}
In this section, we focus on the commutant of the fourth Clifford tensor power \cite{Zhu2016Fail4Design,Gross2021Duality}, which corresponds to the case  $t=4$. It is known that $O_4=S_4$ has order 24; $\Sigma_{4,4}$ has cardinality 30 and can be decomposed into two double cosets,
\begin{equation}\label{eq:Sigma44}
\Sigma_{4,4}=S_4\cup S_4 \caT_4 S_4,
\end{equation}
where  $\caT_4$ is the stochastic Lagrangian subspace of CSS type determined by  the defect subspace  $\caN = \mathrm{span}\bigl\{(1111)^\top\bigr\}$ 
via \eref{eq:TCSS}.
Note that  $\caN$ is the only nontrivial defect subspace in $\bbZ_2^4$. In addition, $\caT_4$ can be represented by the generating matrix
\begin{equation}\label{eq:FormT4}
\left(\begin{array}{cccc|cccc}
1 & 0 & 0 & 1 & 1 & 0 & 0 & 1\\
0 & 1 & 0 & 1 & 0 & 1 & 0 & 1\\
0 & 0 & 0 & 0 & 1 & 1 & 1 & 1\\
1 & 1 & 1 & 1 & 0 & 0 & 0 & 0
\end{array}\right)^{\!\!\!\!\top}.
\end{equation}

\begin{lemma}\label{lem:S3T4}
Suppose  $\tS_3$ is the subgroup of $S_4$ that is generated by the two cycles $(12)$ and $(123)$. Then 
\begin{gather}
\tS_3 \caT_4=S_4 \caT_4 S_4,   \label{eq:S3T4}\\
\tr R(\sigma\caT_4)=d^{l(\sigma)}\quad \forall \sigma\in \tS_3,  \label{eq:sigmaT4tr}
\end{gather}
where $R(\caT)$ is defined in \eref{eq:DefRT} and $l(\sigma)$ denotes the number of cycles in $\sigma$.
\end{lemma}

\begin{proof}
Note that $\caN=\mathrm{span}\bigl\{(1111)^\top
\bigr\}$ is the only nontrivial defect subspace in $\bbZ_2^4$ and is invariant under the action of $S_4$, and $\caT_4$ is determined by $\caN$  via \eref{eq:TCSS}. 
The orthogonal complement $\caN^\perp$ of $\caN$ is a three-dimensional subspace spanned by the three vectors  $(1100)^\top, (0110)^\top, (1111)^\top$, while $\caN^\perp/\caN$ is a two-dimension vector space spanned by  $\bigl[(1100)^\top\bigr]$ and $\bigl[(0110)^\top\bigr]$. The group of defect automorphisms on $\caN^\perp/\caN$ is isomorphic to $S_3$ and $\tS_3$.

In addition, the action of each $\sigma\in S_4$ on $\caT_4$ can be expressed as follows,
\begin{equation}
\sigma\caT_4 = \bigl\{(\sigma\bfx+\sigma\bfz,\bfx+\bfw)\; |\; [\bfx]\in \caN^\perp/\caN,\bfw, \bfz\in \caN
\bigr\}=\bigl\{(\sigma\bfx+\bfz,\bfx+\bfw)\; |\; [x]\in \caN^\perp/\caN,\bfw, \bfz\in \caN
\bigr\}.
\end{equation}
Note that each $\sigma\in S_4$ induces a defect automorphism on $\caN^\perp/\caN$, and
this map is injective when restricted to $\tS_3$. In addition, $\sigma\caT_4=\caT_4$ iff the induced defect automorphism is trivial (coincides with the identity map). Therefore, $|\tS_3\caT_4|=6=|S_4\caT_4S_4|$, which implies \eref{eq:S3T4}. 

Next, suppose $\sigma\in \tS_3$; then 
\begin{align}
\dim\{[\bfx]\in \caN^\perp/\caN \;|\; [\sigma \bfx]=[\bfx]  \}=l(\sigma)-1.
\end{align}
In conjunction with the previous equation we can deduce that
\begin{align}
\tr r(\sigma\caT_4)=2^{l(\sigma)},\quad \tr R(\sigma\caT_4)=[\tr r(\sigma\caT_4)]^n=d^{l(\sigma)},
\end{align}
which confirms \eref{eq:sigmaT4tr} and completes the proof of \lref{lem:S3T4}. 
\end{proof}

When  $n\ge 3$, the  set $\{R(\caT)=\Tensor{r(\caT)}{n}\}_{\caT\in\Sigma_{4,4}}$ is linearly independent and spans  the commutant of $\Cl_n^{\tilde{\otimes}4}$ \cite{Gross2021Duality}. Here we show that $\{R(\caT)\}_{\caT\in\Sigma_{4,4}}$ spans  the commutant of $\Cl_n^{\tilde{\otimes}4}$ for all $n\geq 1$. 
\begin{proposition}\label{pro:CommSpan}
Suppose $n\geq 1$; then $\{R(\caT)\}_{\caT\in\Sigma_{4,4}}$ spans the commutant of $\Cl_n^{\tilde{\otimes}4}$.
\end{proposition}

\begin{proof}
Denote the elements in $\Sigma_{4,4}$ by $\caT_1, \caT_2, \ldots, \caT_{30}$ and let  $\Gamma$ be the Gram matrix of the set $\{R(\caT)\}_{\caT\in\Sigma_{4,4}}$ with  entries 
\begin{equation}
\Gamma_{ij}:= \tr\bigl[R(\caT_i)^\dag R(\caT_j)\bigr]=\bigl\{\tr\bigl[r(\caT_i)^\dag r(\caT_j)\bigr]\bigr\}^n  =d^{\dim \caT_i\cap\caT_j}.
\end{equation}
To simplify the computation, note that $\Sigma_{4,4}$ can be decomposed into two double cosets: $\Sigma_{4,4}=S_4\cup \tS_3 \caT_4$ according to \eref{eq:Sigma44} and \lref{lem:S3T4}. Now, we need to distinguish three cases. 
If $\sigma,\tau\in S_4$, then 
\begin{equation}
\tr\bigl[R(\sigma)^\dag R(\tau)\bigr] = \tr [R(\sigma^{-1}\tau)] = d^{l(\sigma^{-1}\tau)},
\end{equation}
where $l(\sigma^{-1}\tau)$ denotes the number of cycles in $\sigma^{-1}\tau$. 
If $\sigma\in S_4$ and $\caT = \tau \caT_4$ with  $\tau\in \tS_3$, then by virtue of \eref{eq:RORT} we can deduce that
\begin{equation}
\tr\bigl[R(\sigma)^\dag R(\caT)\bigr] = \tr [R(\sigma^{-1}\tau \caT_4)]=\tr[R(\sigma' \caT_4)] = d^{l(\sigma')},
\end{equation}
where $\sigma'$ is the unique permutation in $\tS_3$ that satisfies the condition  $\sigma'\caT_4=\sigma^{-1}\tau\caT_4$, and the last equality follows from \lref{lem:S3T4}.  
If  $\caT_i = \sigma \caT_4$ and $\caT_j = \tau \caT_4$ with $\sigma,\tau\in \tS_3$, then
\begin{equation}
\tr\bigl[R(\caT_i)^\dag R(\caT_j)\bigr] = \tr\bigl[R(\sigma^{-1}\tau) R(\caT_4)^2\bigr] =d\tr[R(\sigma^{-1}\tau) R(\caT_4)]=d\tr[R(\sigma'\caT_4)] = d^{1+l(\sigma')},
\end{equation}
where $\sigma'$ is the unique permutation in $\tS_3$ that satisfies the condition  $\sigma'\caT_4=\sigma^{-1}\tau\caT_4$ as before. Here  the first two equalities hold  because $\caT_4$ is a stochastic Lagrangian subspace of CSS type that  is determined by  the one-dimensional defect subspace  $\caN = \mathrm{span}\{(1111)^\top\}$, which means $R(\caT_4)$ is proportional to a projector and
$R(\caT_4)^2=dR(\caT_4)$ by  \eref{eq:CSSProjector}. In this way, we can calculate all entries of the Gram matrix $\Gamma$.

Now, the eigenvalues of the Gram matrix $\Gamma$ read
\begin{equation}
d(d-1)(d-2)(d-4),\;\; d(d+1)(d+2)(d+4),\;\; d(d^2-1)(d-2),\;\; d(d^2-1)(d+2),
\end{equation}
with multiplicities $1$, $1$, $14$, and $14$, respectively. Since the dimension of the span of $\{R(\caT)\}_{\caT\in\Sigma_{4,4}}$ coincides with the rank of $\Gamma$, it follows that
\begin{equation}
 \dim \mathrm{span}\left(\{R(\caT)\}_{\caT\in\Sigma_{4,4}}\right)=\rank \Gamma = \left\{ 
\begin{aligned}
&15, &n=1,\\
&29, &n=2,\\
&30, &n\ge3.
\end{aligned}\right.
\end{equation}
This result  matches the dimension of the commutant of $\Cl_n^{\tilde{\otimes}t}$~\cite{Zhu20173Design}  and implies that $\{R(\caT)\}_{\caT\in\Sigma_{4,4}}$ spans the commutant of $\Cl_n^{\tilde{\otimes}4}$, which completes the proof of \pref{pro:CommSpan}.
\end{proof}

\subsection{Projectors based on Schur-Weyl duality}\label{app:ProjectorSWD}
Recall that $\caH$ is a $d$-dimensional Hilbert space. According to  Schur-Weyl duality, the $t$th tensor power  $\Tensor{\caH}{t}$ can be decomposed into multiplicity-free irreducible representations of $\rmU(d)\times S_t$ as follows,
\begin{equation}
\Tensor{\caH}{t} =\bigoplus_\lambda \caH_\lambda =  \bigoplus_\lambda \caW_\lambda \otimes \caS_\lambda,
\end{equation}
where each $\lambda$ represents a non-increasing partition of $t$ into no more than $d$ parts, $\caW_\lambda$ is the corresponding  Weyl module carrying an irreducible representation of $\rmU(d)$, and  $\caS_\lambda$ is the Specht module carrying an irreducible representation of $S_t$. In addition, the representations of $\rmU(d)$ (and similarly for $S_t$) associated with different partitions are inequivalent. 

In the rest of this section, we focus on the special case $t=4$. Then  the projector onto $\caH_\lambda$, denoted by $P_\lambda$ henceforth, can be calculated as follows,
\begin{equation}\label{eq:FormPlambda}
P_{\lambda} =\frac{d_\lambda}{24} \sum_{\sigma\in S_4} \chi_\lambda(\sigma) R(\sigma),
\end{equation}
where $\chi_\lambda(\sigma)$ is the character of $\sigma$ corresponding to  the representation $\lambda$ and $R(\sigma)$ is the unitary operator that represents the permutation $\sigma$. 
When $d\ge 4$, there are five distinct  partitions, namely, $[4], [3,1], [2,2], [2,1,1], [1,1,1,1]$.

Next, we decompose $\caH^{\otimes 4}$ into irreducible presentations of $\Cl_n\times S_4$ following \rcite{Zhu2016Fail4Design} and 
 explain the connection with the language used in \sref{app:LagrangianSub}. To this end, we need to introduce a stabilizer code and the corresponding stabilizer projector. Note that $\{P^{\otimes 4}\}_{P\in \bcaP_n}$ is a stabilizer group of CSS type 
and is identical to the stabilizer group
$\CSS(\caN)^{\otimes n}$, where  $\caN=\mathrm{span}\bigl\{(1111)^\top\bigr\}$
is the unique nontrivial defect subspace in $\bbZ_2^4$, and $\CSS(\caN)$ is the CSS stabilizer group defined in \eref{eq:CSS}. 
Let $\caV_n$ be the corresponding stabilizer code and
$P_n$ the stabilizer projector onto $\caV_n$. Then 
\begin{equation}\label{eq:FormRT4}
P_n= \frac{1}{d} R(\caT_4) = \frac{1}{d^2}\sum_{P\in\bcaP_n}P^{\otimes 4}=\frac{1}{d^2} \Tensor{\left(\Tensor{I}{4}+\Tensor{X}{4}+\Tensor{Y}{4}+\Tensor{Z}{4}\right)}{n},
\end{equation}
where  $\caT_4$ is the stochastic Lagrangian subspace determined by the defect subspace $\caN$ via \eref{eq:TCSS}, and the first two equalities follow from \eqsref{eq:DefCSSProjector}{eq:CSSProjector}. Let $\caV_n^\bot$ be the orthogonal complement of $\caV_n$ and define
\begin{align}
\caH_\lambda^+:=H_\lambda\cap\caV_n,\quad \caH_\lambda^-:=H_\lambda\cap\caV_n^\bot,\quad D_\lambda^\pm:=\dim \caH_\lambda^\pm;
\end{align}
then $\caH_\lambda=\caH_\lambda^+\bigoplus \caH_\lambda^-$. 
Explicit formulas for $D_\lambda^\pm$ can be derived by virtue of \tref{tab:SpaceDimension} in \rcite{Zhu2016Fail4Design} (note that $D_\lambda^\pm$ in this paper correspond to $d_\lambda D_\lambda^\pm$ in \rcite{Zhu2016Fail4Design}).
Let $P_\lambda^s$ for $s=\pm$ be the projector onto $\caH_\lambda^\pm$; then 
\begin{align}
P_\lambda^+=P_\lambda P_n =P_n P_\lambda,\quad P_\lambda^-=P_\lambda (\bbone-P_n)= (\bbone-P_n)P_\lambda, 
\end{align}
given that $P_n$ and $P_\lambda$ commute with each other.
Now, $\Tensor{\caH}{4}$ can be decomposed into multiplicity-free irreducible representations of $\Cl_n\times S_4$ as follows \cite{Zhu2016Fail4Design},
\begin{equation}
\Tensor{\caH}{4} =\bigoplus_{\lambda,s=\pm |D_\lambda^s\neq 0} \caH_\lambda^s. 
\end{equation}
Moreover, the representations of $\Cl_n$ associated with different subspaces in the above decomposition are inequivalent.

In preparation for studying the cross moment operator $\Omega(\caU)$ defined in \eref{eq:OmegaU}, here we need to consider a special subgroup of $S_4$. Define
\begin{equation}\label{eq:DefG}
G := \{(e),(12),(34),(12)(34)\}. 
\end{equation}
Simple analysis shows that $G$ has four 
 irreducible representations, all of which are one-dimensional. The characters of these representations are presented in \tref{tab:RepDecom}.
In addition, each irreducible representation of $S_4$ can be decomposed into a direct sum of irreducible representations of $G$, and the corresponding multiplicities are also presented in \tref{tab:RepDecom}, which shows that the decomposition is multiplicity free.

\begin{table}[t]
	\centering
	\normalsize
	\renewcommand{\arraystretch}{1.5}
	\caption{(left) Character table of  irreducible representations of the group $G$ defined in \eref{eq:DefG}. (right) Multiplicities in decomposing each irreducible representation of $S_4$  into irreducible representations of $G$.}
	\begin{minipage}{0.35\textwidth}
		\centering
		\begin{math}
		\begin{array}{c|cccc}
		\hline\hline
		i & \chi_i((e)) & \chi_i((12)) & \chi_i((34)) & \chi_i((12)(34)) \\ \hline
		1 & 1 & 1 & 1 & 1 \\ 
		2  & 1 & -1 & -1 & 1 \\ 
		3 & 1 & -1 & 1 & -1 \\ 
		4  & 1 & 1 & -1 & -1 \\ \hline\hline
		\end{array}
		\end{math}
	\end{minipage}
	\hspace{0.05\textwidth}
	\begin{minipage}{0.35\textwidth}
		\centering
		\begin{math}
		\begin{array}{l|cccc}
		\hline\hline
		\lambda & m_1 & m_2 & m_3 & m_4 \\ \hline
		[4] & 1 & 0 & 0 & 0 \\ 
		{[3,1]}  & 1 & 0 & 1 & 1 \\
		{[2,2]} & 1 & 1 & 0 & 0 \\ 
		{[2,1,1]}  & 0 & 1 & 1 & 1 \\
		{[1,1,1,1]}  & 0 & 1 & 0 & 0 \\ \hline\hline
		\end{array}
		\end{math}   
	\end{minipage}
	\label{tab:RepDecom}
\end{table}

Let  $\caH_G$ be the 
the subspace of $\caH^{\otimes 4}$ that is stabilized by $G$ or, equivalently,  the subspace that carries the trivial representation of $G$; let  $P_G$ be the projector onto $\caH_G$. Then 
\begin{equation}\label{eq:PGHG}
 P_G = \frac{1}{4}\sum_{\sigma\in G} R(\sigma),\quad \caH_G= \supp (P_G). 
\end{equation}
In addition, $P_G$  commutes with the projectors $P_n$, $P_\lambda$, and $P_\lambda^\pm$ introduced above. 
Define
\begin{align}
\caH_{\lambda,G}:=\caH_\lambda\cap\caH_G,\quad D_{\lambda,G}:=\dim\caH_{\lambda,G}\quad \caH_{\lambda,G}^\pm:=\caH_\lambda^\pm\cap\caH_G,\quad D_{\lambda,G}^\pm :=\dim\caH_{\lambda,G}^\pm; 
\end{align}
then  $\caH_{\lambda,G}$ is an invariant subspace of $\rmU(d)\times G$, and  $\caH_{\lambda,G}^\pm$  are invariant subspaces of $\Cl_n\times G$. Explicit formulas for the dimensions $D_{\lambda,G}$ and $D_{\lambda,G}^\pm$ are summarized in \tref{tab:SpaceDimension}. 
Let  $P_{\lambda,G}$ and $P_{\lambda,G}^\pm$ be the projectors onto $\caH_{\lambda,G}$ and $\caH_{\lambda,G}^\pm$, respectively; then 
\begin{align}\label{eq:DefProjectors}
P_{\lambda,G}=P_\lambda P_G=P_G P_\lambda 
, \quad  D_{\lambda,G}=\tr P_{\lambda,G},\quad 
\quad P_{\lambda,G}^\pm=P_\lambda^\pm P_G=P_G P_\lambda^\pm,\quad D_{\lambda,G}^\pm=\tr P_{\lambda,G}^\pm . 
\end{align}

\begin{table}[t]
\renewcommand{\arraystretch}{2}
\normalsize
\caption{Dimension formulas for $D_{\lambda,G}=\dim\caH_{\lambda, G}$ and $D_{\lambda, G}^{s}=\dim\caH_{\lambda, G}^{s}$ with $s=\pm$.}
\centering
\begin{math}
\begin{array}{l|ccc}
\hline\hline
\lambda		& D_{\lambda, G} & D_{\lambda, G}^+ & D_{\lambda, G}^- \\  \hline
[4]	& \frac{d(d+1)(d+2)(d+3)}{24} & \frac{(d+1)(d+2)}{6}  &  \frac{(d^2-1)(d+2)(d+4)}{24}  \\ 
{[3,1]}	& \frac{d(d+2)(d^2-1)}{8}  & 0 & \frac{d(d+2)(d^2-1)}{8}   \\
{[2,2]}	& \frac{d^2(d^2-1)}{12}  & \frac{(d^2-1)}{3}  & \frac{(d^2-4)(d^2-1)}{12}  \\ 
{[2,1,1]}	& 0  & 0 & 0 \\ 
{[1,1,1,1]}& 0  &  0 &  0  \\ 
\hline\hline
\end{array} 
\end{math}
\label{tab:SpaceDimension}
\end{table}

In conjunction with Proposition~5 in \rcite{Zhu2016Fail4Design} and  \tref{tab:RepDecom} we can deduce the following lemma. 
\begin{lemma}\label{lem:RepV4}
	The representations of $\rmU(d)\times G$ and $\Cl_n\times G$ on $\caH^{\otimes 4}$ are both multiplicity free. The representations of $\rmU(d)$ and $\Cl_n$ on $\caH_G$ are also multiplicity free. In particular, $\rmU(d)$ is irreducible on $\caH_{\lambda,G}$ whenever $D_{\lambda,G}\neq 0$; $\Cl_n$ is irreducible on $\caH_{\lambda,G}^s$  whenever $D_{\lambda,G}^s\neq 0$. 
\end{lemma}

According to Schur-Weyl duality and \lref{lem:RepV4}, $\caH_{\lambda}$ and $\caH_{\lambda, G}$ carry  representations of $\rmU(d)$ and are 
invariant under  $\Tensor{U}{4}$ for any $U\in\rmU(d)$, that is,
\begin{align}
\dagtensor{U}{4} P_{\lambda} \Tensor{U}{4} = P_{\lambda},\quad 	\dagtensor{U}{4} P_{\lambda, G} \Tensor{U}{4} = P_{\lambda, G}\quad \forall U\in \rmU(d). \label{eq:UPlambdaG}
\end{align}
Similarly,  $\caH_{\lambda}^+$ and $\caH_{\lambda, G}^+$ are 
invariant under  $\Tensor{U}{4}$ for any $U\in\Cl_n$, that is,
\begin{align}
\dagtensor{U}{4} P_{\lambda}^+ \Tensor{U}{4} = P_{\lambda}^+,\quad 	\dagtensor{U}{4} P_{\lambda, G}^+ \Tensor{U}{4} = P_{\lambda, G}^+\quad \forall U\in \Cl_n. \label{eq:UPlambdaGCl}
\end{align}
In addition, $\caH_G$ can be decomposed into multiplicity-free irreducible representations of $\rmU(d)$ and  $\Cl_n$ as follows,
\begin{align}
\caH_G =\bigoplus_{\lambda |D_{\lambda,G}\neq 0} \caH_{\lambda,G} =\bigoplus_{\lambda,s |D_{\lambda,G}^s\neq 0} \caH_{\lambda,G}^s.
\end{align}
Accordingly, we have
\begin{align}\label{eq:PGdecom}
P_G =\bigoplus_{\lambda |D_{\lambda,G}\neq 0} P_{\lambda,G} =\bigoplus_{\lambda,s |D_{\lambda,G}^s\neq 0} P_{\lambda,G}^s.
\end{align}

Thanks to \eqsref{eq:UPlambdaG}{eq:UPlambdaGCl}, the projectors  $P_{\lambda, G}$  and  $P_{\lambda, G}^s$ commute with $U^{\otimes 4}$ for all $U\in \Cl_n$ and thus can be expanded in terms of  $R(\caT)$ for $\caT\in \Sigma_{4, 4}$, where $\Sigma_{4, 4}$ is the set of stochastic Lagrangian subspaces presented  in 
\eref{eq:Sigma44}. In addition,  $P_{\lambda, G}$  and  $P_{\lambda, G}^s$ are invariant under left and right multiplication of unitary operators associated with permutations in $G$. 
Under the left and right actions of $G$,  the set $\Sigma_{4,4}$ is divided into the following five orbits:
\begin{enumerate}
\item $(e),(12),(34),(12)(34)$;
\item $(13),(23),(14),(24),(123),(132),(124),(142),(134),(143),(234),(243),(1234),(1243),(1342),(1432)$;
\item $(13)(24),(14)(23),(1324),(1423)$;
\item $\caT_4, (12)\caT_4$;
\item $(13)\caT_4, (23)\caT_4, (123)\caT_4, (132)\caT_4$.
\end{enumerate}
Denote by $\scrT_i$ the $i$th orbit and
define 
\begin{equation}\label{eq:DefcaRi}
\caR_i = \sum_{\caT\in \scrT_i} R(\caT),\quad i=1,2,3,4,5. 
\end{equation}
In conjunction with the definitions in \eref{eq:DefProjectors}, we can expand $P_{\lambda, G}$ and $P_{\lambda, G}^+$ in terms of the five operators as follows, 
\begin{equation}\label{eq:RelationPandRi}
\begin{aligned}
P_{[4],G} &= P_{[4]}=\frac{1}{24} (\caR_1 + \caR_2 + \caR_3),&\quad P_{[4],G}^+ &=P_{[4]}^+= \frac{1}{6d} (\caR_4+\caR_5),&
P_{[2,2], G} &= \frac{1}{24} (2\caR_1 - \caR_2 + 2\caR_3),\\
P_{[2,2], G}^+ &=\frac{1}{6d}(2\caR_4-\caR_5),&
P_{[3,1], G} &= \frac{1}{8} (\caR_1 - \caR_3), & P_{[3,1], G}^+&=0.
\end{aligned}
\end{equation}
Based on this equation it is easy to expand $P_{\lambda, G}^-=P_{\lambda, G}-P_{\lambda, G}^+$. Note that $P_{\lambda, G}^\pm=P_{\lambda, G}=0$ when $\lambda=[2,1,1]$ or $\lambda=[1,1,1,1]$ (see \tref{tab:SpaceDimension}).

When $n\geq 3$,  the set $\{R(\caT)\}_{\caT\in \Sigma_{4, 4}}$ is  linearly independent, so $\{\caR_i\}_{i=1}^5$ is also linearly independent. When $n=2$,  $\{\caR_i\}_{i=1}^5$ is linearly independent by direct calculation. In both cases, the decompositions in \eref{eq:RelationPandRi} are unique. 
In the special case $n=1$, $\{\caR_i\}_{i=1}^5$ is not linearly independent, but spans a four-dimensional subspace in the operator space.

\section{Properties of the cross moment operator}\label{app:Omega}
In this section we clarify the properties of the cross moment operator $\Omega(\caU)$ defined in \eref{eq:OmegaU} in a  slightly more general situation
in which the measurement basis is not necessarily the computational basis.
Given a unitary ensemble $\caU$ and an orthonormal  basis $\caB$ on $\caH$, define
\begin{equation}\label{eq:DefOmegacaB}
\Omega(\caU, \caB) := \sum_{|\psi\>,|\varphi\>\in\caB} \bbE_{U\sim\caU} \dagtensor{U}{4} \left[\Tensor{(|\psi\>\<\psi|)}{2}\otimes \Tensor{(|\varphi\>\<\varphi|)}{2}\right] \Tensor{U}{4}.
\end{equation}
 Note that  $\Omega(\caU, \caB)$ is invariant under left and right multiplication by $R(\sigma)$ for $\sigma\in G$, where the group $G$ is defined in \eref{eq:DefG}, 
so $\Omega(\caU, \caB)$ is supported in the subspace $\caH_G$ stabilizer by $G$ as defined in \eref{eq:PGHG}. If $\caB$ is the standard computational basis $\{|\bfb\>\}_{\bfb\in \{0,1\}^n}$, then $\Omega(\caU, \caB)$ reduces to $\Omega(\caU)$.

\begin{lemma}\label{lem:OmegaPlambdaG}
Suppose $\caU$ is any unitary ensemble on $\caH$, $\caB$ is an orthonormal basis of $\caH$, and $\lambda$ is a non-increasing partition of $4$ into no more than $d=2^n$ parts. Then
\begin{gather}
	\tr[\Omega(\caU,\caB)P_{\lambda,G}]=\tr[\Omega(\caU, \caB)P_{\lambda}],  \label{eq:OmegaPlambdaG0}\\
\tr\bigl[\Omega(\caU, \caB)P_{[4]}\bigr]=\frac{d(d+5)}{6},
	\quad \tr\bigl[\Omega(\caU, \caB)P_{[2,2]}\bigr]=\frac{d(d-1)}{3}, \quad  \tr\bigl[\Omega(\caU, \caB)P_{[3,1]}\bigr]=\frac{d(d-1)}{2}, \label{eq:OmegaPlambdaG}\\
\tr\bigl[\Omega(\caU, \caB)P_{[2,1,1]}\bigr]=
	\tr\bigl[\Omega(\caU, \caB)P_{[1,1,1,1]}\bigr]=0, \quad n\geq 2.
	\label{eq:OmegaPlambdaG2}
	\end{gather}
\end{lemma}

\begin{proof}
\Eref{eq:OmegaPlambdaG0} holds because $\Omega(\caU, \caB)$ is supported in $\caH_G$ and $P_{\lambda,G}=P_GP_\lambda P_G$. In conjunction with \eref{eq:UPlambdaG} we can deduce that 

\begin{align}
\tr[\Omega(\caU, \caB)P_{\lambda,G}]=\tr[\Omega(\caU, \caB)P_{\lambda}]=\tr(\Omega_0P_{\lambda}) =\tr(\Omega_0P_{\lambda,G}),
\end{align}
where $\Omega_0:=\sum_{\bfa,\bfb}\bigl[\Tensor{(|\bfa\>\<\bfa|)}{2}\otimes \Tensor{(|\bfb\>\<\bfb|)}{2}\bigr]$. In addition, a straightforward calculation shows that
\begin{align}
\tr\bigl(\Omega_0P_{[4]}\bigr)=\frac{d(d+5)}{6}, \quad
\tr\bigl(\Omega_0P_{[2,2],G}\bigr)=\frac{d(d-1)}{3}, \quad
\tr\bigl(\Omega_0P_{[3,1],G}\bigr)=\frac{d(d-1)}{2}, 
\end{align}
which implies \eref{eq:OmegaPlambdaG}. \Eref{eq:OmegaPlambdaG2} holds because $D_{\lambda,G}=0$ and $P_{\lambda,G}=0$ for $\lambda=[2,1,1]$ and $\lambda=[1,1,1,1]$ according to \tref{tab:SpaceDimension}.
\end{proof}

For the convenience of the following discussion, given  $\lambda=[4],\;[2,2],\;[3,1]$ and $s=\pm$, define
\begin{align}\label{eq:Defkappa}
\kappa_\lambda(\caU,\caB)&:=
\frac{\tr[\Omega(\caU, \caB)P_{\lambda,G}]}{D_{\lambda,G}}, \quad \kappa_\lambda^s(\caU,\caB):=\begin{cases}
\frac{\tr\bigl[\Omega(\caU, \caB)P_{\lambda,G}^s\bigr]}{D_{\lambda,G}^s} &D_{\lambda,G}^s\neq 0,\\
0 &D_{\lambda,G}^s=0.
\end{cases}
\end{align}
For completeness, we set  $\kappa_\lambda^s(\caU,\caB)=\kappa_\lambda(\caU,\caB)=0$ when $\lambda=[2,1,1],\;[1,1,1,1]$. According to \tref{tab:SpaceDimension} and \lref{lem:OmegaPlambdaG}, $\kappa_\lambda(\caU,\caB)$
is independent of  the specific choices of the ensemble $\caU$ and basis $\caB$ and can be written as $\kappa_\lambda(\haar)$ henceforth; $\kappa_{[3,1]}^\pm(\caU,\caB)$ are also independent of  $\caU$ and $\caB$.  To be specific, we have
\begin{align}\label{eq:kappaHaar}
\begin{aligned}
\kappa_{[4]}(\caU,\caB)&=\kappa_{[4]}(\haar)= \frac{4(d+5) }{(d+1)(d+2)(d+3)},&\quad 
\kappa_{[2,2]}(\caU,\caB)&=\kappa_{[2,2]}(\haar)=\frac{4}{d(d+1)},\\
\kappa_{[3,1]}^-(\caU,\caB)&=\kappa_{[3,1]}(\caU,\caB)=\kappa_{[3,1]}(\haar)=\frac{4}{(d+1)(d+2)},&\quad \kappa_{[3,1]}^+(\caU,\caB)&=0.
\end{aligned}
\end{align}
In the special case $n=1$, we have
\begin{align}
\kappa_{[2,2]}^+(\caU,\caB)&=\kappa_{[2,2]}(\caU,\caB)=\kappa_{[2,2]}(\haar)=\frac{4}{d(d+1)},\quad \kappa_{[2,2]}^-(\caU,\caB)=0. 
\end{align}

\subsection{Haar random ensemble}

If the unitary ensemble $\caU$ is the Haar random ensemble or a unitary $4$-design, then the cross moment operator $\Omega(\caU,\caB)$ is independent of  the specific choices of the ensemble $\caU$ and basis $\caB$; it is written as $\Omega(\haar)$ henceforth. 
\begin{proposition}\label{pro:OmegaHaar}
	Suppose $n\ge1$, $\caU$ is the Haar random ensemble or a unitary 4-design on $\caH$, and $\caB$ is any orthonormal basis of $\caH$.
	Then 
	\begin{align}
	\Omega(\caU,\caB)&=\Omega(\haar) =\kappa_{[4]}(\haar)P_{[4]} + \kappa_{[2,2]}(\haar)P_{[2,2], G} + \kappa_{[3,1]}(\haar)P_{[3,1], G}   \nonumber\\
 &= \frac{4[d(d+5) P_{[4]}+(d+2)(d+3)P_{[2,2],G}+d(d+3) P_{[3,1],G}]}{d(d+1)(d+2)(d+3)}, \label{eq:OmegaHaar}
	\end{align}
	where $\kappa_{\lambda}(\haar)$ for $\lambda=[4],\;[2,2],\;[3,1]$ are presented in  \eref{eq:kappaHaar}. 	
\end{proposition}

Since $\Omega(\haar)$ is a linear combination of three mutually orthogonal projectors by \eref{eq:OmegaHaar}, its Schatten $p$-norm  can be calculated analytically,
\begin{align}
\left\|\Omega(\haar) \right\|_p^p=&\;\frac{d(d+1)(d+2)(d+3)}{24}\left[\frac{4(d+5)}{(d+1)(d+2)(d+3)}\right]^p +\frac{d^2(d^2-1)}{12} \left[\frac{4}{d(d+1)}\right]^p \nonumber\\
&+ \frac{d(d+2)(d^2-1)}{8} \left[\frac{4}{(d+1)(d+2)}\right]^p,
\end{align}
which implies that 
\begin{align}\label{eq:OmegaHaarNorm1Inf}
\|\Omega(\haar) \|_1=d^2,\quad \|\Omega(\haar)\|_\infty = \frac{4}{d(d+1)}. 
\end{align}
When $n\gg1$, we have $\Omega(\haar)\approx 4P_G/d^2$ according to \eqsref{eq:PGdecom}{eq:OmegaHaar}.

\begin{proof}[Proof of \pref{pro:OmegaHaar}]
The first equality in \eref{eq:OmegaHaar} follows from the definition of 	$\Omega(\caU,\caB)$ and the assumption that $\caU$ is a unitary 4-design.
The second equality follows from \lref{lem:RepV4},  Schur's lemma, and the definitions in \eref{eq:Defkappa} given that $\Omega(\haar)$ is supported in $\caH_G$. 	
The last equality follows from \eref{eq:kappaHaar}.	
\end{proof}

\subsection{Clifford ensemble}
In this section we turn to the cross moment operator $\Omega(\Cl_n, \caB)$ based on the Clifford group $\Cl_n$ 
 and an orthonormal basis $\caB$. For the convenience of the following discussion, define
\begin{equation}\label{eq:DefLambda12}
\Lambda_1(\caB):= \sum_{|\psi\>,|\varphi\>\in\caB} \left(\|\Xi_{\psi,\phi}\|_2^2+2\tXi_{\psi,\phi}\cdot\Xi_{\psi,\phi}\right),\quad
\Lambda_2(\caB):=\sum_{|\psi\>,|\varphi\>\in\caB}\left(\|\Xi_{\psi,\phi}\|_2^2 - \tXi_{\psi,\phi}\cdot\Xi_{\psi,\phi}\right).
\end{equation}
Here we assume that  $n\ge2$, while the special case $n=1$ is discussed in \aref{app:SingleCase} separately.

\begin{proposition}\label{pro:OmegaCl}
	Suppose $n\ge2$ and $\caB$ is an orthonormal basis for $\caH$.
	Then 
	\begin{align}
	\Omega(\Cl_n, \caB) &= \kappa_{[4]}^+(\Cl_n,\caB)P_{[4]}^++\kappa_{[4]}^-(\Cl_n,\caB)P_{[4]}^- + \kappa_{[2,2]}^+(\Cl_n,\caB)P_{[2,2],G}^+ \nonumber\\
        &\quad+ \kappa_{[2,2]}^-(\Cl_n,\caB)P_{[2,2],G}^- + \kappa_{[3,1]}(\Cl_n,\caB)P_{[3,1],G}, \label{eq:OmegaCl}
	\end{align}
	where 
	\begin{equation}\label{eq:kappaCl}
	\begin{aligned}
	\kappa_{[4]}^+(\Cl_n,\caB) &= \frac{2\Lambda_1(\caB)}{d^2(d+1)(d+2)},&\quad \kappa_{[4]}^-(\Cl_n,\caB) &= \frac{4d^3(d+5)-8\Lambda_1(\caB)}{d^2(d^2-1)(d+2)(d+4)} ,\\
	\kappa_{[2,2]}^+(\Cl_n,\caB) &= \frac{2\Lambda_2(\caB)}{d^2(d^2-1)},
	&\quad \kappa_{[2,2]}^-(\Cl_n,\caB) &= \frac{4d^3(d-1)-8\Lambda_2(\caB)}{d^2(d^2-1)(d^2-4)},\\
	\kappa_{[3,1]}(\Cl_n,\caB) &= \frac{4}{(d+1)(d+2)}.
	\end{aligned}
	\end{equation}
If $\caB$ is the standard computational basis, then 
	\begin{equation}\label{eq:kappaClb}
	\begin{aligned}
	&\kappa_{[4]}^+(\Cl_n,\caB) = \kappa_{[2,2]}^+(\Cl_n,\caB) = \frac{2}{(d+1)},\\
	&\kappa_{[4]}^-(\Cl_n,\caB) = \kappa_{[2,2]}^-(\Cl_n,\caB) = \kappa_{[3,1]}(\Cl_n,\caB) = \frac{4}{(d+1)(d+2)}.
	\end{aligned}
	\end{equation}
\end{proposition}

At this point, it is instructive to quantify the difference between $\Omega(\Cl_n, \caB)$ and $\Omega(\haar)$. 
By virtue of \psref{pro:OmegaHaar} and \ref{pro:OmegaCl} we can determine  the $1$-norm of this difference,
\begin{align}
\left\|\Omega(\Cl_n, \caB)-\Omega(\haar)\right\|_1 =& \; \frac{ (d^2 - 1) (d + 2) (d + 4)}{3 d (d + 3)} \left|\kappa_{[4]}^+(\Cl_n,\caB) - \kappa_{[4]}^-(\Cl_n,\caB) \right|\nonumber\\
&+ \frac{2 (d^2 - 4) (d^2 - 1)}{3 d^2} \left|\kappa_{[2,2]}^+(\Cl_n,\caB)-\kappa_{[2,2]}^-(\Cl_n,\caB)\right|.
\end{align}
In addition, we can derive explicit formulas for  the $1$-norm and  spectral norm of $\Omega(\Cl_n)-\Omega(\haar)$,
\begin{equation}
\begin{aligned}
\left\|\Omega(\Cl_n)-\Omega(\haar)\right\|_1 &= \frac{2 (d-1) \left(3 d^2+6 d-10\right)}{3 d (d+3)}=\caO(d),\\
\left\|\Omega(\Cl_n)-\Omega(\haar)\right\|_\infty &= \frac{2(d-1)(d+4)}{(d+1)(d+2)(d+3)}=\caO(d^{-1}).
\end{aligned}
\end{equation}
As a benchmark, we have $\|\Omega(\haar)\|_1=d^2$ and $\|\Omega(\haar)\|_\infty = 4/[d(d+1)]$ by \eref{eq:OmegaHaarNorm1Inf}.

\begin{proof}[Proof of \pref{pro:OmegaCl}]
\Eref{eq:OmegaCl} follows from \lref{lem:RepV4}, Schur's lemma, \tref{tab:SpaceDimension}, and \eref{eq:Defkappa} given that $\Omega(\Cl_n, \caB)$ is supported in $\caH_G$ and  that
$D_{[3,1],G}^+=D_{[2,1,1],G}^\pm=D_{[1,1,1,1],G}^\pm=0$, $D_{[3,1],G}^-=D_{[3,1],G}$, and $P_{[3,1],G}^-=P_{[3,1],G}$ by \tref{tab:SpaceDimension}. In addition,
by virtue of \lref{lem:CrossChar} we can derive the following results,
	\begin{equation}
	\begin{aligned}
	\tr\left[\Omega(\Cl_n, \caB)P_{[4]}^+\right] &= \frac{1}{6d^2} \sum_{|\psi\>,|\varphi\>\in\caB}\left(\|\Xi_{\psi,\phi}\|_2^2+ \|\tXi_{\psi,\phi}\|_2^2+ 4\tXi_{\psi,\phi}\cdot\Xi_{\psi,\phi} \right)=\frac{\Lambda_1(\caB)}{3d^2}  ,\\
	\tr\left[\Omega(\Cl_n, \caB)P_{[4]}^-\right] &= \tr\left[\Omega(\Cl_n, \caB)P_{[4]}\right]-\tr\left[\Omega(\Cl_n, \caB)P_{[4]}^+\right] = \frac{d(d+5)}{6}-
\frac{\Lambda_1(\caB)}{3d^2} 	,\\
 \tr\left[\Omega(\Cl_n, \caB)P_{[2,2],G}^+\right] &= \frac{1}{3d^2} \sum_{|\psi\>,|\varphi\>\in\caB} \left(\|\Xi_{\psi,\phi}\|_2^2+ \|\tXi_{\psi,\phi}\|_2^2-2\tXi_{\psi,\phi}\cdot\Xi_{\psi,\phi} \right)=\frac{2\Lambda_2(\caB)}{3d^2} ,\\
    \tr\left[\Omega(\Cl_n, \caB)P_{[2,2],G}^-\right] &= \tr\left[\Omega(\Cl_n, \caB)P_{[2,2],G}\right] - \tr\left[\Omega(\Cl_n, \caB)P_{[2,2],G}^+\right] = \frac{d(d-1)}{3} - 
       \frac{2\Lambda_2(\caB)}{3d^2}  ,    \\  
    \tr\left[\Omega(\Cl_n, \caB)P_{[3,1],G}^-\right] &= \tr\left[\Omega(\Cl_n, \caB)P_{[3,1],G}\right]=\frac{d(d-1)}{2},
	\end{aligned}
	\end{equation}
which imply \eref{eq:kappaCl} given 
 the definitions in \eref{eq:Defkappa} and the dimension formulas in \tref{tab:SpaceDimension}.

If $\caB$ is  the standard computational basis, then direct calculation yields
\begin{equation}
	\Lambda_1(\caB) = d^3 + 2d^2,\quad \Lambda_2(\caB) = d^3 - d^2,
\end{equation}
which implies \eref{eq:kappaClb} given \eref{eq:kappaCl} proved above.
\end{proof}

\subsection{The ensemble $\tbbU_k$  based on the simplest circuit}

By virtue of  \pref{pro:OmegaCl} we  can determine $\Omega(\tbbU_k)$ for the ensemble  $\tbbU_k=\Tensor{I}{(n-k)}\otimes\Tensor{(HT)}{k}\Cl_n$ (with $n\geq 1$ and $0\leq k\leq n$), which corresponds to the simplest circuit illustrated in \fref{fig:ModelCircuits}(b). This result  will be useful to proving \thref{thm:VartUk} and several related results. Actually, it is straightforward  to verify that $\Omega(\tbbU_k) = \Omega(\Cl_n, \caB_T)$, where $\caB_T=\{\Tensor{I}{(n-k)}\otimes\Tensor{(T^\dag H^\dag)}{k}|\bfb\>\}_{\bfb}$. According to the definitions in \eref{eq:DefLambda12}, $\Lambda_1(\caB_T)$ and $\Lambda_2(\caB_T)$ can be calculated as
\begin{equation}\label{eq:Lambda12tUk}
\Lambda_1(\caB_T) = d^3\gamma^k+2d^2\nu^k,\quad \Lambda_2(\caB_T) = d^3\gamma^k-d^2\nu^k,
\end{equation}
where $\gamma=3/4$ and $\nu=1/2$. In conjunction with \pref{pro:OmegaCl} we can deduce the following proposition.
\begin{proposition}\label{pro:OmegatUk}
Suppose $n\ge2$ and $0\leq k\leq n$. Then
\begin{align}
\Omega(\tbbU_k) &= \kappa_{[4]}^+(\tbbU_k)P_{[4]}^++\kappa_{[4]}^-(\tbbU_k)P_{[4]}^- + \kappa_{[2,2]}^+(\tbbU_k)P_{[2,2],G}^+ + \kappa_{[2,2]}^-(\tbbU_k)P_{[2,2],G}^- + \kappa_{[3,1]}(\tbbU_k)P_{[3,1],G},
\end{align}
where 
\begin{equation}\label{eq:kappatUk}
\begin{aligned}
\kappa_{[4]}^+(\tbbU_k)& = \frac{2d\gamma^k + 4\nu^k}{(d+1)(d+2)},&\quad \kappa_{[4]}^-(\tbbU_k) &= \frac{4d^2+4d\left(5-2\gamma^k\right)-16\nu^k}{(d^2-1)(d+2)(d+4)} ,\\
\kappa_{[2,2]}^+(\tbbU_k) &= \frac{2d\gamma^k-2\nu^k}{d^2-1},
&\quad \kappa_{[2,2]}^-(\tbbU_k) &= \frac{4d^2-4d\left(1+2\gamma^k\right)+8\nu^k}{(d^2-1)(d^2-4)},\\
\kappa_{[3,1]}(\tbbU_k) &= \frac{4}{(d+1)(d+2)}.
\end{aligned}
\end{equation}
\end{proposition}

\subsection{The ensemble $\bbU_{k,l}$ based on the interleaved Clifford circuit}
In this section we determine the cross moment operator associated with the unitary ensemble  $\bbU_{k, l}$, which corresponds to the interleaved  Clifford circuit illustrated in \fref{fig:ModelCircuits} (a).   
For the convenience of the following discussion, we first introduce two parameters, assuming that $n\ge k \ge1$. Define
\begin{align}\label{eq:alphakbetak}
\alpha_k := \frac{d^2(d+3)\left(d\gamma^k+3\nu^k\right)-4(d+1)(d+2)}{(d^2-1)(d+2)(d+4)},\quad \beta_k := \frac{d^2 \left(d^2 \gamma^k - 4 \right)+4}{(d^2 - 1) (d^2 - 4)},
\end{align}
where $\gamma=3/4$ and $\nu=1/2$, and, when $n=1$, the value of $\beta_1$ is understood as the limit $d\to 2$. Notably, we have
\begin{equation}\label{eq:alpha1beta1}
\alpha_1 = \frac{3d^2-3d-4}{4(d^2-1)}=\frac{3}{4}\left[1-\frac{3d+1}{3(d^2-1)}\right], \quad \beta_1 = \frac{3d^2-4}{4(d^2-1)}=\frac{3}{4}\left[1-\frac{1}{3(d^2-1)}\right].
\end{equation}
The basic properties of $\alpha_k$ and $\beta_k$ are summarized in the following lemma, which is proved in \aref{app:ProofDeviaCoeff}. 
\begin{lemma}\label{lem:PropertyAlphaBeta}
Suppose  $n\ge k \ge1$ and $l\ge0$. Then $\alpha_k$ and $\beta_k$ are monotonically increasing in $d$ and  monotonically decreasing in $k$. In addition,
\begin{gather}
\lim_{d\to\infty}\alpha_k=\lim_{d\to\infty}\beta_k=\gamma^k,\quad 
0<\alpha_k < \gamma^k,\quad 0<\beta_k < \gamma^k,\label{eq:UpperAlphaBeta}\\
\gamma^{kl}\left[1-\frac{kl(3d+1)}{3(d^2-1)}\right]\le\alpha_1^{kl} \le \alpha_k^l \le \beta_k^l \le \beta_1^{kl} \le \alpha_1^{kl} + \frac{kld}{d^2-1}\gamma^{kl}.\label{eq:RelationAlphaBeta}
\end{gather}
\end{lemma}

Before determining the cross moment operator associated with the ensemble  $\bbU_{k, l}$, we need to introduce three auxiliary lemmas. Let  $\hat{T}_k$ be a shorthand for $\Tensor{I}{(n-k)} \otimes \Tensor{T}{k}$. When $n\ge2$, define 
\begin{align}
\Delta_{[4],G}=\Delta_{[4]}:=\frac{ P_{[4]}^+}{ D_{[4]}^+}- \frac{ P_{[4]}^-}{ D_{[4]}^-},\quad  \Delta_{[2,2],G}:= \frac{ P_{[2,2], G}^+}{ D_{[2,2], G}^+}- \frac{ P_{[2,2], G}^-}{ D_{[2,2], G}^-}. 
\end{align}

\begin{lemma}\label{lem:TTwirlingP}
Suppose $n\ge2$ and $1\le k \le n$. Then
\begin{equation}
\dagtensor{\hat{T}_{1}}{4}P_{\lambda, G}^+\Tensor{\hat{T}_{1}}{4}  =P_{\lambda,G}^+ - P_{\lambda, G}^{+,(n-1)}\otimes \left(|\mathbf{0}\>\<\mathbf{1}|+|\mathbf{1}\>\<\mathbf{0}|\right),
\end{equation}
where $|\mathbf{0}\>$ and $|\mathbf{1}\>$ are shorthands for $|0000\>$ and $|1111\>$, and $P_{\lambda, G}^{+,(n-1)}$ is the counterpart of $P_{\lambda, G}^{+}$ associated with the first $(n-1)$-qubits.
\end{lemma}
\begin{lemma}\label{lem:UTTwirling}
Suppose $n\ge2$ and  $1\le k\le n$. Then
\begin{align}
\bbE_{U\sim \Cl_n} \dagtensor{U}{4} \dagtensor{\hat{T}_k}{4} P_{[4]}^+ \Tensor{\hat{T}_k}{4}\Tensor{U}{4}
&=\frac{4(1-\alpha_k)}{d(d+3)}P_{[4]}+\alpha_k P_{[4]}^+,\label{eq:UTTwirlingP4}\\
\bbE_{U\sim \Cl_n} \dagtensor{U}{4} \dagtensor{\hat{T}_k}{4} P_{[2,2],G}^+ \Tensor{\hat{T}_k}{4}\Tensor{U}{4} &= \frac{4(1-\beta_k)}{d^2}P_{[2,2], G} + \beta_k P_{[2,2] , G}^+.\label{eq:UTTwirlingP22}
\end{align}
\end{lemma}

\begin{lemma}\label{lem:UTTwirlingDelta}
Suppose $n\ge2$ and $1\le k \le n$. Then
\begin{gather}
\bbE_{U\sim \Cl_n} \dagtensor{U}{4} \dagtensor{\hat{T}_k}{4}\Delta_{[4],G} \Tensor{\hat{T}_k}{4}\Tensor{U}{4}=\alpha_k\Delta_{[4],G},\quad 
\bbE_{U\sim \Cl_n} \dagtensor{U}{4} \dagtensor{\hat{T}_k}{4}\Delta_{[2,2],G} \Tensor{\hat{T}_k}{4}\Tensor{U}{4}=\beta_k\Delta_{[2,2],G}.
\end{gather}
\end{lemma}
\Lref{lem:TTwirlingP} shows the effect of the twirling based on $\hat{T}_1$ and \lref{lem:UTTwirlingDelta}  is a simple corollary of \lref{lem:UTTwirling}. \Lref{lem:TTwirlingP} and \ref{lem:UTTwirling} are proved in \aref{app:ProofTTwirling}.

Now, we are ready to determine the cross moment operator associated  with the unitary ensemble  $\bbU_{k, l}$ and an orthonormal basis $\caB$. 
\begin{theorem}\label{thm:OmegaUkl}
Suppose $n\ge2$, $1\le k\le n$, $l\ge 0$, and $\caB$ is an orthonormal basis of $\caH$. Then
\begin{align}
\Omega(\bbU_{k,l},\caB) =&\;\kappa_{[4]}^+\left(\bbU_{k,l},\caB\right)P_{[4]}^+ +\kappa_{[4]}^-\left(\bbU_{k,l},\caB\right)P_{[4]}^- + \kappa_{[2,2]}^+\left(\bbU_{k,l},\caB\right)P_{[2,2],G}^+\nonumber\\
&+\kappa_{[2,2]}^-\left(\bbU_{k,l},\caB\right)P_{[2,2],G}^- + \kappa_{[3,1]}\left(\bbU_{k,l},\caB\right)P_{[3,1],G},  \label{eq:OmegaUkl}
\end{align}
where
\begin{equation}
\begin{aligned}\label{eq:kappaUkl}
\kappa_{[4]}^+\left(\bbU_{k,l},\caB\right) &=\frac{4(d+5)}{(d+1)(d+2)(d+3)}+\alpha_k^l\left[\kappa_{[4]}^+\left(\Cl_n,\caB\right)-\frac{4(d+5)}{(d+1)(d+2)(d+3)}\right],\\
\kappa_{[4]}^-\left(\bbU_{k,l},\caB\right)&=\frac{4(d+5)}{(d+1)(d+2)(d+3)}-\alpha_k^l\left[\frac{4(d+5)}{(d+1)(d+2)(d+3)}-\kappa_{[4]}^-\left(\Cl_n,\caB\right)\right],\\
\kappa_{[2,2]}^+\left(\bbU_{k,l},\caB\right)&= \frac{4}{d(d+1)} + \beta_k^l\left[\kappa_{[2,2]}^+\left(\Cl_n,\caB\right)-\frac{4}{d(d+1)}\right],\\
\kappa_{[2,2]}^-\left(\bbU_{k,l},\caB\right)&=\frac{4}{d(d+1)}-\beta_k^l\left[\frac{4}{d(d+1)}-\kappa_{[2,2]}^-\left(\Cl_n,\caB\right)\right],\\
\kappa_{[3,1]}\left(\bbU_{k,l},\caB\right)& = \frac{4}{(d+1)(d+2)}.
\end{aligned}
\end{equation}
If $\caB$ is the standard computational basis, then
\begin{equation}\label{eq:kappaUklb}
\begin{aligned}
\kappa_{[4]}^+ (\bbU_{k, l} ,\caB) &= \frac{4(d+5)+2(d-1)(d+4)\alpha_k^l}{(d+1)(d+2)(d+3)},&\quad \kappa_{[4]}^- (\bbU_{k, l}, \caB) &= \frac{4(d+5)-8\alpha_k^l}{(d+1)(d+2)(d+3)},\\
\kappa_{[2,2]}^+ (\bbU_{k, l}, \caB) &= \frac{4+2(d-2)\beta_k^l}{d(d+1)},&\quad \kappa_{[2,2]}^- (\bbU_{k, l}, \caB) &= \frac{4(d+2)-8\beta_k^l}{d(d+1)(d+2)}.
\end{aligned}
\end{equation}
\end{theorem}

\begin{proof}
\Eref{eq:OmegaUkl}
follows from \lref{lem:RepV4}, Schur's lemma, \tref{tab:SpaceDimension}, and \eref{eq:Defkappa}, just like \eref{eq:OmegaCl}.
The formula for $\kappa_{[3,1]}\left(\bbU_{k,l},\caB\right)$ in  \eref{eq:kappaUkl} follows from \eref{eq:kappaHaar}. In conjunction with \lref{lem:OmegaPlambdaG} we can deduce that the operator
$P_{\lambda,G}[\Omega(\bbU_{k,l},\caB)-\Omega(\haar)] P_{\lambda,G}$ is traceless and proportional to $\Delta_{\lambda,G}$ for $\lambda=[4], \; [2,2]$. When $l=0$, the ensemble $\bbU_{k,l}$ reduces to the Clifford ensemble, in which case \eqsref{eq:kappaUkl}{eq:kappaUklb} follow from \eqsref{eq:kappaCl}{eq:kappaClb}, respectively. 
When $l\geq 1$, by virtue of \lref{lem:UTTwirlingDelta} and  \eref{eq:Defkappa}  we can deduce that
\begin{align}
P_{\lambda,G}[\Omega(\bbU_{k,l},\caB)-\Omega(\haar)] P_{\lambda,G}&= P_{\lambda,G}[\Omega(\bbU_{k,l-1},\caB)-\Omega(\haar)] P_{\lambda,G}\times \begin{cases}
\alpha_k & \lambda=[4],\\
\beta_k & \lambda=[2,2], 
\end{cases}\\
\kappa_\lambda^\pm(\bbU_{k,l},\caB)-\kappa(\haar)&= [\kappa_\lambda^\pm(\bbU_{k,l-1},\caB)-\kappa(\haar)]\times \begin{cases}
\alpha_k & \lambda=[4],\\
\beta_k & \lambda=[2,2],
\end{cases}
\end{align}
which imply \eqsref{eq:kappaUkl}{eq:kappaUklb} by induction. 
\end{proof}

When $l=0$, \thref{thm:OmegaUkl} recovers \pref{pro:OmegaCl}. In general, the $1$-norm of the difference $\Omega(\bbU_{k,l},\caB)-\Omega(\haar)$ can be calculated as follows,
\begin{align}
\left\|\Omega(\bbU_{k,l},\caB) - \Omega(\haar)\right\|_1 &=\frac{ (d^2 - 1) (d + 2) (d + 4)\alpha_k^l }{3 d (d + 3)} \left|\kappa_{[4]}^+\left(\Cl_n,\caB\right)-\kappa_{[4]}^-\left(\Cl_n,\caB\right)\right|\nonumber\\
&\quad +  \frac{2 (d^2 - 4) (d^2 - 1)\beta_k^l}{3 d^2} \left|\kappa_{[2,2]}^+\left(\Cl_n,\caB\right)-\kappa_{[2,2]}^-\left(\Cl_n,\caB\right)\right|\nonumber\\
&=\left\|\Omega(\Cl_n, \caB)-\Omega(\haar)\right\|_1\left[\gamma^{kl}+\caO(d^{-1})\right].
\end{align}
It is mainly determined by the total number of $T$ gates employed to generate the ensemble $\bbU_{k,l}$, which echoes \thsref{thm:VarUkl} and \ref{thm:VarFUkl} in the main text. Therefore, two ensembles $\bbU_{k,l}$ and $\bbU_{k',l'}$ will have similar performances in thrifty shadow estimation whenever $kl=k'l'$. Notably, the ensemble $\bbU_{k,1}$
is similar to $\bbU_{1,k}$ with this regard although it is much easier to realize.

\subsubsection{Proof of \lref{lem:PropertyAlphaBeta}}\label{app:ProofDeviaCoeff}
\begin{proof}
According to the definitions in \eref{eq:alphakbetak} it is easy to verify that $\alpha_k$ and $\beta_k$ are monotonically decreasing in $k$, given that $0<\nu<\gamma < 1$. In addition, a direct calculation shows that $\partial \alpha_k/\partial d=f_1/[(d^2-1)^2(d+2)^2(d+4)^2]$ with  
\begin{align}
f_1&=3 d^6 \left(\gamma ^k-\nu ^k\right)+2 d^5 \left(7 \gamma ^k-9 \nu ^k+4\right)+3 d^4 \left(\gamma ^k-11 \nu ^k+20\right)+4 d^3 \left(44 - 17 \gamma ^k-9 \nu ^k\right)\nonumber\\
&\quad +18 d^2 \left(14 - 4 \gamma ^k - 7\nu ^k\right)+16 d \left(11 -9 \nu ^k\right)+48> 0,  
\end{align}
which implies that $\alpha_k$ is  monotonically increasing in $d$. When $k=1$, by virtue of \eref{eq:alpha1beta1}, 
it is straightforward to prove that $\beta_k$ is monotonically increasing in $d$. When $k\ge2$, which means $d\ge4$, by virtue of \eref{eq:alphakbetak} we can deduce that
\begin{align}
 \frac{\partial \beta_k}{\partial d}=\frac{2d[4(d^2-1)^2-d^2\gamma^k(5d^2-8)]}{(d^2-1)^2(d^2-4)^2} >  \frac{d(d^2-4)^2}{2(d^2-1)^2(d^2-4)^2}>0, 
\end{align}
where the first inequality holds because $5d^2-8>0$ and $\gamma^k< 3/4$ given that $\gamma=3/4$ and $k\geq 2$. Therefore, $\beta_k$ is monotonically increasing in $d$.

By definitions in \eref{eq:alphakbetak} it is straightforward to verify the limits in 
\eref{eq:UpperAlphaBeta}, which imply the inequalities $\alpha_k,\beta_k<\gamma^k$, given that $\alpha_k$ and $\beta_k$ are  monotonically increasing in $d$.  The inequalities  $\alpha_k,\beta_k>0$ also follow from the definitions in \eqsref{eq:alphakbetak}{eq:alpha1beta1} given that $k\leq n$ and $d\gamma^k> d\nu^k\geq 1$. This observation completes the proof of \eref{eq:UpperAlphaBeta}, and it remains to prove \eref{eq:RelationAlphaBeta}.

By virtue of \eref{eq:alpha1beta1} we can deduce that
\begin{equation}\label{eq:LowerAlpha1}
\alpha_1^{kl} = \gamma^{kl} \left[1-\frac{3d+1}{3(d^2-1)}\right]^{kl} \ge \gamma^{kl} \left[1-\frac{kl(3d+1)}{3(d^2-1)}\right],
\end{equation}
which confirms the first inequality in \eref{eq:RelationAlphaBeta}. Here the inequality holds because  $(1-x)^q \ge 1-qx$ when $x\in[0,1]$  and $q\ge0$.

To prove  the second
inequality in \eref{eq:RelationAlphaBeta}, it suffices to prove the inequality $\alpha_1^k \le \alpha_k$, which is trivial when $k=1$.  When $k=2$ (which means $d\geq 4$), we have
\begin{equation}
\alpha_k - \alpha_1^k=\frac{d(3d^3+2d^2-29d-24)}{16(d+2)(d^2-1)}>0. 
\end{equation}
When  $k \ge3$ (which means $d\geq 8$), by virtue of the fact that $(1-x)^q \le 1 - qx + q(q-1)x^2/2$ for $x\in[0, 1]$ and $q\ge1$ we can deduce that
\begin{equation}
\alpha_1^{k}= \gamma^k \left[1-\frac{3d+1}{3(d^2-1)}\right]^{k} \le \gamma^{k}
\left[1-\frac{k(3d+1)}{3(d^2-1)}+ \frac{k(k-1)(3d+1)^2}{18(d^2-1)^2}\right]\le \gamma^{k} \left[1-\frac{3kd}{4(d^2-1)}\right],
\end{equation}
where the second inequality holds because 
\begin{align}\label{eq:TechiqueUpper}
\frac{(k-1)(3d+1)}{18(d^2-1)} \le \frac{(n-1)(3d+1)}{18(d^2-1)} \le \frac{1}{12}.
\end{align}
Therefore,
\begin{align}
\alpha_k-\alpha_1^k &\ge \alpha_k-\gamma^k \left[1-\frac{3kd}{4(d^2-1)}\right]=\frac{f_2}{4 (d^2-1) (d+2) (d+4)},
\end{align}
where 
\begin{align}
f_2=3 d^3 \left[(k-4) \gamma ^k+4 \nu ^k\right]+2 d^2 \left[(9 k-14) \gamma ^k+18 \nu ^k-8\right]+24 d \left[(k+1) \gamma ^k-2\right]+32 \left(\gamma ^k-1\right). 
\end{align}
If in addition $k\geq 5$ (which means $n\geq 5$ and $d\geq 32$),  then we have $d\gamma^k \ge d\gamma^n \ge (3/2)^5 > 7$, which implies that $f_2 > 21d^2-16d^2-48d-32\geq 0$; if instead $k=3,4$, then $f_2\geq 0$ by direct calculations. In any case, we have $\alpha_k-\alpha_1^k\geq0$, which implies the second inequality in \eref{eq:RelationAlphaBeta}.

To prove  the third
inequality in \eref{eq:RelationAlphaBeta}, it suffices to prove the  inequality $\alpha_k \le \beta_k$. When $k=1$, this inequality follows from \eref{eq:alpha1beta1}. 
When $k\geq 2$, by definitions in \eref{eq:alphakbetak} we can deduce that
\begin{align}
\beta_k - \alpha_k &= \frac{3 d \left[d^3 \left(\gamma^k-\nu^k\right)+d^2 \left(2 \gamma^k-\nu^k\right)+d \left(6 \nu ^k - 4\right)-4\right]}{(d^2-1)(d^2-4)(d+4)}\ge\frac{3 d (d^2-d+2)}{(d^2-1)(d^2-4)(d+4)}\ge0,
\end{align}
which implies the third
inequality in \eref{eq:RelationAlphaBeta}.
Here the second inequality holds because $\gamma^k\geq 2\nu^k>0$ and $d\nu^k\geq d\nu^n=1$.

To prove the fourth inequality in \eref{eq:RelationAlphaBeta}, it suffices to prove the inequality $\beta_k \le \beta_1^k$. By virtue of \eref{eq:alpha1beta1} we can deduce that
\begin{equation}
\beta_1^k = \gamma^{k} \left[1-\frac{1}{3(d^2-1)}\right]^k \ge \gamma^{k}\left[1-\frac{k}{3(d^2-1)}\right].
\end{equation}
Therefore,  
\begin{equation}
\beta_1^k - \beta_k \ge \gamma^{k}\left[1-\frac{k}{3(d^2-1)}\right]-\beta_k=\frac{12 d^2-12 -\left[k(d^2-4) +15 d^2-12\right] \gamma ^k}{3 (d^2-4)(d^2-1)}\ge0,
\end{equation}
which implies the fourth inequality in \eref{eq:RelationAlphaBeta}. Here
 the second inequality holds because $[k(d^2-4) +15 d^2-12]\gamma ^k$ decreases monotonically in $k$ when $d\ge 2$.

Finally, by virtue of \eref{eq:alpha1beta1} we can deduce that
\begin{equation}
\beta_1^{kl} - \alpha_1^{kl} = \gamma^{kl} \left[\left(1-\frac{1}{3(d^2-1)}\right)^{kl}-\left(1-\frac{3d+1}{3(d^2-1)}\right)^{kl}\right] \le \frac{kld}{d^2-1}\gamma^{kl},
\end{equation}
which confirms the last inequality in  \eref{eq:RelationAlphaBeta} and completes the proof of \lref{lem:PropertyAlphaBeta}.
\end{proof}

\subsubsection{Proofs of \lsref{lem:TTwirlingP} and~\ref{lem:UTTwirling}}\label{app:ProofTTwirling}
\begin{proof}[Proof of \lref{lem:TTwirlingP}]
By virtue of \eref{eq:FormRT4} we can determine the conjugation action of $\Tensor{\hat{T}_{1}}{4}$ on the stabilizer projector $P_n$, with the result
\begin{align}
\dagtensor{\hat{T}_{1}}{4} P_n \Tensor{\hat{T}_{1}}{4} = P_n - P_{n-1}\otimes \left(|\mathbf{0}\>\<\mathbf{1}|+|\mathbf{1}\>\<\mathbf{0}|\right),
\end{align}
where $|\mathbf{0}\>$ and $|\mathbf{1}\>$ are shorthands for $|0000\>$ and $|1111\>$, respectively. As a corollary, 
\begin{align}
\dagtensor{\hat{T}_{1}}{4}P_{\lambda, G}^+\Tensor{\hat{T}_{1}}{4}  
&=P_{\lambda, G}\dagtensor{\hat{T}_{1}}{4}P_n\Tensor{\hat{T}_{1}}{4}
=P_{\lambda, G} \left[P_n - P_{n-1}\otimes \left(|\mathbf{0}\>\<\mathbf{1}|+|\mathbf{1}\>\<\mathbf{0}|\right)\right]\nonumber\\
&=P_{\lambda, G}^+ - \sum_{\sigma\in S_4} c_{\lambda, G}(\sigma) R(\sigma) \left[P_{n-1}\otimes \left(|\mathbf{0}\>\<\mathbf{1}|+|\mathbf{1}\>\<\mathbf{0}|\right)\right]\nonumber\\
&=P_{\lambda, G}^+ - \sum_{\sigma\in S_4} \bigl\{ c_{\lambda, G}(\sigma)\left[  \Tensor{r(\sigma)}{n-1}P_{n-1}\right]\otimes \left[r(\sigma)\left(|\mathbf{0}\>\<\mathbf{1}|+|\mathbf{1}\>\<\mathbf{0}|\right)\right]\bigr\}\nonumber\\
&=P_{\lambda, G}^+ - \left[\sum_{\sigma\in S_4}  c_{\lambda, G}(\sigma)\Tensor{r(\sigma)}{n-1}P_{n-1}\right]\otimes \left(|\mathbf{0}\>\<\mathbf{1}|+|\mathbf{1}\>\<\mathbf{0}|\right)\nonumber\\
&= P_{\lambda,G}^+ - P_{\lambda, G}^{+,(n-1)}\otimes \left(|\mathbf{0}\>\<\mathbf{1}|+|\mathbf{1}\>\<\mathbf{0}|\right),
\end{align}
where $P_{\lambda, G} = \sum_{\sigma\in S_4} c_{\lambda, G}(\sigma) R(\sigma)$ and $c_{\lambda, G}(\sigma)$ is independent of the qubit number $n$ by definitions in Eqs.~(\ref{eq:FormPlambda}),~(\ref{eq:PGHG}), and~(\ref{eq:DefProjectors}). The first equality holds because $P_{\lambda, G}^+ = P_{\lambda, G}P_n$ and $\dagtensor{\hat{T}_{1}}{4} P_{\lambda, G}\Tensor{\hat{T}_1}{4}=P_{\lambda, G}$; the fifth equality follows from the fact that $r(\sigma)|\mathbf{0}\>=|\mathbf{0}\>$ and $r(\sigma)|\mathbf{1}\>=|\mathbf{1}\>$ for any $\sigma\in S_4$, where $r(\sigma)$ is defined in \eref{eq:DefRT}.
\end{proof}

\begin{proof}[Proof of \lref{lem:UTTwirling}]
By virtue of  \lref{lem:TTwirlingP} we can determine the action  of $\hat{T}_k$ on the projector $P_{\lambda,G}^+$, with the result
\begin{equation}\label{eq:TkTwirlingP}
\dagtensor{\hat{T}_{k}}{4}P_{\lambda,G}^+\Tensor{\hat{T}_{k}}{4} =\sum_{j=0}^{k} (-1)^j\binom{k}{j} P_{\lambda,G}^{+, (n-j)} \otimes \Tensor{\left(|\mathbf{0}\>\<\mathbf{1}|+|\mathbf{1}\>\<\mathbf{0}|\right)}{j}.
\end{equation}
Note that $P_{\lambda,G}^{+, (n-j)}$ may project onto different $(n-j)$-qubit subsystems. Since the precise partitions of these projectors are not required in the following analysis, we denote them by the same notation for the sake of simplicity. By virtue of \lref{lem:RepV4} and Schur's lemma, for $\lambda$ such that $D_{\lambda, G}^\pm \neq 0$, we have
\begin{equation}\label{eq:TTwirlingP+}
\bbE_{U\sim \Cl_n} \dagtensor{U}{4} \dagtensor{\hat{T}_k}{4} P_{\lambda, G}^+ \Tensor{\hat{T}_k}{4}\Tensor{U}{4} =\sum_{s=\pm} \frac{\tr\left[\dagtensor{\hat{T}_k}{4} P_{\lambda, G}^+ \Tensor{\hat{T}_k}{4} P_{\lambda, G}^s\right]}{D_{\lambda, G}^s} P_{\lambda, G}^s.
\end{equation}
The overlaps involved can be calculated as follows,
\begin{align}
\tr\left[\dagtensor{\hat{T}_k}{4} P_{\lambda, G}^+ \Tensor{\hat{T}_k}{4} P_{\lambda, G}^+\right] &= \tr\left[\dagtensor{\hat{T}_k}{4} P_{\lambda, G}^+ \Tensor{\hat{T}_k}{4} P_{\lambda, G}P_n\right] = \tr\left[\dagtensor{\hat{T}_k}{4} P_{\lambda, G}^+ \Tensor{\hat{T}_k}{4} P_n\right] \nonumber\\
&=  \sum_{j=0}^{k} (-1)^j\binom{k}{j} D_{\lambda, G}^{+,(n-j)} \tr\left[ P_1 \left(|\mathbf{0}\>\<\mathbf{1}|+|\mathbf{1}\>\<\mathbf{0}|\right)\right]^j=\sum_{j=0}^{k} (-1)^j\binom{k}{j} D_{\lambda, G}^{+,(n-j)},\label{eq:trTPTP}\\
\tr\left[\dagtensor{\hat{T}_k}{4} P_{\lambda, G}^+ \Tensor{\hat{T}_k}{4} P_{\lambda, G}^-\right] &= D_{\lambda, G}^{+} - \tr\left[\dagtensor{\hat{T}_k}{4} P_{\lambda, G}^+ \Tensor{\hat{T}_k}{4} P_{\lambda, G}^+\right],
\end{align}
where $D_{\lambda, G}^{+,(n-j)}$ is the dimension of $\caH_{\lambda, G}^{(n-j)}$, which is associated  with $n-j$ qubits. The second equality in \eref{eq:trTPTP} follows from the equality $P_{\lambda, G}^+P_{\lambda, G} = P_{\lambda, G}^+$ and \eref{eq:UPlambdaG}; the last two equalities in \eref{eq:trTPTP} follow from  \eref{eq:TkTwirlingP} and the fact that $\tr\left[ P_1 \left(|\mathbf{0}\>\<\mathbf{1}|+|\mathbf{1}\>\<\mathbf{0}|\right)\right]=1$.
In conjunction with the dimension formulas in \tref{tab:SpaceDimension} we can deduce  that
\begin{align}
&\tr\left[\dagtensor{\hat{T}_k}{4} P_{[4]}^+ \Tensor{\hat{T}_k}{4} P_{[4]}^+\right] = \frac{d^2\gamma^k+3 d\nu^k}{6},&\quad& \tr\left[\dagtensor{\hat{T}_k}{4} P_{[4]}^+ \Tensor{\hat{T}_k}{4} P_{[4]}^-\right] = \frac{d^2\left(1-\gamma^k\right)+3 d\left(1-\nu^k\right)+2}{6}, \\
&\tr\left[\dagtensor{\hat{T}_k}{4} P_{[2,2], G}^+ \Tensor{\hat{T}_k}{4} P_{[2,2], G}^+\right] = \frac{d^2\gamma^k}{3},&\quad& \tr\left[\dagtensor{\hat{T}_k}{4} P_{[2,2], G}^+ \Tensor{\hat{T}_k}{4} P_{[2,2], G}^-\right] = \frac{d^2 \left( 1 - \gamma^k \right) - 1}{3}.
\end{align}
Together with \eref{eq:TTwirlingP+} and  \tref{tab:SpaceDimension}, the two equations means
\begin{align}
\bbE_{U\sim \Cl_n} \dagtensor{U}{4} \dagtensor{\hat{T}_k}{4} P_{[4]}^+ \Tensor{\hat{T}_k}{4}\Tensor{U}{4} &= \frac{d^2\gamma^k+3 d\nu^k}{(d+1)(d+2)} P_{[4]}^+ + \frac{4\left[d^2\left(1-\gamma^k\right)+3 d\left(1-\nu^k\right)+2\right]}{(d^2-1)(d+2)(d+4)} P_{[4]}^-,\\
\bbE_{U\sim \Cl_n} \dagtensor{U}{4} \dagtensor{\hat{T}_k}{4} P_{[2,2],G}^+ \Tensor{\hat{T}_k}{4}\Tensor{U}{4} &=\frac{d^2\gamma^k}{d^2-1} P_{[2,2], G}^+ + \frac{4\left[d^2 \left( 1 - \gamma^k \right) - 1\right]}{(d^2-1)(d^2-4)} P_{[2,2], G}^-,
\end{align}
which imply  \eqsref{eq:UTTwirlingP4}{eq:UTTwirlingP22} and complete the proof of \lref{lem:UTTwirling}.
\end{proof}

\subsection{Analysis for the single-qubit case}\label{app:SingleCase}
In this section, we determine the cross moment operator for the single-qubit case, that is $n=1$. Now, special attention is required because  $D_{[2,2],G}^-=0$ and $P_{[2,2], G} = P_{[2,2], G}^+$, which means $\Lambda_2(\caB) = 4$ for any orthonormal basis $\caB$ on $\caH$. Nevertheless,  the basic idea is similar  to the case $n\ge2$.
So we shall state the main results and omit the detailed proofs. 

\begin{proposition}\label{pro:OmegaClSingle}
Suppose $n=1$ and $\caB$ is any orthonormal basis of $\caH$.
Then 
\begin{align}
\Omega(\Cl_1, \caB) &= \kappa_{[4]}^+(\Cl_1,\caB)P_{[4]}^++\kappa_{[4]}^-(\Cl_1,\caB)P_{[4]}^- + \kappa_{[2,2]}(\Cl_1,\caB)P_{[2,2],G} + \kappa_{[3,1]}(\Cl_1,\caB)P_{[3,1],G},
\end{align}
where 
\begin{equation}
\kappa_{[4]}^+(\Cl_1,\caB) = \frac{\Lambda_1(\caB)}{24},\quad\kappa_{[4]}^-(\Cl_1,\caB) = \frac{28-\Lambda_1(\caB)}{36} ,\quad\kappa_{[2,2]}(\Cl_1,\caB) = \frac{2}{3},
\quad\kappa_{[3,1]}(\Cl_1,\caB) = \frac{1}{3}.
\end{equation}
If $\caB$ is the standard computational basis, then
\begin{equation}
\kappa_{[4]}^+(\Cl_1,\caB) = \kappa_{[2,2]}(\Cl_1,\caB) = \frac{2}{3},\quad\kappa_{[4]}^-(\Cl_1,\caB) = \kappa_{[3,1]}(\Cl_1,\caB) = \frac{1}{3}.
\end{equation}
\end{proposition}

\begin{proposition}
Suppose $n=1$, $l\ge0$, and $\caB$ is any orthonormal basis of $\caH$. Then
\begin{align}
\Omega(\bbU_{1,l},\caB) =\kappa_{[4]}^+\left(\bbU_{1,l},\caB\right)P_{[4]}^+ +\kappa_{[4]}^-\left(\bbU_{1,l},\caB\right)P_{[4]}^- + \kappa_{[2,2]}\left(\bbU_{1,l},\caB\right)P_{[2,2],G}+ \kappa_{[3,1]}\left(\bbU_{1,l},\caB\right)P_{[3,1],G},
\end{align}
where
\begin{equation}
\begin{aligned}
\kappa_{[4]}^+\left(\bbU_{1,l},\caB\right) &= \frac{7}{15}+\left(\frac{1}{6}\right)^l\left[\frac{\Lambda_1(\caB)}{24}-\frac{7}{15}\right],&\quad\kappa_{[4]}^-\left(\bbU_{1,l},\caB\right) &= \frac{7}{15}-\left(\frac{1}{6}\right)^l\left[\frac{\Lambda_1(\caB)}{36}-\frac{14}{45}\right],\\
\kappa_{[2,2]}\left(\bbU_{1,l},\caB\right) &= \frac{2}{3},&\quad \kappa_{[3,1]}\left(\bbU_{1,l},\caB\right) &= \frac{1}{3}.
\end{aligned}
\end{equation}
\end{proposition}
Note that for the unitary ensemble $\bbU_{k, l}$, the only nontrivial choice for $k$ is 1, so we can focus on  $\bbU_{1, l}$ without loss of generality. In addition, combining \eref{eq:Lambda12tUk} with \pref{pro:OmegaClSingle} we can obtain the counterpart of \pref{pro:OmegatUk} for the special case $n=1$.

\subsection{Alternative expression of the cross moment operator in terms of $\caR_i$}
According to \aref{app:ProjectorSWD}, the projectors $P_{\lambda, G}$ and $P_{\lambda, G}^s$ can be expanded in terms of the operators  $\caR_i$ defined in \eref{eq:DefcaRi}. Suppose  $\caU$ corresponds to one of the following four unitary ensembles: Haar, $\Cl_n$, $\bbU_{k,l}$, and $\tbbU_k$.
Then we can expand the cross moment operator $\Omega(\caU, \caB)$
in terms of $\caR_i$ as follows,
\begin{equation}
\Omega(\caU, \caB) = \sum_{i=1}^{5} g_i(\caU, \caB) \caR_i.
\end{equation}
When $n\geq 2$, the set  $\{\caR_i\}_{i=1}^5$ is linearly independent, so the coefficients $g_i(\caU, \caB)$ are uniquely  determined by $\caU$ and $\caB$. In the special case $n=1$, there is some freedom in choosing the coefficients.
In conjunction with \eref{eq:DefV*2Design}, we can express the variance  $V_*(O,\rho)$ in thrifty shadow estimation as follows,
\begin{equation}\label{eq:DefV*gi}
V_*(O,\rho) = (d+1)^2\sum_{i=1}^5 g_i(\caU,\caB)\tr\bigl[\caR_i\Tensor{(O\otimes \rho)}{2}\bigr] - [\tr(O\rho)]^2.
\end{equation}
When $\caB$ is the standard computational basis, which is our main focus, the coefficients $g_i(\caU, \caB)$ can be abbreviated as $g_i(\caU)$. \Lsref{lem:giHaar}-\ref{lem:gitUk} below determine these coefficients associated with the four  unitary ensembles mentioned above.
Note that these results hold whenever $n\geq 1$, although some coefficients are not unique when $n=1$.

\begin{lemma}\label{lem:giHaar}
Suppose  $n\geq 1$. Then
\begin{equation}\label{eq:giHaar}
\begin{aligned}
g_1(\haar) &= \frac{(d^2+4d+2)}{d(d+1)(d+2)(d+3)},&\quad& g_2(\haar) = - \frac{1}{d(d+1)(d+2)(d+3)},\\
g_3(\haar) &= \frac{1}{d(d+1)(d+3)}, &\quad& g_4(\haar) = g_5(\haar) = 0. 
\end{aligned}
\end{equation}
\end{lemma}

\begin{lemma}\label{lem:giCl}
Suppose $n\geq 1$. Then 
\begin{align}\label{eq:giCl}
g_1(\Cl_n) = g_4(\Cl_n)=\frac{1}{(d+1)(d+2)},\quad g_2(\Cl_n) = g_3(\Cl_n) = g_5(\Cl_n) =0.
\end{align}
\end{lemma}

\begin{lemma}\label{lem:giUkl}
Suppose  $n\ge k\ge 1$ and $l\ge 0$. Then
\begin{equation}
\begin{aligned}
g_1(\bbU_{k, l }) &= \frac{d^2+4d+2}{d(d+1)(d+2)(d+3)}-\frac{\alpha_k^l}{3(d+1)(d+2)(d+3)}-\frac{2\beta_k^l}{3d(d+1)(d+2)},\\
g_2(\bbU_{k, l }) &=-\frac{1}{d (d+1) (d+2) (d+3)} -\frac{\alpha_k^l}{3(d+1) (d+2) (d+3)}+\frac{\beta_k^l}{3d (d+1) (d+2)},\\
g_3(\bbU_{k,l}) &=\frac{1}{d (d+1) (d+3)} - \frac{\alpha_k^l}{3(d+1)(d+2)(d+3)}-\frac{2\beta_k^l}{3d(d+1)(d+2)},\\
g_4(\bbU_{k, l}) &=\frac{\alpha_k^l}{3(d+1)(d+2)}+\frac{2\beta_k^l}{3(d+1)(d+2)}, \quad g_5(\bbU_{k, l}) = \frac{\alpha_k^l}{3(d+1)(d+2)}-\frac{\beta_k^l}{3(d+1)(d+2)},
\end{aligned}
\end{equation}
where $\alpha_k,\beta_k$ are defined in \eref{eq:alphakbetak}.
\end{lemma}

\begin{lemma}\label{lem:gitUk}
Suppose $n\ge k\ge 2$. Then
\begin{equation}\label{eq:gitUk}
\begin{aligned}
g_1(\tbbU_k) &= \frac{(-d^2-2 d) \gamma^k+4\nu^k+d^3+2 d^2-8 d+4}{(d^2-4) (d^2-1)(d+4)},\quad g_2(\tbbU_k) = \frac{d \left(2\gamma^k-1-\nu^k\right)}{(d^2-4) (d^2-1) (d+4)},\\
g_3(\tbbU_k) &=\frac{\left(d^2+2 d\right) \left(1-\gamma^k\right)-4\left(1-\nu^k\right)}{(d^2-4) (d^2-1) (d+4)},\\
g_4(\tbbU_k) &=\frac{d\left(d^2+3 d-2\right) \gamma^k-(2d+4)\nu^k-2 (d^2+3d-6)}{(d^2-4) (d^2-1) (d+4)},\\
g_5(\tbbU_k) &= -\frac{(d^2+2d)\gamma^k-(d^2+2d-4)\nu^k-4}{(d^2-4) (d^2-1) (d+4)},
\end{aligned}
\end{equation}
where $\gamma=3/4$ and $\nu=1/2$.
If instead $n\geq k=1$, then
\begin{equation}\label{eq:gitU1}
\begin{aligned}
g_1(\tbbU_1) &= \frac{4 d-3}{4 (d-1) (d+1) (d+2)},\quad g_2(\tbbU_1) = 0, \quad g_3(\tbbU_1)=\frac{1}{4 (d-1) (d+1) (d+2)},\\
g_4(\tbbU_1) &= \frac{3 d-5}{4 (d-1) (d+1) (d+2)},\quad g_5(\tbbU_1) = -\frac{1}{4 (d-1) (d+1) (d+2)}.
\end{aligned}
\end{equation}
\end{lemma}
\Eref{eq:gitU1} can also be regarded as a special case of \eref{eq:gitUk} [when $n=1$ we need to take a proper limit].

Next, we provide informative bounds for $\tr\left[\caR_i \Tensor{(O\otimes \rho)}{2}\right]$ with $i=1,2,3,4,5$, which will be very useful to studying the variance $V_*(O,\rho)$ in view of \eref{eq:DefV*gi}.
\begin{lemma}\label{lem:TrRiOrho}
Suppose $\rho\in\caD(\caH)$ and $O\in \caL_0^\rmH (\caH)$. Then $\tr\left[\caR_i \Tensor{(O\otimes \rho)}{2}\right]\ge0$ for $i=1,2,3,4$. Moreover,
\begin{equation}\label{eq:TrRiOrho}
\begin{aligned}
\tr\left[\caR_1 \Tensor{(O\otimes \rho)}{2}\right] & =[\tr(O\rho)]^2\le \|O\|_\infty^2 \le \frac{d-1}{d}\|O\|_2^2,\\
\tr\left[\caR_2 \Tensor{(O\otimes \rho)}{2}\right] &= \tr(O^2)+4\tr(O^2\rho) +2\tr(O^2\rho^2)+2\tr(O\rho O\rho)\le \|O\|_2^2+8\|O\|_\infty^2\le \frac{9d-8}{d}\|O\|_2^2,\\
\tr\left[\caR_3 \Tensor{(O\otimes \rho)}{2}\right] & = \tr(O^2)\tr(\rho^2)+ [\tr(O\rho)]^2 + 2\tr(O^2\rho^2)\le \|O\|_2^2+3\|O\|_\infty^2\le \frac{4d-3}{d}\|O\|_2^2,\\
\tr\left[\caR_4 \Tensor{(O\otimes \rho)}{2}\right] &=\frac{1}{d}\left( \|\Xi_{\rho, O}\|_2^2 + \tXi_{\rho, O}\cdot\Xi_{\rho, O}\right) =  \frac{d+2}{d+1}V_\triangle(O,\rho)\le 2\|O\|_2^2,\\
\left|\tr\left[\caR_5 \Tensor{(O\otimes \rho)}{2}\right]\right| &\le 4\|O\|_2^2.\\
\end{aligned}    
\end{equation}
\end{lemma}

\begin{lemma}\label{lem:TrRiOrhoF}
Suppose $\rho \in \caD(\caH)$ and $O=|\phi\>\<\phi|-\bbone/d$ with $|\phi\>\in\caH$. Then
\begin{equation}\label{eq:TrRiOrhoF}
\begin{aligned}
\tr\left[\caR_1 \Tensor{(O\otimes \rho)}{2}\right] & =\left(F-\frac{1}{d}\right)^2,\\
\tr\left[\caR_2 \Tensor{(O\otimes \rho)}{2}\right] &= \frac{d^2-d+4}{d^2} + 2F^2 + 4\left(1-\frac{2}{d}\right)F+2\left(1-\frac{4}{d}\right)\<\phi|\rho^2|\phi\> + \frac{4}{d^2}\wp(\rho),\\
\tr\left[\caR_3 \Tensor{(O\otimes \rho)}{2}\right] & = \left(F-\frac{1}{d}\right)^2 + 2\left(1-\frac{2}{d}\right)\<\phi|\rho^2|\phi\> + \frac{d^2-d+2}{d^2} \wp(\rho),\\
\tr\left[\caR_4 \Tensor{(O\otimes \rho)}{2}\right] & = \frac{1}{d}\left(\|\Xi_{\rho, \phi}\|_2^2 + \tXi_{\rho, \phi}\cdot\Xi_{\rho, \phi}-1-2F+\frac{1}{d}\right)< \frac{2}{d}\|\Xi_\phi^2\|_{[d]},
\end{aligned}
\end{equation}
where $F = \<\phi|\rho|\phi\>$ and $\wp(\rho) = \tr(\rho^2)$. If in addition $\rho = |\phi\>\<\phi|$, then
\begin{equation}\label{eq:TrRiOphiF}
\begin{aligned}
&\tr\left[\caR_1 \Tensor{(O\otimes \rho)}{2}\right]  =\frac{(d-1)^2}{d^2},&\quad&\tr\left[\caR_2 \Tensor{(O\otimes \rho)}{2}\right] = \frac{(d-1)(9d-8)}{d^2},\\
&\tr\left[\caR_3 \Tensor{(O\otimes \rho)}{2}\right] = \frac{(d-1)(4d-3)}{d^2},&\quad& \tr\left[\caR_4 \Tensor{(O\otimes \rho)}{2}\right] = \frac{2^{1-M_2(\phi)}d^2-3d+1}{d^2},\\
&\tr\left[\caR_5 \Tensor{(O\otimes \rho)}{2}\right] = \frac{2^{2-M_2(\phi)}d^2-7d+3}{d^2}.
\end{aligned}
\end{equation}
\end{lemma}

\subsubsection{Proofs of \lsref{lem:giHaar}-\ref{lem:TrRiOrhoF}}

\begin{proof}[Proof of \lref{lem:giHaar}]
\Lref{lem:giHaar} follows from 	\eref{eq:RelationPandRi} and \pref{pro:OmegaHaar}. 
\end{proof}

\begin{proof}[Proof of \lref{lem:giCl}]
When $n\ge2$,  \eref{eq:giCl} follows from  \pref{pro:OmegaCl} and \eref{eq:RelationPandRi}. 
When  $n=1$, by virtue of \pref{pro:OmegaClSingle}
and  \eref{eq:RelationPandRi} we can deduce that
\begin{equation}
\Omega(\Cl_1) =\frac{1}{9}\caR_1 -\frac{1}{72}\caR_2 +\frac{1}{36}\caR_3 + \frac{1}{36} \caR_4 + \frac{1}{36} \caR_5 = \frac{1}{12} (\caR_1+\caR_4),
\end{equation}
where  the second equality holds because $2\caR_1 - \caR_2 + 2\caR_3 = 4\caR_4 - 2\caR_5$ by  \eref{eq:RelationPandRi} given that
 $P_{[2,2], G} = P_{[2,2], G}^+$ when $n=1$. This observation confirms \eref{eq:giCl} and completes the proof of \lref{lem:giCl}. 
\end{proof}

\Lsref{lem:giUkl} and \ref{lem:gitUk} can be proved in a similar way as shown above, and the details are omitted for simplicity.

\begin{proof}[Proof of \lref{lem:TrRiOrho}]
We first prove \eref{eq:TrRiOrho}. The equalities in \eref{eq:TrRiOrho} follow from  the definitions of $\caR_i$ in \eref{eq:DefcaRi} and the definition of $V_\triangle(O,\rho)$ in \eref{eq:VTriangleDef}.  The inequalities in the first three lines in \eref{eq:TrRiOrho} follow from the inequalities below,
\begin{equation}
\begin{gathered}
\tr(O\rho) \le \|O\|_\infty \|\rho\|_1 = \|O\|_\infty,\quad \tr(O^2\rho) \le \|O^2\|_\infty \|\rho\|_1 = \|O\|_\infty^2,\quad 
\tr(O^2\rho^2) \le \|O^2\|_\infty \|\rho^2\|_1 \le \|O\|_\infty^2, \\
|\tr(O\rho O\rho)| \le \tr(O^2\rho^2) \le  \|O\|_\infty^2,\quad \|O\|_\infty^2 \le \frac{d-1}{d} \|O\|_2^2.
\end{gathered}
\end{equation}
Here the last inequality holds because $O$ is Hermitian and traceless, while the other inequalities 
follow from the operator H\"{o}lder inequality (including the operator Cauchy-Schwarz inequality). The inequality in the fourth line in \eref{eq:TrRiOrho} follows from \pref{pro:VTriangle}. The inequality in the last line holds because 
\begin{equation}\label{eq:trRiOrhoUpper}
\left|\tr\left[R(\caT)\Tensor{(O\otimes \rho)}{2}\right] \right| \le \|O\|_2^2 \quad \forall \caT\in\Sigma_{4,4}
\end{equation}
according to Lemma~9 in the Appendix of \rcite{Helsen2023MultiShot}. 

Finally, by virtue of \eref{eq:TrRiOrho} and \lref{lem:CharOrhoNorm} it is easy to verify that $\tr\left[\caR_i \Tensor{(O\otimes \rho)}{2}\right]\ge0$ for  $i=1,2,3,4$, which completes the proof of \lref{lem:TrRiOrho}.
\end{proof}

\begin{proof}[Proof of \lref{lem:TrRiOrhoF}]
The first three equalities in \eref{eq:TrRiOrhoF} follow from \lref{lem:TrRiOrho}, and the fourth one follows from  \lsref{lem:CharOrhoNormF} and~\ref{lem:TrRiOrho}. When $\rho=|\phi\>\<\phi|$, we have $F=1$ and $\|\Xi_{\rho, \phi}\|_2^2 = \tXi_{\rho, \phi}\cdot\Xi_{\rho, \phi}=2^{-M_2(\phi)}d$ by \lref{lem:CharOrhoNormF}, so the first four equalities in \eref{eq:TrRiOphiF} are simple corollaries of  \eref{eq:TrRiOrhoF}. To prove  the last equality in  \eref{eq:TrRiOphiF}, we determine the contribution of each term in the summation $\caR_5 = \sum_{\caT\in\scrT_5}R(\caT)$ separately, where $\scrT_5=\{(13)\caT_4, (23)\caT_4, (123)\caT_4, (132)\caT_4\}$. For example,
\begin{align}
\tr\left[R\bigl((13)\caT_4\bigr) \Tensor{(O\otimes \rho)}{2}\right]&=\frac{1}{d}
\sum_{P\in\bcaP_n}\tr\left[\left(|\phi\>\<\phi|-\frac{\bbone}{d}\right)P\left(|\phi\>\<\phi|-\frac{\bbone}{d}\right)P\right]\<\phi|P|\phi\>^2\nonumber\\
&=\frac{1}{d}\sum_{P\in\bcaP_n}\left[\<\phi|P|\phi\>^2-\frac{1}{d}\right]\<\phi|P|\phi\>^2=\frac{2^{-M_2(\phi)}d-1}{d},
\end{align}
where we have used the definition of $M_2(\phi)$ and the fact that $\sum_{P\in\bcaP_n}\<\phi|P|\phi\>^2=d$. The other three terms can be determined in a similar way, with the results
\begin{align}
\tr\left[R\bigl((23)\caT_4\bigr) \Tensor{(O\otimes \rho)}{2}\right]=\tr\left[R\bigl((123)\caT_4\bigr) \Tensor{(O\otimes \rho)}{2}\right]=\tr\left[R\bigl((132)\caT_4\bigr) \Tensor{(O\otimes \rho)}{2}\right]=\frac{2^{-M_2(\phi)}d^2-2d+1}{d^2}.
\end{align}
The above two equations together imply the  last equality in  \eref{eq:TrRiOphiF} and complete the proof of \lref{lem:TrRiOrhoF}. 
\end{proof}

\section{Proof of \thref{thm:VarHaar}}\label{app:ProofHaar}
\begin{proof}
According to  \eref{eq:DefV*gi} and \lref{lem:giHaar} we have
\begin{align}
&V_*(O,\rho) = (d+1)^2\tr\left[\Omega(\haar) \Tensor{(O\otimes \rho)}{2}\right]-[\tr(O\rho)]^2\nonumber\\
&= \frac{2}{d(d+2)(d+3)}\tr\left[\caR_1\Tensor{(O\otimes \rho)}{2}\right]- \frac{d+1}{d(d+2)(d+3)}\tr\left[\caR_2\Tensor{(O\otimes \rho)}{2}\right] + \frac{d+1}{d(d+3)}\tr\left[\caR_3\Tensor{(O\otimes \rho)}{2}\right]. \label{eq:CalVarStabHaar}
\end{align}
In conjunction with \lref{lem:TrRiOrho} we can deduce that
\begin{equation}
V_*(O) = \max_\rho V_*(O,\rho) \le \left[\frac{2(d-1)}{d^2(d+2)(d+3)}+\frac{(4d-3)(d+1)}{d^2(d+3)}\right] \|O\|_2^2 = \frac{4 d^3+9 d^2+d-8}{d^2 (d+2) (d+3)} \|O\|_2^2 \le \frac{4}{d} \|O\|_2^2,
\end{equation}
which confirms \eref{eq:VarFHaarBound}.

When $O=|\phi\>\<\phi|-\bbone/d$, by virtue of \eref{eq:CalVarStabHaar} and \lref{lem:TrRiOrhoF} we can deduce that
\begin{align}
V_*(O,\phi)&= \frac{4(d-1)}{(d+2)(d+3)}<\frac{4}{d},\\
V_*(O,\rho) &= \frac{1-d+(d^2+d+2)F^2-2(3d+1)F+2(d^2-1)\<\phi|\rho^2|\phi\>+(d+1)^2\wp(\rho)}{d(d+2)(d+3)}\nonumber\\
& \le \frac{1-d+(d^2+d+2)F^2+2(d^2-3d-2)F+(d+1)^2\wp(\rho)}{d(d+2)(d+3)}\nonumber\\
&\le \frac{1-d+(d^2+d+2)+2(d^2-3d-2)+(d+1)^2}{d(d+2)(d+3)}=\frac{4(d-1)}{(d+2)(d+3)}= V_*(O,\phi),
\end{align}
which imply \eref{eq:VarFHaar}. Here the first inequality holds because $\<\phi|\rho^2|\phi\> \le \<\phi|\rho|\phi\>=F$, and the second inequality holds because $0\leq F,\wp(\rho)\le 1$.

Incidentally, when $d\geq 4$, 
the upper bound $V_*(O,\phi)$ for
$V_*(O,\rho)$  is saturated  iff $\rho= |\phi\>\<\phi|$; in the special case $d=2$, the upper bound is saturated iff $\rho= |\phi\>\<\phi|$ or $\rho=\bbone- |\phi\>\<\phi|$.
\end{proof}

\section{Proofs of \thsref{thm:VarCl},~\ref{thm:VarFCl} and \coref{cor:VarCl}}\label{app:ProofCl}
In this appendix, we prove our main results on thrifty shadow based on the Clifford group, including \thsref{thm:VarCl},~\ref{thm:VarFCl} and \coref{cor:VarCl}.

\begin{proof}[Proof of \thref{thm:VarCl}]
By virtue of  \eref{eq:DefV*gi} and \lref{lem:giCl} we can deduce that
\begin{align}
V_*(O,\rho) &= -\frac{1}{d+2}\tr\left[\caR_1 \Tensor{(O\otimes \rho)}{2}\right]+ \frac{d+1}{d+2}\tr\left[\caR_4 \Tensor{(O\otimes \rho)}{2}\right]\nonumber\\
&=-\frac{1}{d+2}[\tr\left(O\rho\right)]^2 + \frac{d+1}{d+2}\tr\left[\caR_4 \Tensor{(O\otimes \rho)}{2}\right]\nonumber\\
&=V_\triangle(O,\rho)-\frac{1}{d+2}[\tr(O\rho)]^2 \le V_\triangle(O,\rho) \le \frac{2(d+1)}{d(d+2)} \|\Xi_{\rho, O}\|_2^2 \le  \frac{2(d+1)}{d+2} \|O\|_2^2,
\end{align}
which confirms \eref{eq:VarClUpper}. Here the last two inequalities follow from \pref{pro:VTriangle}.
\end{proof}

\begin{proof}[Proofs of \coref{cor:VarCl} and \thref{thm:VarFCl}]
The two conclusions are simple corollaries of  \thref{thm:VarCl} and \pref{pro:VTriangle}.
\end{proof}

\section{Proofs of \thsref{thm:VarUkl},~\ref{thm:VarFUkl} and \pref{pro:VarUklFphi}}\label{app:ProofUkl}
\begin{proof}[Proof of \thref{thm:VarUkl}]
When $kl=0$, \eref{eq:VarUklUpper} follows from \thref{thm:VarCl}.

Next, suppose $kl\ge1$. Then the coefficients
$g_i(\bbU_{k, l})$ presented in \lref{lem:giUkl} satisfy the following inequalities,
\begin{gather}
\frac{2}{d(d+1)^2(d+2)(d+3)} - \frac{\gamma^{kl}}{d(d+1)(d+3)} <g_1(\bbU_{k,l})-\frac{1}{(d+1)^2} < \frac{2}{d(d+1)^2(d+2)(d+3)},\label{eq:g1}\\
-\frac{1}{d(d+1)(d+2)(d+3)}<g_2(\bbU_{k, l})  \le 0,\label{eq:g2}\\
0<\frac{1-\gamma^{kl}}{d(d+1)(d+3)}<g_3(\bbU_{k,l}) < \frac{1}{d(d+1)(d+3)},\\
0<\frac{\gamma^{kl}}{(d+1)(d+2)}-\frac{(3d+1)kl\gamma^{kl}}{3(d+1)(d^2-1)(d+2)}<g_4(\bbU_{k,l}) < \frac{\gamma^{kl}}{(d+1)(d+2)},\\
-\frac{dkl\gamma^{kl}}{3(d^2-1)(d+1)(d+2)}\le g_5(\bbU_{k,l})< 0\label{eq:g5}.
\end{gather}
All these  inequalities  are simple corollaries of \lref{lem:PropertyAlphaBeta}, except for the second inequality in \eref{eq:g2}, which can be proved as follows. When $kl=1$, it is straightforward to verify that  $g_2(\bbU_{k, l})=0$. When $n=1$, which means $k=1$, $\alpha_k=1/6$, and $\beta_k=2/3$, we have
\begin{align}
g_2(\bbU_{k, l})=\frac{1}{360}\left[5\times \left(\frac{2}{3}\right)^l-3- \frac{2}{6^l}\right]\leq 0.
\end{align}
When  $n\ge2$ and $kl\ge2$, by virtue of \lref{lem:PropertyAlphaBeta} we can deduce that
\begin{equation}
g_2(\bbU_{k, l}) = \frac{(d+3)\beta_k^l -d\alpha_k^l -3}{3d(d+1)(d+2)(d+3)} < \frac{\left(3+\frac{d^2kl}{d^2-1}\right)\gamma^{kl}-3}{3d(d+1)(d+2)(d+3)} < 0,
\end{equation}
given  that $\left[3+(d^2kl)/(d^2-1)\right]\gamma^{kl}-3$ decreases monotonically with $d$ and $kl$ when $d\ge4$ and $kl\ge2$. This observation completes the proof of the second inequality in \eref{eq:g2}.

By virtue of  Eqs.~\eqref{eq:DefV*gi}, \eqref{eq:g1}-\eqref{eq:g5}, and \lref{lem:TrRiOrho} we can deduce that $V_*(O,\rho) = \gamma^{kl}V_\triangle (O,\rho) + \caO\bigl(d^{-1}\bigr) \|O\|_2^2$, which confirms \eref{eq:VarUklAppro} and highlights the leading contribution to the variance $V_*(O,\rho)$. Let $\xi_i=\tr\bigl[\caR_i\Tensor{(O\otimes \rho)}{2}\bigr]$ for $i=1,2,3,4,5$; then  we can further drive a  precise upper bound for $V_*(O,\rho)$ as follows,
\begin{align}
V_*(O,\rho) &\le \frac{2\xi_1}{d(d+2)(d+3)} + \frac{(d+1)\xi_3}{d(d+3)} + \frac{(d+1)\gamma^{kl}\xi_4}{d+2} +\frac{dkl\gamma^{kl}|\xi_5|}{3(d-1)(d+2)}\nonumber\\
&\le \gamma^{kl} V_\triangle(O,\rho) + \left[\frac{2(d-1)}{d^2(d+2)(d+3)}+\frac{(4d-3)(d+1)}{d^2(d+3)}+ \frac{4dkl\gamma^{kl}}{3(d-1)(d+2)} \right] \|O\|_2^2\nonumber\\
&\le \gamma^{kl} V_\triangle(O,\rho) + \left[\frac{2(d-1)}{d^2(d+2)(d+3)}+\frac{(4d-3)(d+1)}{d^2(d+3)}+ \frac{27d}{16(d-1)(d+2)} \right]\|O\|_2^2\nonumber\\
&\le \gamma^{kl} V_\triangle(O,\rho) + \frac{6}{d}\|O\|_2^2.
\end{align}
Here the first two inequalities follow from  Eqs.~\eqref{eq:DefV*gi}, \eqref{eq:g1}-\eqref{eq:g5}, and \lref{lem:TrRiOrho} as before, the third inequality holds because $kl\gamma^{kl}\le 81/64$, and the last inequality  is easy to verify. Following a similar reasoning,  we can drive a  lower bound for $V_*(O,\rho)$,
\begin{align}
V_*(O,\rho) &\ge \left[\frac{2-(d+1)\gamma^{kl}}{d(d+3)}\right]\xi_1 -\frac{(d+1)\xi_2}{d(d+2)(d+3)} + \left[\frac{(d+1)\gamma^k}{d+2}-\frac{(3d+1)kl\gamma^{kl}}{3(d-1)(d+2)} \right]\xi_4 -\frac{dkl\gamma^{kl}|\xi_5|}{3(d-1)(d+2)}\nonumber\\
&\ge \gamma^{kl} V_\triangle(O,\rho)- \left[\frac{(d^2-1)\gamma^{kl}}{d^2(d+3)}+\frac{(d+1)(9d-8)}{d^2(d+2)(d+3)}+\frac{2(3d+1)kl\gamma^{kl}}{3(d-1)(d+2)}+\frac{4dkl\gamma^{kl}}{3(d-1)(d+2)}\right]\|O\|_2^2\nonumber\\
&\ge\gamma^{kl} V_\triangle(O,\rho)- \left[\frac{3(d^2-1)}{4d^2(d+3)}+\frac{(d+1)(9d-8)}{d^2(d+2)(d+3)}+\frac{27(3d+1)}{32(d-1)(d+2)}+\frac{27d}{16(d-1)(d+2)}\right]\|O\|_2^2\nonumber\\
&\ge \gamma^{kl} V_\triangle(O,\rho) -\frac{6}{d}\|O\|_2^2. 
\end{align}
The above two equations together imply \eref{eq:VarUklUpper} and complete the proof of \thref{thm:VarUkl}.
\end{proof}

\begin{proof}[Proof of \pref{pro:VarUklFphi}]
	\Eref{eq:VarUklFphi} follows from \eref{eq:DefV*gi} and  \lsref{lem:giUkl}, \ref{lem:TrRiOrhoF}. When $l=0$, \eref{eq:VarUklFphi} also  follows from \thref{thm:VarFCl}. 
\end{proof}

\begin{proof}[Proof of \thref{thm:VarFUkl}]
\Eref{eq:VarFUklBound} follows from \thref{thm:VarUkl} and \pref{pro:VTriangle}, given that $\|O\|_2^2\leq 1$ by assumption. When 
$kl=0$, we have $\bbU_{k,l}=\Cl_n$, so  \eqsref{eq:VarFUklBound}{eq:VarFUklIdeal} follow from \thref{thm:VarFCl}.

It remains to prove \eref{eq:VarFUklIdeal} for the case  $kl\geq1$. By virtue of the formula for  $V_*(O,\phi)$ in \pref{pro:VarUklFphi}, we can derive the following upper bound for $V_*(O,\phi)$,
\begin{align}
V_*(O,\phi) &= \frac{2^{1-M_2(\phi)}\alpha_k^l(d+1)}{(d+2)} + \frac{4(d-1)-8(d+1)\alpha_k^l}{(d+2)(d+3)}< 2^{1-M_2(\phi)}\alpha_k^l + \frac{4}{d} < 2^{1-M_2(\phi)}\gamma^{kl} + \frac{4}{d}, 
\end{align}
where the inequalities  hold because $0<\alpha_k< \gamma^k$ by \lref{lem:PropertyAlphaBeta}. 
In addition, we can derive a lower bound  as follows,
\begin{align}
V_*(O, \phi) &= 2^{1-M_2(\phi)}\alpha_k^l + \frac{4(d-1)-\left[ 2^{1-M_2(\phi)}(d+3)+8(d+1)\right]\alpha_k^l}{(d+2)(d+3)}> 2^{1-M_2(\phi)}\gamma^{kl}-f,  \label{eq:VarUklLower}
\end{align}
where
\begin{align}
f =\frac{2(3d+1)kl\gamma^{kl}}{3(d^2-1)}+ \frac{(10d+14)\gamma^{kl}-4(d-1)}{(d+2)(d+3)},
\end{align}
and the inequality follows from  \lref{lem:PropertyAlphaBeta} and the fact that $M_2(\phi)\ge 0$. If  $kl=1$, then 
\begin{align}
f=\frac{10d^3+45d^2+16d-23}{2(d+2)(d+3)(d^2-1)}< \frac{5}{d}. 
\end{align}
If $kl\geq 2$, then $\gamma^{kl}\leq 9/16$, $kl\gamma^{kl} \le 81/64$, and 
\begin{align}
f\leq \frac{27(3d+1)}{32(d^2-1)}+\frac{13d+95}{8(d+2)(d+3)}< \frac{6}{d}. 
\end{align}
The above  equations together imply that $-6/d<V_*(O,\phi)-2^{1-M_2(\phi)}\gamma^{kl}<4/d$, which confirms  \eref{eq:VarFUklIdeal} and completes the proof of \thref{thm:VarFUkl}.
\end{proof}

\section{Proofs of \thref{thm:VartUk} and \pref{pro:VarFtUk}}\label{app:ProoftUk}
\begin{proof}[Proof of \thref{thm:VartUk}]
When $k=0$, \eref{eq:VartUkBound} follows from \thref{thm:VarCl} (see also \thref{thm:VarUkl}).

Next, let $\xi_i=\tr\bigl[\caR_i\Tensor{(O\otimes \rho)}{2}\bigr]$ for $i=1,2,3,4,5$ as in the proof of \thref{thm:VarUkl}; note that 
$\xi_1=[\tr(O\rho)]^2$ and $\xi_4=(d+2)V_\triangle(O,\rho)/(d+1)$ by \lref{lem:TrRiOrho}. 
When $k=1$, by virtue of Eqs.~\eqref{eq:DefV*gi} and \lref{lem:gitUk}  we can deduce that
\begin{align}
V_*(O,\rho)&=\frac{-(3d-5)\xi_1+(d+1)\xi_3+(d+1)(3d-5)\xi_4+(d+1)\xi_5}{4(d-1)(d+2)},\\
V_*(O,\rho)-\gamma^kV_\triangle(O,\rho) &=\frac{-(3d-5)\xi_1+(d+1)\xi_3-2(d+1)\xi_4+(d+1)\xi_5 }{4(d-1)(d+2)},
\end{align}
given that $\gamma=3/4$. In conjunction with  \lref{lem:TrRiOrho} we can deduce that
\begin{align}
V_*(O,\rho)-\gamma^kV_\triangle(O,\rho) &\leq \frac{2(d+1) }{(d-1)(d+2)}\|O\|_2^2\leq \frac{3}{d}\|O\|_2^2,\\
V_*(O,\rho)-\gamma^kV_\triangle(O,\rho)&\geq -\frac{11d+3 }{4(d-1)(d+2)}\|O\|_2^2\geq -\frac{25}{8d}\|O\|_2^2,
\end{align}
which imply \eref{eq:VartUkBound}.

Next, we assume that
$k\ge2$, which means $n\geq 2$ and $d\geq 4$. By virtue of  \lref{lem:gitUk} we can deduce that
\begin{gather}
- \frac{3d-5}{4(d-1)(d+1)^2(d+2)} \le g_1(\tbbU_k) - \frac{1}{(d+1)^2} = \frac{\left(-d^3-3 d^2-2 d\right) \gamma ^k + (4 d+4) \nu ^k + 2 d^2+8 d-12}{ (d-1) (d+1)^2 (d^2-4) (d+4)} < 0,\label{eq:g1tUk}\\
-\frac{d}{(d^2-4)(d^2-1)(d+4)}< g_2(\tbbU_k) = \frac{d \left(2 \gamma ^k-\nu ^k-1\right)}{(d^2-4) (d^2-1) (d+4)} \le 0,\\
0\le g_3(\tbbU_k) =\frac{d^2+2d-4 - (d^2+2d) \gamma ^k+4 \nu ^k}{(d^2-4)(d^2-1) (d+4)} < \frac{d^2+2 d-4}{(d^2-4)(d^2-1) (d+4)},\\
-\frac{2d^2- d-12}{(d^2-4)(d^2-1) (d+4)} < g_4(\tbbU_k)- \frac{\gamma^k}{(d+1)(d+2)}= \frac{2 \left(d^2+4 d-4\right) \gamma ^k - 2 (d+2) \nu ^k-2 \left(d^2+3 d-6\right)}{(d^2-4)(d^2-1) (d+4)}  < 0, \label{eq:g4tUk}\\
-\frac{1}{4(d^2-1)(d+2)}\le g_5(\tbbU_k) =\frac{-(d^2+2d) \gamma ^k+(d^2+2d-4) \nu ^k+4}{(d^2-4) (d^2-1) (d+4)} \le 0\label{eq:g5tUk}.
\end{gather}
Here the inequalities follow from the facts  $\gamma=3/4$, $\nu=1/2$, $\gamma^k/\nu^k=(3/2)^k\geq 9/4$,  $d\gamma^k\geq d\gamma^n=(3/2)^n\geq  9/4$ and direct calculation in a few special cases. In conjunction with  \eref{eq:DefV*gi} and \lref{lem:TrRiOrho} (cf. the proof of \thref{thm:VarUkl}) we can derive the following results,
\begin{align}
 &V_*(O, \rho) -\gamma^k V_\triangle(O,\rho) \le \frac{(d+1)(d^2+2d-4)}{(d^2-4)(d-1)(d+4)}\xi_3 + \frac{d+1}{4(d-1)(d+2)}\left|\xi_5\right|\nonumber\\
&\quad \le \left[\frac{(4d-3)(d+1)(d^2+2d-4)}{d(d^2-4)(d-1)(d+4)} + \frac{d+1}{(d-1)(d+2)}\right]\|O\|_2^2\le \frac{6}{d}\|O\|_2^2,\\
& \gamma^k V_\triangle(O,\rho) -V_*(O, \rho)  \le \biggl[\frac{(3d-5)\xi_1+(d+1)|\xi_5| }{4(d-1)(d+2)} + \frac{d(d+1)\xi_2+(d+1)(2d^2- d-12)\xi_4}{(d^2-4)(d-1)(d+4)} \biggr]\|O\|_2^2\nonumber\\
&\quad \leq \biggl[\frac{(3d-5)}{4d(d+2)}+\frac{d+1}{(d-1)(d+2)}+\frac{(d+1)(9d-8)}{(d^2-4)(d-1)(d+4)}+\frac{2(d+1) (2d^2- d-12)}{(d^2-4)(d-1) (d+4)}\biggr]\|O\|_2^2\le\frac{6}{d}\|O\|_2^2,
\end{align}
which imply \eref{eq:VartUkBound} and complete the proof of \thref{thm:VartUk}.
\end{proof}

\begin{proof}[Proof of \pref{pro:VarFtUk}]
\Eref{eq:VarFtUkBound} follows  from \thref{thm:VartUk} and \pref{pro:VTriangle} given that $\|O\|_2^2<1$ by assumption.
When $k=0$, \eqsref{eq:VarFtUkExact}{eq:VarFtUkIdealBound} follow from \thref{thm:VarFCl} (see also \thref{thm:VarFUkl}). It remains to prove \eqsref{eq:VarFtUkExact}{eq:VarFtUkIdealBound} for the case $k\geq 1$.

By virtue of \eref{eq:DefV*gi} and \lsref{lem:gitUk}, \ref{lem:TrRiOrhoF} we can deduce \eref{eq:VarFtUkExact} and the following results,
\begin{align}
V_*(O,\phi) - 2^{1-M_2(\phi)}\gamma^k &= \frac{2^{1-M_2(\phi)}\left[-\left(d^2-d-8\right)  \gamma ^k +\left(2 d^2+8 d+6\right)\nu^k-2 d^2-12 d-10\right] }{(d-1) (d+2) (d+4)}\nonumber\\
&\quad-\frac{4 \left[(2 d^2+2d) \gamma ^k+(4 d+4) \nu ^k- d^2-3 d-8\right]}{(d-1) (d+2) (d+4)}\nonumber\\
&< \frac{-8 \left(d^2+d\right) \gamma ^k-16 (d+1) \nu ^k + 4 \left(d^2+3 d+8\right)}{(d-1) (d+2) (d+4)}\nonumber\\
&<\frac{4d^2+4d+8}{(d-1)(d+2)(d+4)}<\frac{4}{d},\label{eq:VarFtUkUpper}\\
V_*(O,\phi) - 2^{1-M_2(\phi)}\gamma^k &\ge \frac{-2 (5 d+8) \gamma ^k + 4 (d+1) \nu ^k-12}{(d+2) (d+4)}\ge -\frac{11}{2 d+4} > -\frac{6}{d}.\label{eq:VarFtUkLower}
\end{align}
Here the first inequalities in \eqsref{eq:VarFtUkUpper}{eq:VarFtUkLower} hold because $V_*(O,\phi) - 2^{1-M_2(\phi)}\gamma^k$ is monotonically increasing in  $M_2(\phi)$. The second inequality in \eref{eq:VarFtUkUpper} holds because  $d\gamma^k>1$ and $(d+1)\nu^k>1$, while  the second inequality in \eref{eq:VarFtUkLower} holds because $k\geq 1$ by assumption and  $-2 (5 d+8) \gamma ^k + 4 (d+1) \nu ^k-12$ is monotonically increasing in $k$. This observation confirms  \eref{eq:VarFtUkIdealBound} and completes the proof of \pref{pro:VarFtUk}.
\end{proof}

\end{document}